\documentclass[10pt,journal,compsoc]{IEEEtran}

\ifCLASSOPTIONcompsoc
\usepackage[nocompress]{cite}
\else
\usepackage{cite}
\fi

\usepackage{hyperref}
\hypersetup{hidelinks}

\usepackage{subfigure}
\usepackage{amsmath}
\usepackage{amsthm}
\usepackage{graphicx} 
\usepackage{caption}
\usepackage{graphicx,subfigure}
\usepackage{multicol} 
\usepackage{graphicx} 
\usepackage{amssymb}
\usepackage{float} 
\usepackage{wrapfig} 
\usepackage{booktabs}
\usepackage{lipsum} 
\usepackage{indentfirst}
\usepackage{bm}
\usepackage{color}
\usepackage{algorithm}
\usepackage{algorithmic}
\usepackage{amsmath}
\usepackage{enumerate}
\usepackage{extarrows}
\usepackage{url}
\usepackage{graphicx}
\usepackage{epstopdf}
\usepackage{multirow}
\usepackage{tablefootnote}
\usepackage[flushleft]{threeparttable}
\usepackage{slashbox}
\usepackage{array}

\usepackage{enumitem}
\setlist{leftmargin=5mm}

\newtheorem{assumption}{Assumption}

\newtheorem{question}{\textbf{Question}}

\newtheorem{lemma}{\textbf{Lemma}}

\newtheorem{definition}{\textbf{Definition}}
\newtheorem{theorem}{\textbf{Theorem}}

\newtheorem{game}{Game}

\def\eq{\triangleq}

\def\val{\theta}
\def\valmax{\val_\text{max}}

\def\cut{\beta}
\def\dmean{\bar{d}}
\def\dmax{D}

\def\payoff{U}
\def\payoffexp{\bar{U}}

\def\revenue{R}
\def\revenueexp{\bar{R}}
\def\cost{C}
\def\profit{W}
\def\spedata{\tau}

\def\plan{\mathcal{T}}

\def\dcap{Q}
\def\pcap{\Pi}
\def\adfee{\pi}
\def\mechanism{\kappa}

\def\QoS{\rho}
\def\load{L}
\def\subscription{\Phi}

\def\pricing{\boldsymbol{s}}
\def\valthr{\sigma}
\def\Threshold{\boldsymbol{\valthr}}
\def\Mechanism{\boldsymbol{\mechanism}}

\def\costqos{\psi}

\def\usage{V}

\def\valthr{\sigma}
\def\valeq{\tilde{\valthr}}

\def\trad{\mathtt{T}}
\def\roll{\mathtt{R}}
\def\na{\mathtt{Na}}

\begin{document}

\title{Duopoly Competition for Mobile Data Plans with Time Flexibility}

\author{Zhiyuan~Wang,
	Lin~Gao,~\IEEEmembership{Member,~IEEE,}
	and~Jianwei~Huang,~\IEEEmembership{Fellow,~IEEE}
	
	\IEEEcompsocitemizethanks{
		\IEEEcompsocthanksitem Part of the results appeared in WiOpt 2018 \cite{Zhiyuan2018duopoly}.
		\IEEEcompsocthanksitem Zhiyuan Wang and Jianwei Huang are with the Network Communications
		and Economics Lab, Department of Information Engineering, The Chinese University of Hong Kong, Shatin, N.T., Hong Kong, China.
		 Jianwei Huang is also with the School of Science and Engineering, The Chinese University of Hong Kong, Shenzhen. 
		E-mail:	$\{$wz016, jwhuang$\}$@ie.cuhk.edu.hk
	\IEEEcompsocthanksitem Lin Gao is with the School of Electronic and Information Engineering, Harbin Institute of Technology, Shenzhen, China. 
		E-mail: gaol@hit.edu.cn
}
	}

%
%


%

\IEEEtitleabstractindextext{%
	\begin{abstract}
		The growing competition drives the mobile network operators (MNOs) to explore adding time flexibility to the traditional data plan, which consists  of a monthly subscription fee, a data cap, and a per-unit fee for exceeding the data cap.
		The rollover  data plan, which allows the unused data of the previous month  to be used in the current  month, provides the  subscribers with the time flexibility.
		In this paper, we formulate two MNOs' market competition as a three-stage game, where the MNOs decide their data mechanisms (traditional or rollover)  in Stage I and the pricing strategies in Stage II, and then users make their subscription decisions in Stage III.
		Different from the monopoly market where an MNO always prefers  the rollover mechanism  over the traditional plan in terms of profit,  MNOs may adopt different data mechanisms at an equilibrium.
		Specifically, the high-QoS MNO would gradually abandon the rollover mechanism as its QoS advantage diminishes.
		Meanwhile, the low-QoS MNO would progressively upgrade to the rollover mechanism.
		The numerical results show that the market competition significantly limits MNOs' profits, but both MNOs obtain higher profits with the possible choice of the rollover data plan.  
	\end{abstract}
	
	}

\maketitle

\IEEEdisplaynontitleabstractindextext

\IEEEpeerreviewmaketitle

\section{Introduction\label{Section: Introduction}} 
\subsection{Background and Motivation}\label{Subsection: Background and Motivation}  
The Mobile Network Operators (MNOs) profit from the mobile data services through offering carefully designed  mobile data plans.
The {traditional} and most widely implemented data plan is a three-part tariff, involving a monthly one-time subscription fee, a data cap that is free to use (with the paid subscription fee), and a per-unit fee for any data consumption exceeding the data cap.
In today's telecommunication market, the most commonly adopted data caps include 1GB, 2GB, and 3GB \cite{sen2013survey}.
However, the corresponding subscription fee for the same data cap varies significantly in different MNOs' data plans.
For example, the subscription fee of the 2GB data plan is \$55 for AT\&T subscribers \cite{ATT2GB}, while it is \$35 for Verizon subscribers \cite{verizon2GB}.
The different pricing decisions mainly result from the MNOs' \textbf{market competition}, since different MNOs usually offer different quality of services (QoS) and experience different costs  for wireless data services \cite{sen2015smart}.

To maintain the competitiveness in the market competition, many MNOs (e.g., AT\&T in the US \cite{ATTrollover}, and China Mobile in mainland China \cite{CMrollover}) have recently adopted  the \textbf{rollover data plans}, allowing the unused data from the previous month to be used in the current month.
Such a rollover mechanism is more time-flexible than the traditional mechanism, since it reduces users' concerns of  the  possible \textit{wasting data} within the data cap and the possible \textit{overage data consumption} above the data cap when the user cannot accurately estimate his future data demand.
Hence, the rollover data plan is very attractive to the mobile users.

Our earlier study in \cite{Zhiyuan2018TMC,Zhiyuan2018MobiHoc} found that in a \textbf{monopoly market} with a single MNO, the rollover mechanism can increase both the MNO's profit and users' payoffs, hence improves the social welfare.
That is, a monopoly MNO should definitely adopt the rollover mechanism.
In this paper, we want to understand whether this is still true when considering  the ubiquitous market competition.
In practice, the four main MNOs in the US market all adopt the rollover mechanism.
In the Europe and Hong Kong, however, we do not observe all of the MNOs adopting the rollover mechanism.
For example, some MNOs (e.g., Orange and China Mobile Hong Kong) are still using the traditional mechanism without time flexibility.
These observations motivate us to ask the following two key questions in a {competitive market}: 
\begin{question}\label{Question: 1}
	Will all MNOs offering rollover mechanism become the equilibrium of the market competition?
\end{question}
\begin{question}\label{Question: 2}
	How will the different data mechanisms change the MNOs' pricing competition?
\end{question}

To address the above questions, this paper studies the MNOs' market competition in terms of their rollover data mechanism offering and the pricing strategy. 
To abstract the interactions among the competitive MNOs and the heterogeneous users, we focus on the two-MNO case (i.e., duopoly market) in this paper,  and we will study the multiple MNO case (i.e., oligopoly market) in our future work  (with some preliminary discussions in Appendix A).
We hope that our results in this paper could help understand how the competitive MNOs choose their data mechanisms and  make the pricing decisions.

\subsection{Solutions and Contributions} 
We study the economic interactions between two competitive MNOs and a group of heterogeneous mobile users.
We use a three-stage game model to characterize the MNOs' rollover mechanism adoption and pricing decisions, as well as  users' subscriptions.
To be more specific, the two MNOs simultaneously decide their data mechanisms (traditional or rollover) in Stage I, and the corresponding pricing strategies (including subscription fee and the per-unit fee) for the same data caps in Stage II.
Finally, users make their subscription decisions to maximize their payoffs in Stage III.

The main results and key contributions of this paper are summarized  as follows: 
\begin{itemize} 

	\item \textit{Duopoly Competition for Mobile Data Plans with Time Flexibility:} 
	To the best of our knowledge, this is the first paper that systematically studies the MNOs' duopoly competition considering their rollover data mechanism offering and pricing decisions.
	\item \textit{A Three-Stage Competition Model:}
	We formulate the MNOs' market competition and users' subscription as a three-stage game.
	Despite the complexity of the game, we characterize the equilibrium considering MNOs' heterogeneity in the Quality-of-Service (QoS) and the operational cost, as well as  users' heterogeneity in their data valuations. 
	\item \textit{Data Mechanism Equilibrium:}
	Our analysis shows that the high-QoS MNO would gradually abandon the rollover mechanism as its QoS advantage diminishes (due to its increasing cost or the competitor's decreasing cost).
	In this progress, however, the low-QoS MNO has the opportunity to upgrade to the rollover mechanism.
	Particularly, the market competition shares some similarity with that of  the anti-coordination game when the MNOs have  similar QoS and experience comparable cost. 
	That is, no matter who adopts the rollover mechanism, the other will choose to  adopt the traditional one.
	\item \textit{Evaluation based on Empirical Data:}
	The numerical results based on the empirical data show that the market competition significantly reduces MNO's profits.
	Furthermore, both MNOs can obtain higher profits when they have the choice of adopting the rollover mechanism in the duopoly market.
\end{itemize}

The remainder of this paper is organized as follows.
In Section \ref{Section: Related Literature}, we review the related works.
Section \ref{Section: System Model} introduces the system model. 
Section \ref{Section: User Subscription in Stage III} studies users' subscriptions.
Section \ref{Section: MNOs' Pricing Competition in Stage II} investigates MNOs' pricing competition.
Section \ref{Section: MNO's Data Mechanism Selection in Stage I} analyzes MNOs' data mechanism adoptions.
Section \ref{Section: Numerical Results} provides numerical results.
Finally, we conclude this paper in Section \ref{Section: Conclusions and Future Work}.

\section{Related Literature\label{Section: Related Literature}}  
There have been  many excellent studies on mobile data pricing (e.g., \cite{zhengoptimizing,ma2016time,xiong2017economic,wang2016user}). 
However, these prior studies  did not take into account the recently introduced rollover mechanism or the ubiquitous market competition in practice.

The  rollover data mechanism has been recently studied (e.g., \cite{Zhiyuan2018TMC,Zhiyuan2018MobiHoc,zheng2016understanding,wei2018novel,Zhiyuan2019Infocom}). 
Zheng \emph{et al.} in \cite{zheng2016understanding} examined such an innovative data mechanism and found that moderately price-sensitive users  can benefit from subscribing to the rollover data plan. 
Wei \emph{et al.} in \cite{wei2018novel} analyzed the impact of different rollover period lengths from the MNO's perspective. 
In our previous work \cite{Zhiyuan2018TMC,Zhiyuan2018MobiHoc}, we proposed a unified framework for different rollover mechanisms and studied the corresponding optimal design under the single-cap and multi-cap schemes.
Moreover, we examined the economic viability of the rollover mechanism and the data trading market \cite{Zhiyuan2019Infocom}.
However, all of these studies only considered the monopoly case and neglected the ubiquitous market competition in practice.

There are many studies related to multiple MNOs' market competitions in terms of their pricing decisions (e.g., \cite{gibbens2000internet,chau2014economic,ma2016usage,ren2012data}).
For example, Gibbens \emph{et al.} in \cite{gibbens2000internet} focused on the Paris Metro pricing scheme and analyzed the competition between two ISPs who offer multiple service classes.
Later on Chau \emph{et al.} in \cite{chau2014economic} further considered a more general competition model and derived the necessary conditions for the equilibrium.
Ma \emph{et al.} in \cite{ma2016usage} focused on the usage-based scheme and considered the congestion-prone scenario with multiple MNOs.
Ren \textit{et al.} in \cite{ren2012data} focused on users' aggregate data demand dynamics, and optimized the MNO's data plans and long-term network capacity decisions.
However, none of these studies considered the MNOs' different rollover data mechanisms offering.

\begin{table}
	\setlength{\abovecaptionskip}{1pt}
	\setlength{\belowcaptionskip}{0pt}
	\renewcommand{\arraystretch}{1.1}		
	\caption{Comparing Mobile Data Pricing Literatures.} 
	\label{table: Various Plans}
	\centering
	\begin{tabular}{c c c c c c}
		\toprule
		\textbf{Literature} 			& \textbf{Rollover Mechanism}	&  \textbf{Market Competition}	\\
		\midrule
		{[12]-[15]}						&  $\times$						& $\times$						\\
		{[8][9][14]-[16]}				&  \checkmark					& $\times$						\\
		{[17]-[20]}						& $\times$						& \checkmark					\\
		{This Paper}					& \checkmark					& \checkmark					\\
		\bottomrule
	\end{tabular}  
\end{table}

In this paper, we study MNOs' market competition in terms of the rollover mechanism adoption and the pricing decisions.

\section{System Model\label{Section: System Model}} 
We study the market competition between two MNOs who face a common pool of mobile users.
We formulate the system interactions as a three-stage game and characterize how the MNOs' heterogeneity in the Quality-of-Service (QoS) and the operational cost, as well as  users' heterogeneity in the data valuations, affect the various decisions.

Each MNO-$n$ ($n=1,2$) offers a mobile data plan specified by a tuple $\plan_n=\{\dcap_n,\pcap_n,\adfee_n,\mechanism_n\}$:
a subscriber of MNO-$n$ needs to pay a monthly subscription fee $\pcap_n$ for the data cap $\dcap_n$, and possibly some  usage-based overage fee $\adfee_n$ for each unit of data consumption exceeding the data cap $\dcap_n$.
Here $\mechanism_n\in\{\trad,\roll\}$ denotes the data mechanism that the MNO-$n$ adopts.
Specifically, $\mechanism_n=\trad$ represents the traditional mechanism, while $\mechanism_n=\roll$ represents the rollover mechanism.

For $\mechanism_n=\roll$, the rollover data ``inherited'' from  the previous month is consumed {prior to} the current monthly data cap and expires at the end of the current month.\footnote{In practice, there are two different consumption priorities, which has been studied in our previous work. We refer interested readers to \cite{Zhiyuan2017pricing,Zhiyuan2018TMC} for more details.} 
Basically, the rollover data enlarges a subscriber's \textit{effective data cap} within which no overage fee involved.
Based on our previous study of the monopoly market  in \cite{Zhiyuan2017pricing,Zhiyuan2018TMC}, here we denote $\spedata$ as a user's rollover data at the beginning of a month, and $\dcap_{\mechanism_n}^e(\spedata)$ as the \textit{effective cap} of the current month  under MNO-$n$'s data plan $\plan_n=\{\dcap_n,\pcap_n,\adfee_n,\mechanism_n\}$.
\begin{itemize}
	\item The case of $\mechanism_n=\trad$ denotes the traditional data mechanism without rollover data, i.e., $\tau=0$. 
	The effective cap of each month is $\dcap_{\trad}^e(\tau)=\dcap_n$; 
	\item The case of $\mechanism_n=\roll$ denotes the rollover data mechanism. 
	The rollover data $\tau\in[0,\dcap_n]$ from the previous month enlarges the effective cap of the current month, i.e., $\dcap_{\roll}^e(\tau)=\dcap_n+\tau$.
\end{itemize}

We study the two MNOs' market competition in terms of their rollover mechanism adoption and pricing  decisions, given the same data caps (e.g., $\dcap_1=\dcap_2=1$GB).\footnote{In practice, MNOs usually offer multiple data caps, e.g., 1GB, 2GB, or 3GB. We have studied the multi-cap optimization problem for a monopoly MNO in \cite{Zhiyuan2018MobiHoc}. Here we assume that the two MNOs have the same data cap, and focus on the impact of the pricing and the choice of data mechanism.} 
As illustrated in Fig. \ref{fig: SystemModel}, in Stage I, two MNOs simultaneously announce their data mechanisms $\mechanism_1$ and $\mechanism_2$. 
In Stage II, two MNOs simultaneously determine the corresponding prices $s_1=\{\pcap_1,\adfee_1\}$ and $s_2=\{\pcap_2,\adfee_2\}$.\footnote{The competition model is motivated by practical observations: an MNO usually fixes a data mechanism over a relatively long time (e.g., three or five years), and updates the price choices more frequently (e.g., on a yearly basis). This formulation captures the MNO's different decisions at different time scales. Moreover, to reveal the impact of the rollover mechanism on MNOs' market competition, we assume that MNOs make simultaneous decisions in each stage. We will consider the sequential decision process (as in \cite{duan2015economic}) in our future work.}
Finally, users will make their subscription choices in Stage III.
\begin{figure}
	\centering
		\setlength{\abovecaptionskip}{1pt}
	\includegraphics[width=0.85\linewidth]{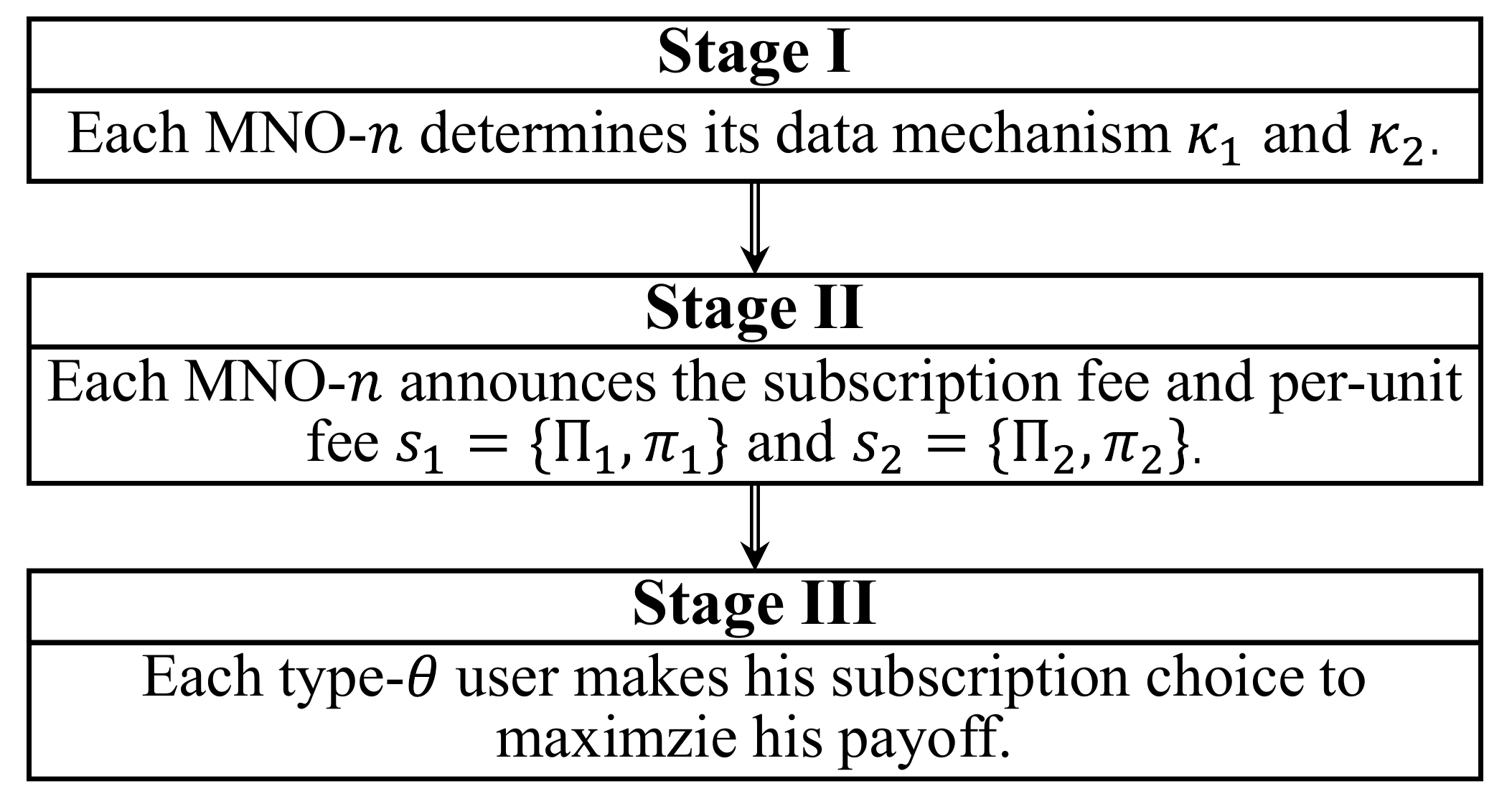}
	\caption{Three-stage competition model. }
	\label{fig: SystemModel}
\end{figure}

Next we introduce users' payoffs and the MNOs' profits in Section \ref{Subsection: Users' Payoffs} and Section \ref{Subsection: MNOs' Profits}, respectively.
Table \ref{Table: Key notations} summarizes some key notations in this paper.

\subsection{Users' Payoffs\label{Subsection: Users' Payoffs}} 
\subsubsection{User Characteristics}
Now we introduce three characterizations of a user: the data demand $d$, the data valuation $\val$, and the network substitutability $\cut$. 
Based on these, we will derive a user's monthly expected payoff. 


First, we model a user's data demand $d$ as a discrete random variable with a probability mass function $f(d)$, a mean value of $\dmean$, and a finite integer support $\{0,1,2,...,\dmax\}$ \cite{ren2012data}.
Here the data demand is measured in the minimum data unit (e.g., $1$KB or $1$MB according to the MNOs).

Second, we denote $\val$ as a user's utility from consuming one unit of data, i.e., the user's data valuation as in \cite{wang2017role,ma2016usage}.
According to the empirical results from \cite{Zhiyuan2018TMC}, in the telecom market of mainland China, the users' data valuations follow a gamma distribution, and falls into the range between $10$ RMB/GB and $60$ RMB/GB with a large probability.

Third, we consider  a user's behavior change after he exceeds the \textit{effective cap}. 
Although the user will still continue consuming data, he will reduce his consumption by relying more heavily on  alternative networks (such as office or home Wi-Fi networks).
Following \cite{sen2012economics}, we characterize such a behavior by a user's network substitutability $\cut\in[0,1]$, which denotes the fraction of overage usage shrink. 
For example, $\cut=0.6$ means that on average, $60\%$ of the user's portion of data demand above the effective cap will be reduced. 
A larger $\cut$ corresponds to more overage usage reduction (hence, a better network substitutability).
The empirical results in \cite{Zhiyuan2018TMC} show that most people would shrink $85\%\sim95\%$ overage usage. 
That is, users do not differ significantly in terms of their network substitutability.

\begin{table}
	\setlength{\abovecaptionskip}{1pt}
	\setlength{\belowcaptionskip}{0pt}
	\renewcommand{\arraystretch}{1.1}		
	\caption{Key Notations.} 
	\label{Table: Key notations}
	\centering
	\begin{tabular}{c|c|l}
		\toprule
		\multicolumn{2}{c}{\textbf{Symbols}} 		& $\qquad\qquad\qquad$\textbf{Physical Meaning}					\\
		\midrule
		\multirow{10}{*}{MNO}	& $\dcap_n$			& The data cap of MNO-$n$.	\\
		& $\pcap_n$				& The subscription fee of MNO-$n$.	\\
		& $\adfee_n$			& The overage usage fee of MNO-$n$.	\\
		& $\mechanism_n$		& The data mechanism of MNO-$n$.	\\
		& $s_n$					& The pricing decision ($s_n=\{\pcap_n,\adfee_n\}$) of MNO-$n$.		\\
		\cline{2-3}
		& $\QoS_n$				& The quality of service (QoS) of MNO-$n$.	\\
		& $c_n$					& The marginal operational cost of MNO-$n$.	\\
		& $\costqos_n$			& The cost-quality ratio of MNO-$n$, i.e., $\costqos_n=c_n/\QoS_n$.\\
		& $\valthr_n$	 		& The threshold user type of MNO-$n$.	\\
		& $\profit_n$			& The expected profit of MNO-$n$.	\\
		& $\valeq$				& The neutral user type of the two MNOs.	\\
		\hline
		\multirow{4}{*}{User}	& $\val$			& A user's  valuation for consuming one unit data.	\\
		& $\cut$			& Users' common network substitutability.	\\
		& $\usage_{\mechanism_n}$		& The expected usage of an MNO-$n$'s subscriber.\\
		& $\payoffexp_n$ 	& The expected payoff of an MNO-$n$'s subscriber.	\\								
		\bottomrule
	\end{tabular}  
\end{table}

In the following, we normalize the total population size to be one and  follow \cite{luo2016integrated},\cite{luo2015mine} by exploring users' heterogeneity in the data valuation $\val$.
Hence we  characterize each user according to his type $\val$.
The distribution of $\val$ of the entire user population has a PDF $h(\val)$ and CDF $H(\val)$ with the support $[0,\valmax]$.

\subsubsection{User Payoff under Different Data Mechanisms}
Next we introduce users' payoffs based on their characteristics and the effect of different data mechanisms.

A user's payoff is the difference between his utility and the  total payment.
More specifically, for an MNO-$n$'s subscriber with $d$ units of data demand and an effective data cap $\dcap_{\mechanism_n}^e(\spedata)$, his actual data consumption is $d-\cut[d-\dcap_{\mechanism_n}^e(\spedata)]^+$, where $[x]^+=\max\{0,x\}$.
Moreover, we use $\QoS_n$ to represent the MNO-$n$'s \textit{average} quality of service (QoS).\footnote{An MNO's wireless data service depends on the 	network congestion, which has been studied before (e.g., \cite{ma2016usage,sen2015smart}). In this work, instead of modeling the detailed congestion-aware control, we are more interested in the long-term average quality of the MNO's wireless data service. Hence the parameter $\QoS$ represents the impact of the MNO's average QoS on users' utilities of consuming data. }
Mathematically, the parameter $\QoS_n$ is a utility multiplicative coefficient, thus the subscriber's utility is $\QoS_n\val(d-\cut[d-\dcap_{\mechanism_n}^e(\spedata)]^+)$.
In addition, the subscriber's total payment consists of the subscription fee $\pcap_n$ and the overage charge $\adfee_n(1-\cut)[d-\dcap_{\mechanism_n}^e(\spedata)]^+$.
Therefore, \textit{a type-$\val$ MNO-$n$ subscriber's payoff} with a data demand $d$ and an effective cap $\dcap_{\mechanism_n}^e(\tau)$ is 
\begin{equation}\label{Equ: payoff realized}
	\begin{aligned} 
		&\payoff_n(\plan_n,\val,d,\spedata)
		=\textstyle \QoS_n \val\big(d-\cut[d-\dcap_{\mechanism_n}^e(\spedata)]^+\big) \\
		&\qquad\qquad\qquad\qquad -\adfee_n(1-\cut)[d-\dcap_{\mechanism_n}^e(\spedata)]^+-\pcap_n.
	\end{aligned}	
\end{equation}

For $\mechanism_n=\roll$, the data demand $d$ and the rollover data $\spedata$ in (\ref{Equ: payoff realized}) are two random variables that change in each month.
However, for $\mechanism_n=\trad$, the rollover data $\spedata$ in (\ref{Equ: payoff realized}) is always zero, the randomness only lies in the monthly data demand $d$.
Therefore, we take the expectation over $d$ (and $\spedata$) to get \textit{a type-$\val$ user's expected monthly payoff} under $\plan_n$ as follows:
\begin{equation}\label{Equ: system model payoffexp}
	\begin{aligned}
		\payoffexp_n(\plan_n,\val)=&\mathbb{E}_{d,\spedata} \big\{ \payoff_n\left( \plan_n,\val,d,\spedata \right) \big\} \\
		=	&\QoS_n \val\left[\dmean-\cut A_{\mechanism_n}(\dcap_n)  \right] \\
		&\qquad\qquad -\adfee_n(1-\cut) A_{\mechanism_n}(\dcap_n)-\pcap_n.
	\end{aligned}
\end{equation}

Here $A_{\mechanism_n}(\dcap_n)$ is the user's \textit{expected monthly overage data consumption} under $\plan_n$, which is given by
\begin{equation}\label{Equ: system model A_i(Q_i)}
A_{\mechanism_n}(\dcap_n)
=\left\{
\begin{aligned}
&  \textstyle \sum\limits_{d} \left[d-\dcap_n^e(\spedata) \right]^+f(d) ,	&\text{if }& \mechanism_n=\trad ,\\
&  \textstyle \sum\limits_{\spedata}\sum\limits_{d} \left[d-\dcap_n^e(\spedata) \right]^+f(d)p_{\roll}(\spedata),&\text{if }& \mechanism_n=\roll.
\end{aligned}	
\right.	
\end{equation}

Note that in (\ref{Equ: system model payoffexp}), the difference between the traditional and rollover  mechanisms is entirely captured by $A_{\mechanism_n}(\cdot)$.
Specifically, $p_{\roll}(\cdot)$ represents the distribution of the subscriber's rollover data.
In our previous work \cite{Zhiyuan2017pricing,Zhiyuan2018TMC}, we have introduced how to compute $p_{\roll}(\cdot)$ in detail, and obtained the following inequality 
\begin{equation}\label{Equ: flexibility}
A_{\trad}(\dcap)>A_{\roll}(\dcap),\  \forall\ \dcap\in(0,\dmax),
\end{equation}
which indicates that a user incurs less overage data consumption under the rollover mechanism $\roll$.
This is why we say that {the rollover mechanism $\roll$ offers a better time flexibility than the traditional mechanism $\trad$}.
In this paper, we will directly use this conclusion, and refer interested readers to Section 4 in \cite{Zhiyuan2018TMC} for more details.


Later on, we will study two MNOs' competition given their same data caps.
To facilitate our later analysis, here we further define $\usage_{\mechanism_n}$ as 
\begin{equation}\label{Equ: alpha}
	\usage_{\mechanism_n}\eq\dmean-\cut A_{\mechanism_n}(\dcap_n),
\end{equation}
which represents the user's \textit{expected monthly data consumption} under the data mechanism $\mechanism_n$.
According to (\ref{Equ: flexibility}), we know that the rollover mechanism $\roll$ encourages users to consume more data, i.e.,
\begin{equation}\label{Equ: alpha flexibility}
\usage_{\roll}>\usage_{\trad}.
\end{equation}

The inequality (\ref{Equ: alpha flexibility}) plays a significant roles when we analyze MNOs' competition over their data mechanisms, we will further discuss it in Section \ref{Section: MNO's Data Mechanism Selection in Stage I}.

Substituting (\ref{Equ: alpha}) into (\ref{Equ: system model payoffexp}), we can write \textit{a type-$\val$ $\plan_n$ subscriber's expected monthly payoff} as
\begin{equation}\label{Equ: system model payoffexp a}
	\textstyle \payoffexp_n(\plan_n,\val) = \QoS_n\usage_{\mechanism_n}\val - \adfee_n\left(\cut^{-1}-1\right)\left(\dmean-\usage_{\mechanism_n}\right) -\pcap_n,
\end{equation}
where $\QoS_n\usage_{\mechanism_n}$ represents the user's utility increment for unit data valuation increment under the subscription of MNO-$n$.
In the following analysis, we will directly use (\ref{Equ: system model payoffexp a}).
To emphasize its dependence on the data mechanism $\mechanism_n$ and the pricing strategy $s_n=\{\pcap_n,\adfee_n\}$, sometimes we will write $\payoffexp_n(\plan_n,\val)$ as $\payoffexp_n(\mechanism_n,s_n,\val)$.


\subsection{MNOs' Profits in Competition\label{Subsection: MNOs' Profits}}
Next we focus on two MNOs' market competition and derive their profits given their data mechanism adoption $\Mechanism=\{\mechanism_1,\mechanism_2\}$ and the pricing strategies $\pricing=\{s_1,s_2\}$.

\subsubsection{MNO Revenue}
The MNO's revenue from a single subscriber includes  the subscription fee and the overage payment.
Therefore, the \textit{expected monthly revenue of MNO-$n$ from a type-$\val$ subscriber} is 
\begin{equation}
	\revenueexp_n(\mechanism_n,s_n,\val)=
	\textstyle \adfee_n\left(\cut^{-1}-1\right)\left(\dmean-\usage_{\mechanism_n}\right) + \pcap_n ,
\end{equation}
where $\adfee_n(\cut^{-1}-1)(\dmean-\usage_{\mechanism_n})$ is the subscriber's expected monthly overage payment.
Therefore, the \textit{expected monthly revenue of MNO-$n$ from all its subscribers} is
\begin{equation}\label{Equ: }
	\revenue_n(\Mechanism,\pricing)
	=  \int_{\subscription_n(\Mechanism,\pricing)} \revenueexp_n(\mechanism_n,s_n,\val)h(\val) {\rm d}\val.
\end{equation}
where $\subscription_n(\Mechanism,\pricing)\subseteq[0,\valmax]$ denotes the subscribers of MNO-$n$ under  data mechanism $\Mechanism$ and pricing strategy $\pricing$.
We will further discuss the calculation over the user type $\val$ in Section \ref{Section: User Subscription in Stage III}.

\subsubsection{MNO Cost}
As for the MNO's cost, we focus on its operational expenditure (OpEx).
Specifically, it is proportional to the total data consumption of the MNO's subscribers \cite{nabipay2011flat}.
The \textit{total expected data consumption of MNO-$n$'s  subscribers} is
\begin{equation}
	\load_n(\Mechanism,\pricing)=  \int_{\subscription_n(\Mechanism,\pricing)} \usage_{\mechanism_n} h(\val) {\rm d}\val,
\end{equation}
where $\usage_{\mechanism_n}$ defined in (\ref{Equ: alpha})  represents a user's expected monthly data consumption under $\plan_n$.
For analysis tractability, we follow \cite{nabipay2011flat} by considering a linear cost, and denote $c_n$ as MNO-$n$'s marginal cost from unit data consumption.\footnote{Such a linear-form cost has been widely used to model the operator's operational cost (e.g., \cite{luo2016integrated},\cite{duan2012duopoly}).}
Accordingly, the \textit{total cost of MNO-$n$} is 
\begin{equation}\label{Equ: cost}
	\cost_n(\Mechanism,\pricing) = \load_n(\Mechanism,\pricing)\cdot c_n.
\end{equation}

\subsubsection{MNO Profit}
The \textit{MNO-$n$'s profit} $\profit_n(\Mechanism,\pricing)$ is defined as the difference between its revenue and cost, i.e.,
\begin{equation}\label{Equ: Profit}
	\profit_n(\Mechanism,\pricing)	
	=\revenue_n(\Mechanism,\pricing)-\cost_n(\Mechanism,\pricing). \\
\end{equation}

Now we have formulated MNOs' profits in the duopoly market, and introduced MNOs' two orthogonal characteristics: the QoS parameter $\QoS_n$ as in (\ref{Equ: system model payoffexp a}) and the marginal cost parameter as $c_n$ in (\ref{Equ: cost}). 
Next we use backward induction to study the three-stage game.


\begin{figure*}
	\begin{minipage}{0.22\textwidth}
		\centering
		\includegraphics[width=1\linewidth]{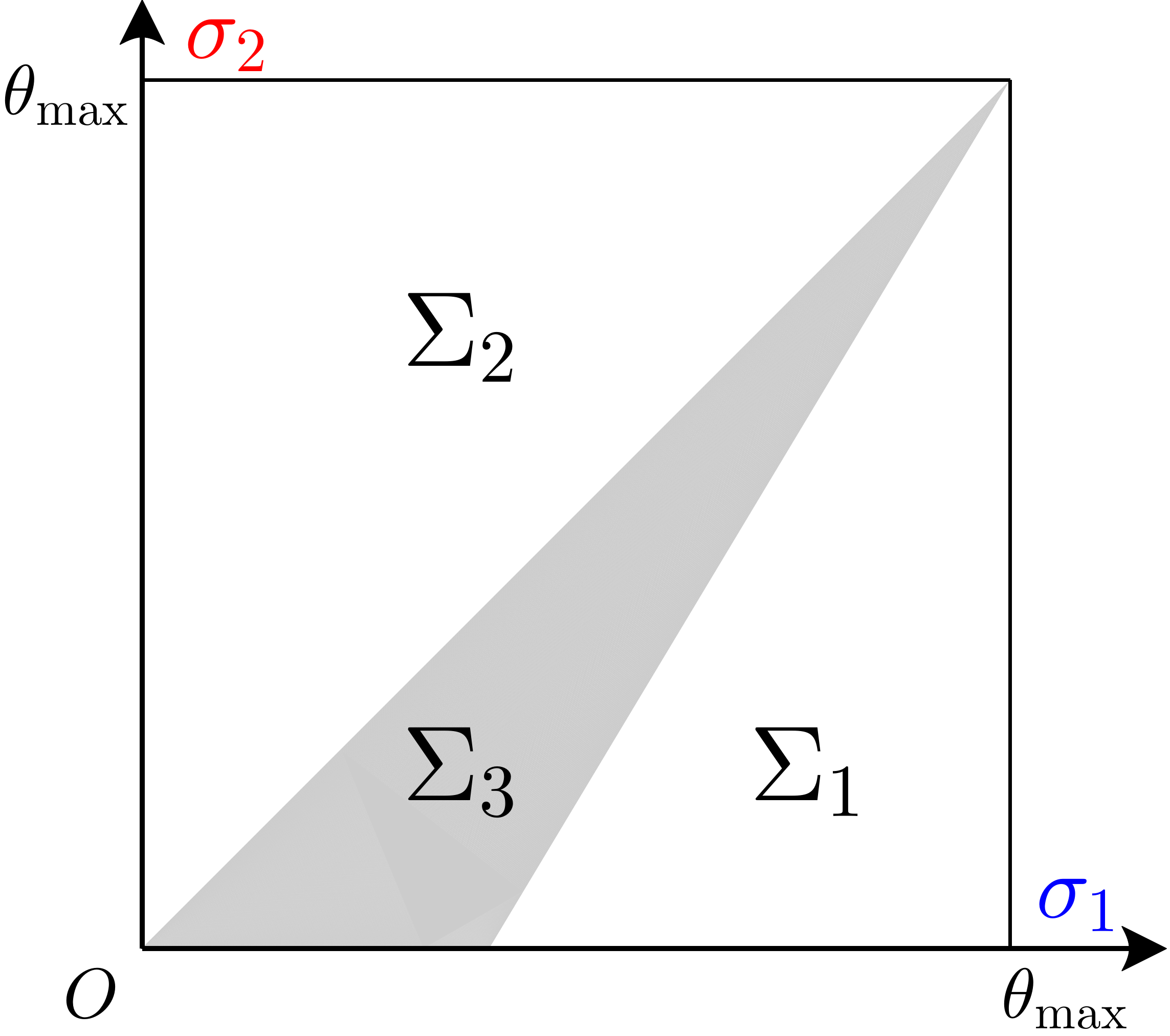}
		\caption{Partition structure.}
		\label{fig: Partition structure}
	\end{minipage}\   
	\begin{minipage}{0.77\textwidth}
		\centering
		\subfigure[MNO-2 surviving ($\Sigma_1$)]{\label{fig: Duopoly_1}{\includegraphics[width=0.3\linewidth]{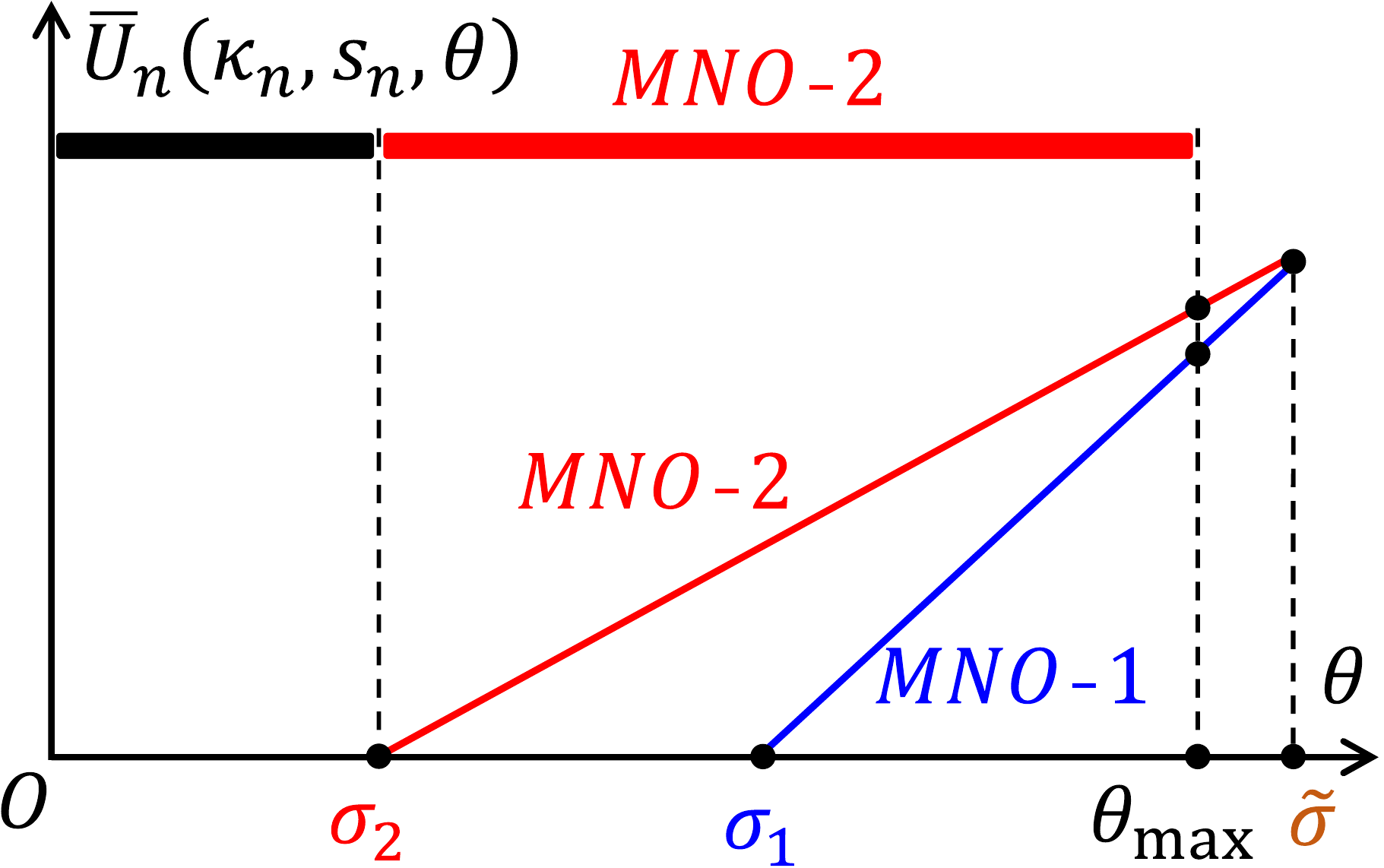}}}\quad
		\subfigure[MNO-1 surviving ($\Sigma_2$)]{\label{fig: Duopoly_3}{\includegraphics[width=0.29\linewidth]{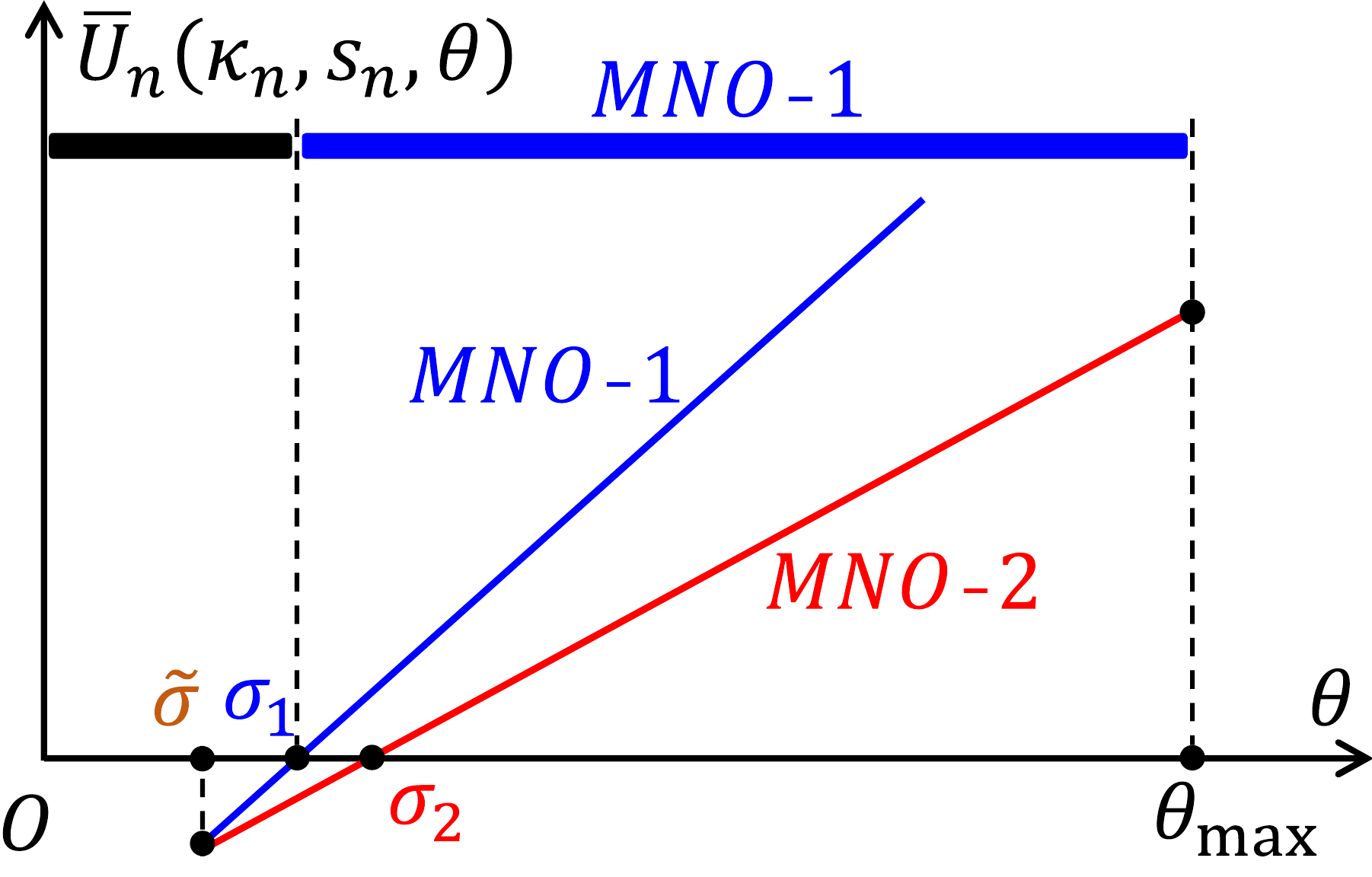}}}\quad 
		\subfigure[Share the market ($\Sigma_3$).]{\label{fig: Duopoly_2}{\includegraphics[width=0.3\linewidth]{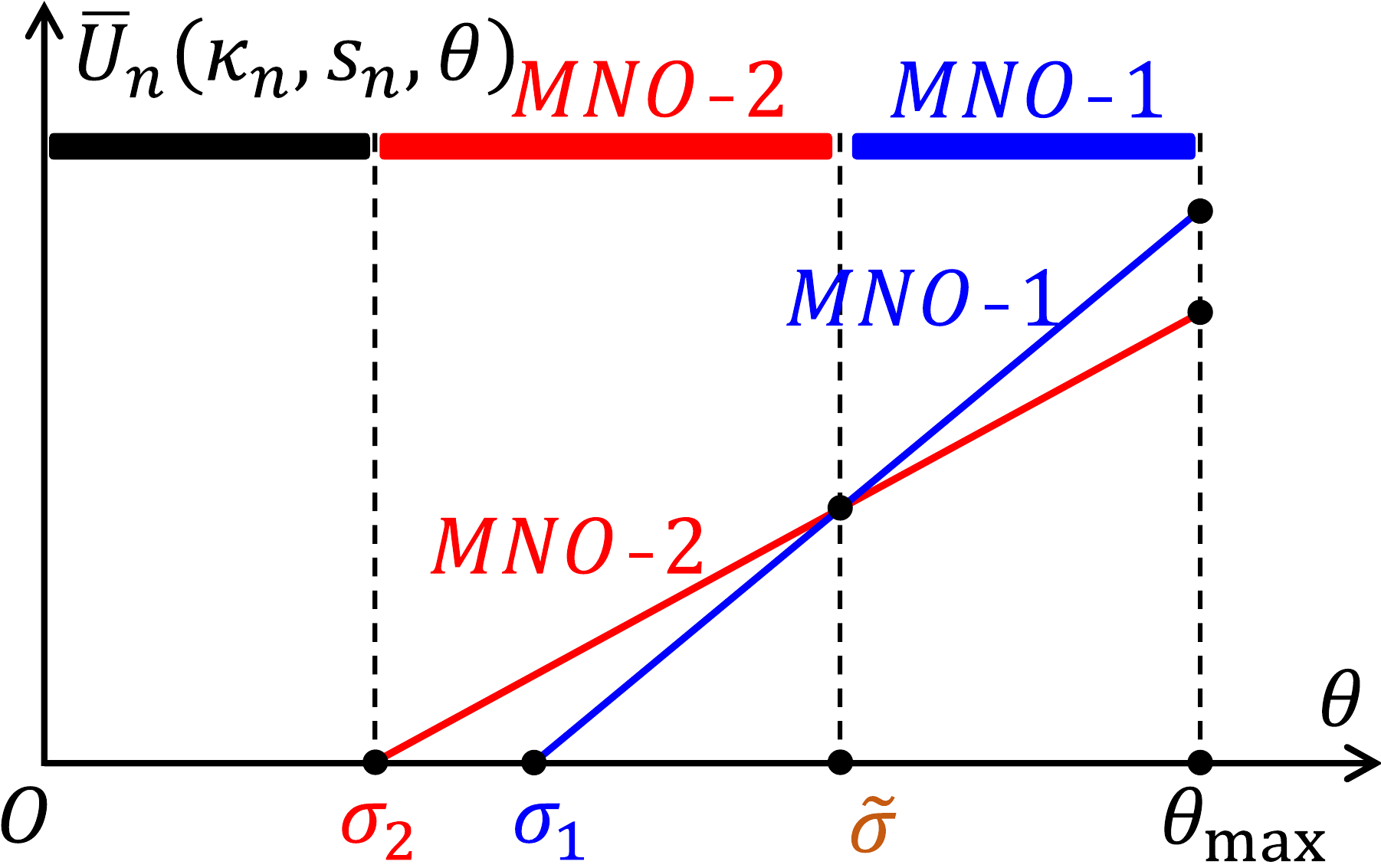}}}		
		\caption{Three market partition modes in the duopoly market. }
		\label{fig: Partition}
	\end{minipage} 
\end{figure*}

\section{User Subscription in Stage III}\label{Section: User Subscription in Stage III}
In Stage III, each user makes his subscription decision given the two MNOs' pricing strategies $\pricing=\{s_1,s_2\}$ in Stage II and data mechanism $\Mechanism=\{\mechanism_1,\mechanism_2\}$ in Stage I.
Specifically, a type-$\val$ user will subscribe to MNO-$n$ if MNO-$n$ can bring him a larger (among the two MNOs) and non-negative payoff.
For notation simplicity, we first introduce two definitions in Section \ref{Subsection: Threshold User Type and Neutral User Type}, then we present the duopoly market partition in Section \ref{Subsection: Duopoly Market Partition}.

\subsection{Threshold User Type and Neutral User Type}\label{Subsection: Threshold User Type and Neutral User Type}
\subsubsection{Threshold User Type}
We introduce the \textit{threshold user type} $\valthr_n$ of MNO-$n$ in Definition \ref{Definition: Threshold User Type}.
\begin{definition}[Threshold User Type]\label{Definition: Threshold User Type}
	The MNO-$n$'s threshold user type $\valthr_n\in[0,\valmax]$ corresponds to the user who achieves a zero expected payoff, i.e., $\payoffexp_n(\mechanism_n,s_n,\valthr_n)=0$.\footnote{Among the user group defined by $\val\in[0,\valmax]$, all users experience a negative payoff if $\valthr_n(\mechanism_n,s_n)>\valmax$.}
	Thus, $\valthr_n$ is
	\begin{equation}\label{Equ: val_n} 
	\valthr_n(\mechanism_n,s_n)
	\triangleq\frac{\adfee_n(\cut^{-1}-1) (\dmean-\usage_{\mechanism_n}) + \pcap_n}{\QoS_n\usage_{\mechanism_n}}.
	\end{equation}	
\end{definition}

Note from (\ref{Equ: val_n}) that the MNO-$n$'s threshold user type $\valthr_n(\mechanism_n,s_n)$ depends on its data mechanism $\mechanism_n$ and the pricing strategy $s_n=\{\pcap_n,\adfee_n\}$.
Moreover, $\QoS_n\usage_{\mechanism_n}$ in (\ref{Equ: val_n}) represents the user's utility increment for unit data valuation increment under the subscription of MNO-$n$.
For notation simplicity, we define $\xi(\Mechanism)$ as in (\ref{Equ: xi}), which will be used in Section \ref{Subsubsection: Neutral User Type} for the definition of the neutral user type.
\begin{equation}\label{Equ: xi} 
\xi(\Mechanism) \eq \frac{\QoS_2\usage_{\mechanism_2}}{\QoS_1\usage_{\mechanism_1}}.
\end{equation}

Note that $\xi(\Mechanism)$ depends on the MNOs' data mechanisms $\Mechanism=\{\mechanism_1,\mechanism_2\}$.
Since we will analyze users' subscription (in Section \ref{Section: User Subscription in Stage III}) and MNOs' pricing competition (in Section \ref{Section: MNOs' Pricing Competition in Stage II}) given data mechanism selection $\Mechanism$ from Stage I, without loss of generality, we make Assumption \ref{Assumption: xi} in Section \ref{Section: User Subscription in Stage III} and Section \ref{Section: MNOs' Pricing Competition in Stage II}.
That is, under the data mechanism selection $\Mechanism$ in Stage I, MNO-1 has an advantage in terms of $\QoS_n\usage_{\mechanism_n}$ among the two MNOs, i.e., its subscribers' marginal utility change  for one unit data valuation increment is larger.
Hence we say MNO-1 is ``stronger''.
\begin{assumption}\label{Assumption: xi}
	Given the data mechanism $\Mechanism$, $\xi(\Mechanism)<1$.
\end{assumption}

Note that \textit{Assumption \ref{Assumption: xi} is not a technical assumption that limits our contributions} for two reasons.
First, we can use a similar approach to analyze the case of $\xi(\Mechanism)>1$ (by switching the indices of the two MNOs).
Second, the case of $\xi(\Mechanism)=1$ corresponds to the well-known Bertrand competition \cite{huang2013wireless}, which is actually a degeneration of the case $\xi(\Mechanism)<1$ and $\xi(\Mechanism)>1$.
We refer interested readers to Appendix \ref{Appendix: Bertrand} for more detailed discussions.

\subsubsection{Neutral User Type}\label{Subsubsection: Neutral User Type}
We introduce the \textit{neutral user type} $\valeq$ between the two MNOs in Definition \ref{Definition: Neutral User Type}.
\begin{definition}[Neutral User Type]\label{Definition: Neutral User Type}
	The neutral user type, denoted by $\valeq$, is a user type who can achieve the same payoff by subscribing to either MNO, i.e., $\payoffexp_1(\mechanism_1,s_1,\valeq)=\payoffexp_2(\mechanism_2,s_2,\valeq)$. 
	We can derive  $\valeq$ as follows
	\begin{equation}\label{Equ: theta_eq}
	 \valeq (\valthr_1,\valthr_2)= \frac{ \valthr_1 - \xi(\Mechanism) \cdot \valthr_2  }{ 1 - \xi(\Mechanism) },
	\end{equation}
	where $\valthr_1$ and $\valthr_2$ are two MNOs' threshold user types defined in {Definition \ref{Definition: Threshold User Type}}.
\end{definition}

So far we have introduced the threshold user type and the neutral user type.
Next let us move on to the subscription equilibrium in the duopoly market.

\subsection{Duopoly Market Partition}\label{Subsection: Duopoly Market Partition}
Based on Definitions \ref{Definition: Threshold User Type} and \ref{Definition: Neutral User Type},  we summarize the market partition in Theorem \ref{Theorem: Market Partition}.
The proof is given in Appendix \ref{Appendix: Partition BRs}.
\begin{theorem}[Market Partition Equilibrium]\label{Theorem: Market Partition}
	Consider  MNOs' threshold user types $\valthr_1$ and $\valthr_2$  under the data mechanism $\Mechanism=\{\mechanism_1,\mechanism_2\}$ and the pricing strategy $\pricing=\{s_1,s_2\}$.
	The market partition equilibrium, denoted by $\subscription_1^{*}(\Mechanism,\pricing)$ and $\subscription_2^{*}(\Mechanism,\pricing)$, has three cases in the $(\valthr_1,\valthr_2)$ plane shown in Fig. \ref{fig: Partition structure}.
	\begin{enumerate}
		\item $\Sigma_1$: MNO-1 has a much larger threshold user type than MNO-2, i.e., $(\valthr_1,\valthr_2)\in\Sigma_1$ where $\Sigma_1$ is
		\begin{equation}
		\Sigma_1\eq\{(\valthr_1,\valthr_2):\valthr_1-\valthr_2 \ge (1-\xi(\Mechanism))(\valmax-\valthr_2)\}.
		\end{equation}		
		In this case, 	MNO-2's market share corresponds to the users with $\val$ in  $\subscription_2^{*}(\Mechanism,\pricing)=[\valthr_2,\valmax]$, while MNO-1 has a  zero market share $\subscription_1^{*}(\Mechanism,\pricing)=\varnothing$, as shown in Fig. \ref{fig: Duopoly_1}.
		
		\item $\Sigma_2$: MNO-1 has a smaller threshold user type than MNO-2, i.e., $(\valthr_1,\valthr_2)\in\Sigma_2$ where $\Sigma_2$ is
		\begin{equation}
		\Sigma_2\eq\{(\valthr_1,\valthr_2):\valthr_1 -\valthr_2 \le 0\}.
		\end{equation}
		In this case, MNO-1 has a market share of $\subscription_1^{*}(\Mechanism,\pricing)=[\valthr_1,\valmax]$, while MNO-2 has a zero market share of $\subscription_2^{*}(\Mechanism,\pricing)=\varnothing$, as shown in Fig. \ref{fig: Duopoly_3}.
		
		\item $\Sigma_3$: MNO-1 has a slightly larger threshold user type than MNO-2, i.e., $(\valthr_1,\valthr_2)\in\Sigma_3$ where $\Sigma_3$ is
		\begin{equation}
		\Sigma_3\eq\{(\valthr_1,\valthr_2):0 < \valthr_1 -\valthr_2 < (1-\xi(\Mechanism))(\valmax-\valthr_2)\}.
		\end{equation} 
		In this case, MNO-1 has a market share of  $\subscription_1^{*}(\Mechanism,\pricing)=[\valeq,\valmax]$, and MNO-2 has a market share of  $\subscription_2^{*}(\Mechanism,\pricing)=[\valthr_2,\valeq]$, as shown in Fig. \ref{fig: Duopoly_2}.
	\end{enumerate}
\end{theorem}

Theorem \ref{Theorem: Market Partition}  reveals two market partition equilibriums, i.e., coexistence (i.e., $\Sigma_3$) or one-MNO-surviving (i.e., $\Sigma_1$, and $\Sigma_2$), depending on the data mechanism $\Mechanism$ and the pricing strategy $\pricing$.
Note that the one-MNO-surviving result is different from the monopoly case, since the zero market share MNO in the one-MNO-surviving case  might still affect the decisions of the surviving MNO.
We will further discuss it in Section \ref{Subsection: Best Response Analysis}.

\section{MNOs' Pricing Competition in Stage II}\label{Section: MNOs' Pricing Competition in Stage II}
In Stage II, the MNOs simultaneously determine the pricing strategies $\pricing$=$\{s_1,s_2\}$, given their data mechanisms $\Mechanism$=$\{\mechanism_1,\mechanism_2\}$ in Stage I and  the market partition equilibrium in Stage III.

We substitute the subscription equilibrium $\subscription_1^{*}(\Mechanism,\pricing)$ and $\subscription_2^{*}(\Mechanism,\pricing)$ from Theorem \ref{Theorem: Market Partition} into (\ref{Equ: Profit}) and derive the two MNOs' profits as follows:
\begin{equation}\label{Equ: Profit-1 III}
\profit_1(\Mechanism,\Threshold) 
= \QoS_1\usage_{\mechanism_1} \left[ \valthr_1 - \textstyle\frac{c_1}{\QoS_1} \right]  \Big[ 1-H\big(\valeq(\valthr_1,\valthr_2)  \big) \Big] , 
\end{equation}
\begin{equation}\label{Equ: Profit-2 III}
\profit_2(\Mechanism,\Threshold) 
= \QoS_2\usage_{\mechanism_2} \left[ \valthr_2 - \textstyle\frac{c_2}{\QoS_2} \right] \Big[ H\big( \valeq(\valthr_1,\valthr_2) \big) - H( \valthr_2) \Big],
\end{equation}
where $\valthr_1$, $\valthr_2$, and $\valeq$ depend on the pricing strategies $\pricing$ and data mechanisms $\Mechanism$.

According to (\ref{Equ: Profit-1 III}) and (\ref{Equ: Profit-2 III}), we note that the MNOs' profits at the equilibrium of Stage III are uniquely determined by the data mechanism $\Mechanism=\{\mechanism_1,\mechanism_2\}$ and the threshold user types $\Threshold=\{\valthr_1,\valthr_2\}$.
Moreover, (\ref{Equ: val_n}) shows that MNO-$n$ is able to achieve an arbitrary threshold user type $\valthr_n(\mechanism_n,s_n)$ by adjusting the pricing strategies $s_n=\{\pcap_n,\adfee_n\}$.
Hence the MNOs' price competition in Stage II is equivalent to the following threshold competition game:
\begin{game}[Threshold Competition in Stage II]\label{Game: threshold}
	Given the data mechanism $\Mechanism=\{\mechanism_1,\mechanism_2\}$, the two MNOs' threshold competition in Stage II can be modeled as the following game:
	\begin{itemize}
		\item Players: MNO-$n$ for both $n=1,2$.
		\item Strategies: Each MNO-$n$ determines its threshold user type $\valthr_n\in[c_n/\QoS_n,\valmax]$.
		\item Preferences: Each MNO-$n$  obtains a  profit $\profit_n(\Mechanism,\Threshold)$.
	\end{itemize}
\end{game}

Next we will study the MNOs' best responses of \textit{Game \ref{Game: threshold}} in Section \ref{Subsection: Best Response Analysis}, then find the the equilibrium which corresponds to the fixed point of the best responses in Section \ref{Subsection: Equilibrium Analysis}.
Before that, we introduce two notations as follows.
\begin{itemize}
	\item $\valthr_n^\text{MP}(\mechanism_n)$: Part of the best response analysis is related to the MNO-$n$'s optimal threshold user type $\valthr_n^\text{MP}(\mechanism_n)$ given the data mechanism $\mechanism_n$ in the \textit{monopoly} market, which is studied in our previous work \cite{Zhiyuan2017pricing,Zhiyuan2018TMC}.
	In the following, we will directly use $\valthr_n^\text{MP}(\mechanism_n)$ and provide more details in Appendix \ref{Appendix: Monopoly Market as Benchmark}.	
	\item $\costqos_n$: As we will see later, the MNO's cost-QoS ratio $c_n/\QoS_n$ plays a significant role in the best response analysis.
	For notation simplicity, we define $\costqos_n$ as follows:
	\begin{equation}
	\costqos_n \eq \frac{c_n}{\QoS_n},\ \forall\ n\in\{1,2\}.
	\end{equation}
\end{itemize}

\begin{figure*}
	\centering
	\begin{minipage}{0.25\textwidth}
		\setlength{\abovecaptionskip}{13pt}
		\includegraphics[height=0.82\linewidth]{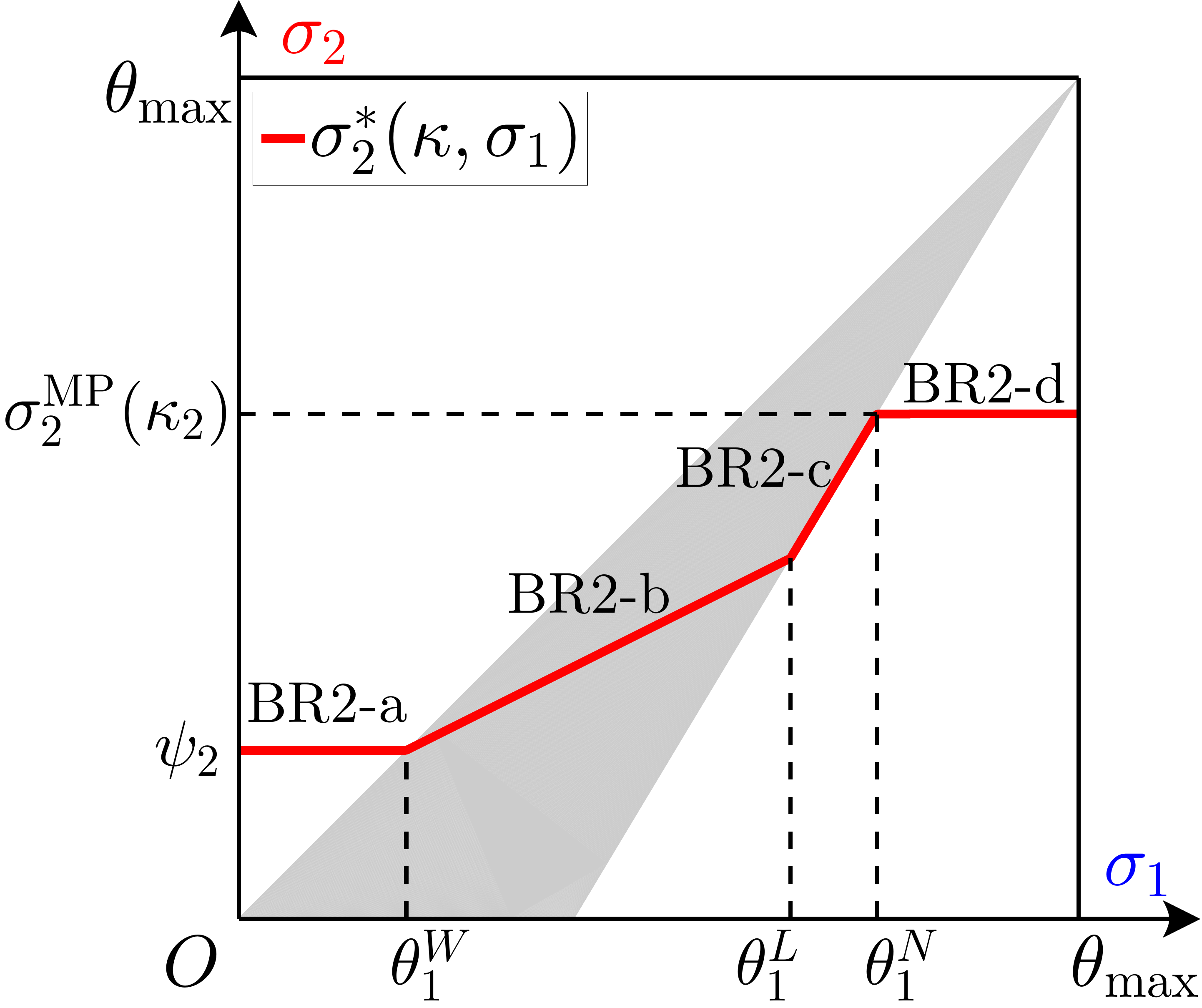}
		\caption{Illustration of $\valthr_2^{*}(\Mechanism,\valthr_1)$.}
		\label{fig: Best response of MNO-2}
	\end{minipage} 
	\begin{minipage}{0.48\textwidth}
		\setlength{\abovecaptionskip}{-1pt}
		\setlength{\belowcaptionskip}{0pt}
		\subfigure[$\costqos_1 > (1-\xi(\Mechanism))\valmax$]{\label{fig: Best response of MNO-1 2}{\includegraphics[height=0.43\linewidth]{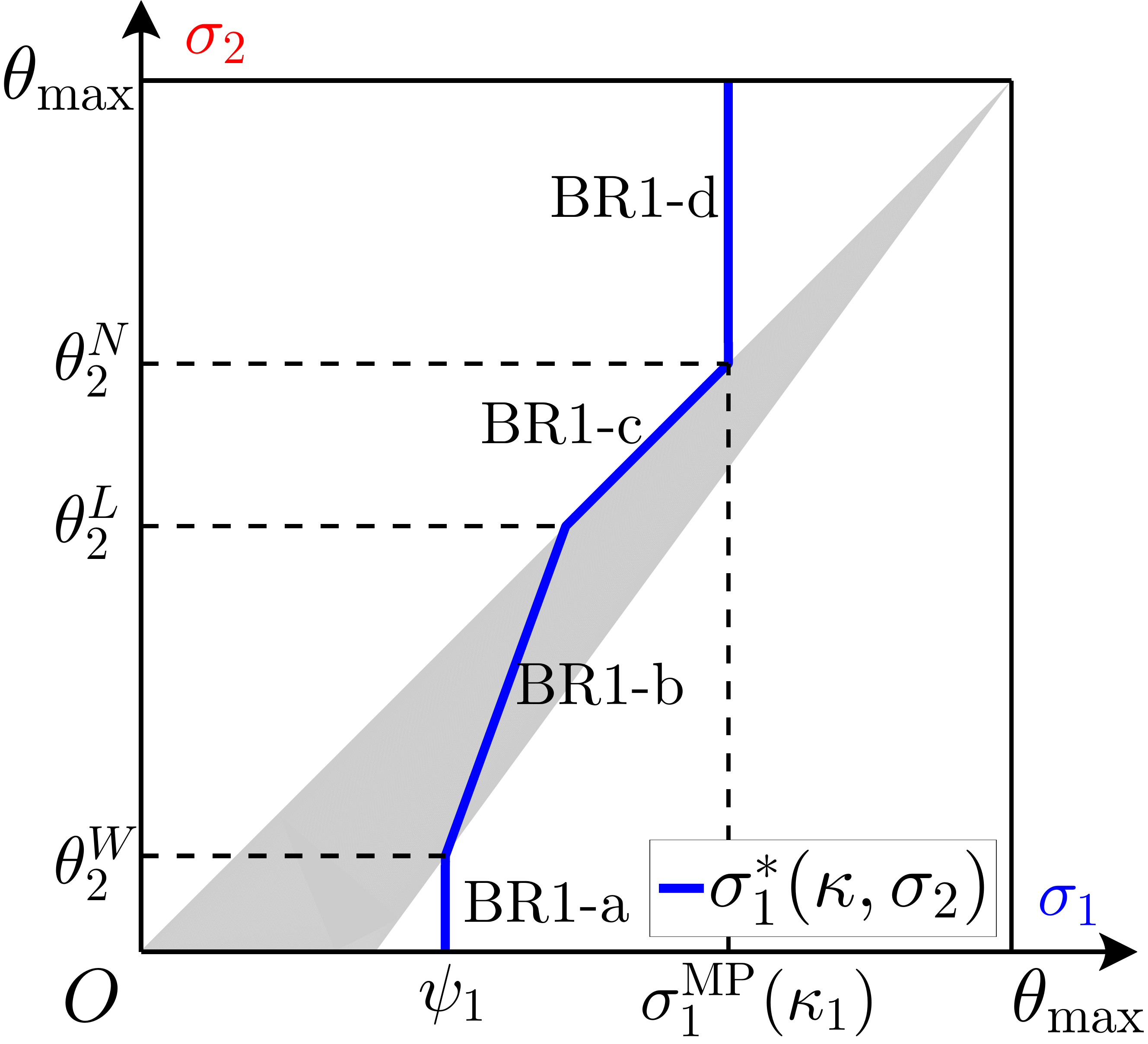}}}\ 
		\subfigure[$\costqos_1 < (1-\xi(\Mechanism))\valmax$]{\label{fig: Best response of MNO-1 1}{\includegraphics[height=0.43\linewidth]{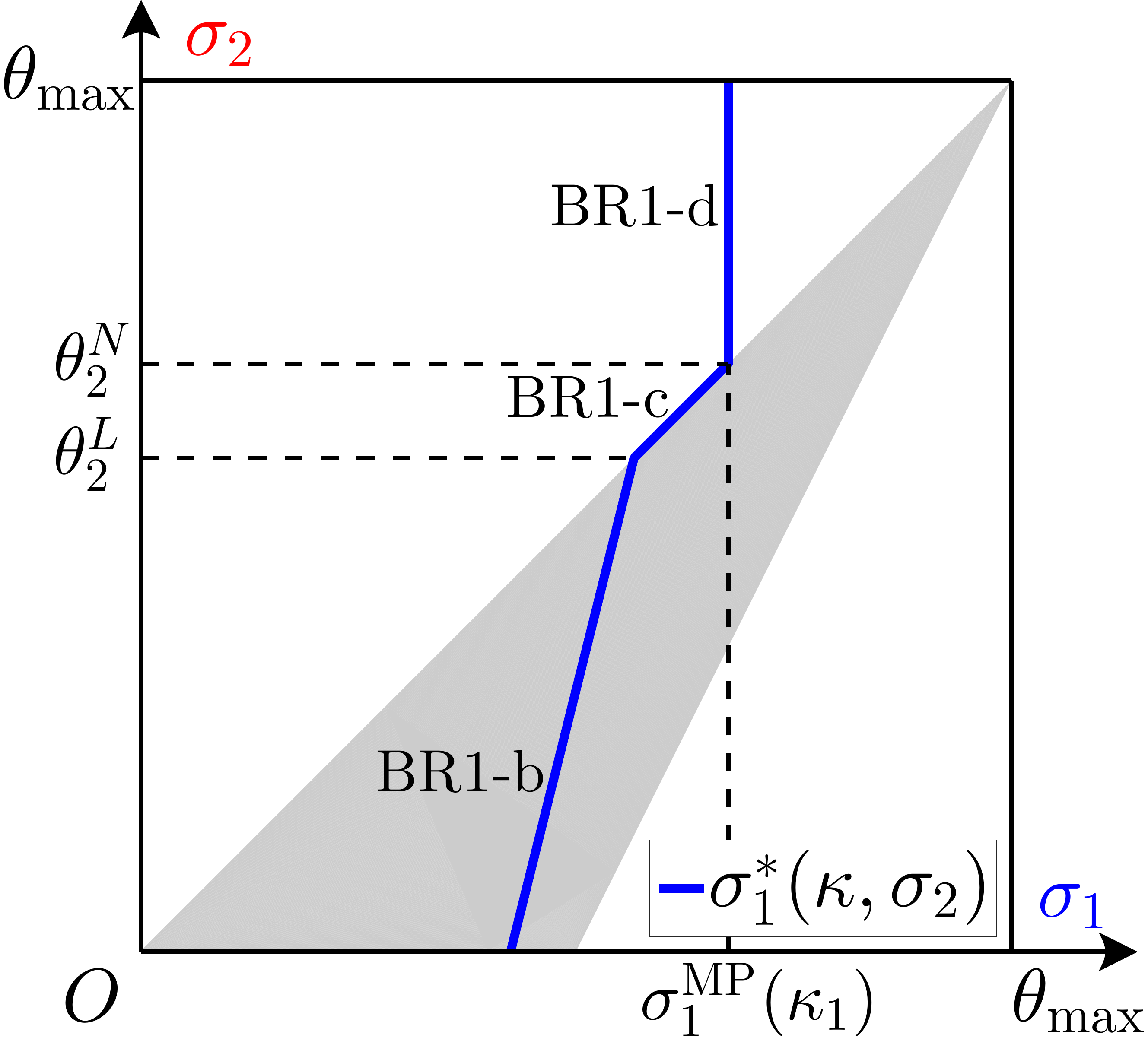}}} 
		\caption{Illustration of $\valthr_1^{*}(\Mechanism,\valthr_2)$.}
		\label{fig: Best response of MNO-1}
	\end{minipage} 
	\begin{minipage}{0.26\textwidth}
		\setlength{\abovecaptionskip}{10pt}
		\includegraphics[width=0.98\linewidth]{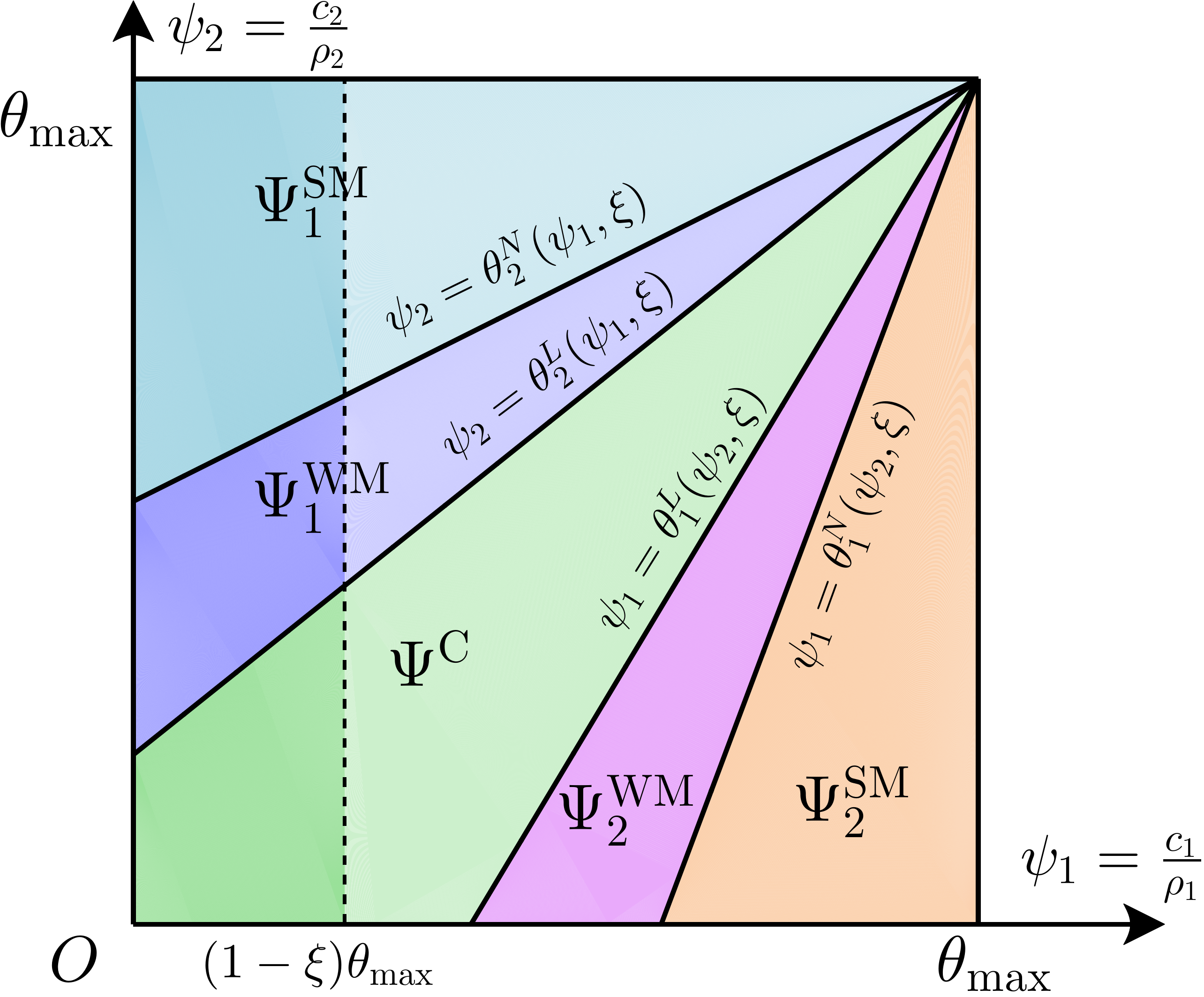}
		\caption{The outcomes in Game \ref{Game: threshold}.}
		\label{fig: Structure}
	\end{minipage}  \vspace{-10pt}
\end{figure*}

\begin{figure*}  
	\centering
	\setlength{\abovecaptionskip}{0pt}
	\subfigure[$(\costqos_1,\costqos_2)\in\Psi_1^\text{SM}$]{\label{fig: R4}\includegraphics[height=0.182\linewidth]{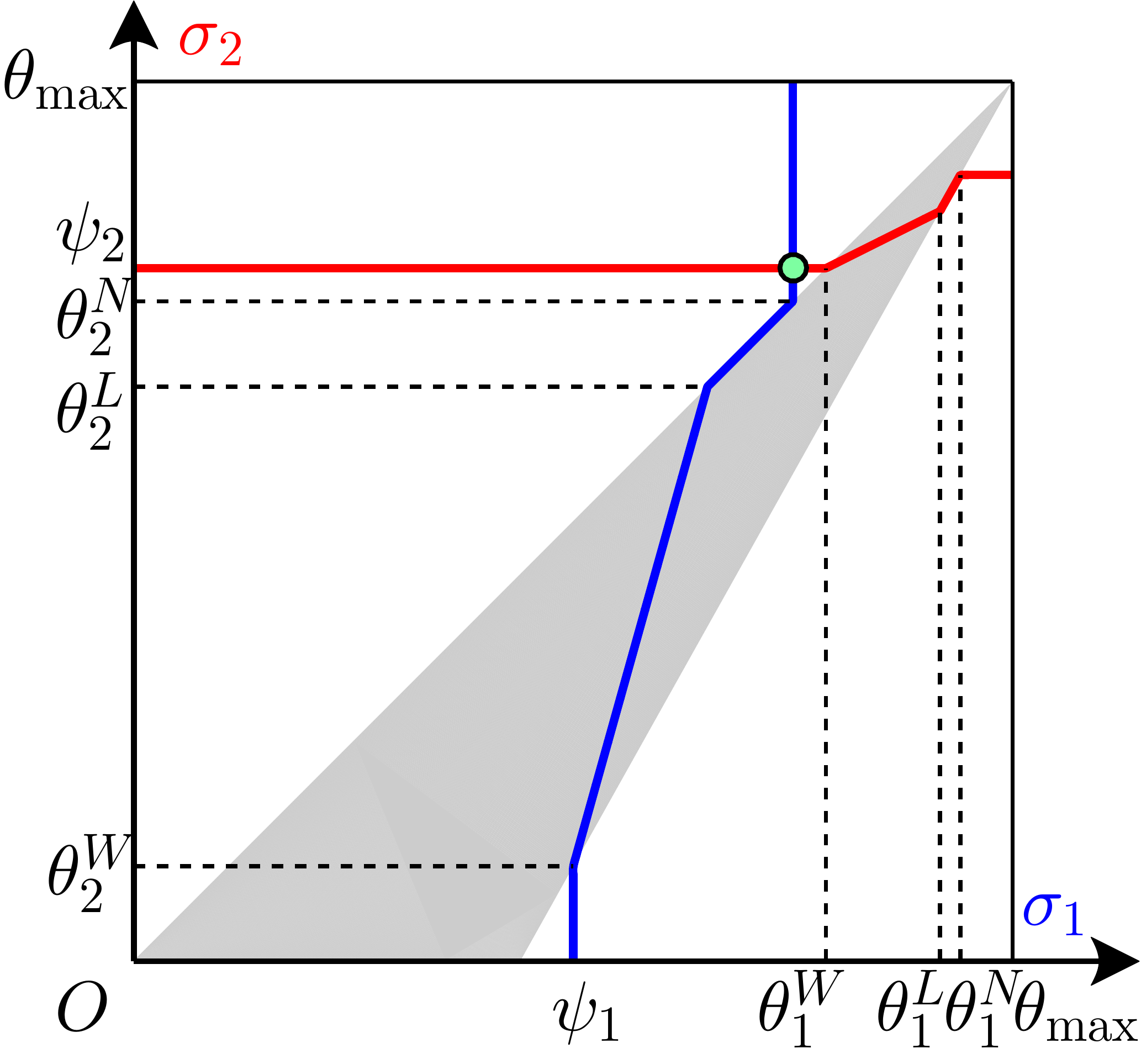}}
	\subfigure[$(\costqos_1,\costqos_2)\in\Psi_1^\text{WM}$]{\label{fig: R5}\includegraphics[height=0.182\linewidth]{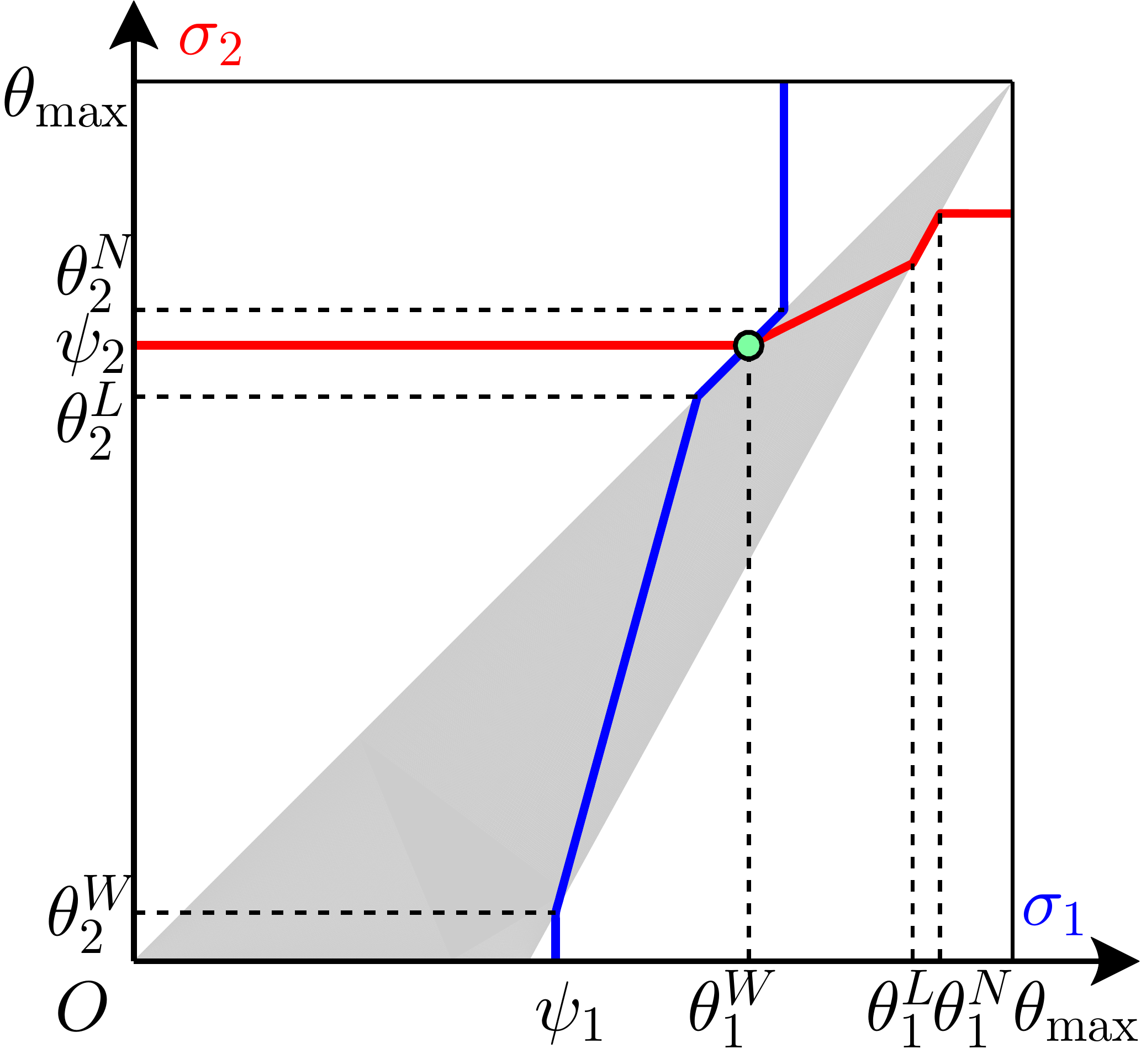}}
	\subfigure[$(\costqos_1,\costqos_2)\in\Psi^\text{C}$]{\label{fig: R6}\includegraphics[height=0.182\linewidth]{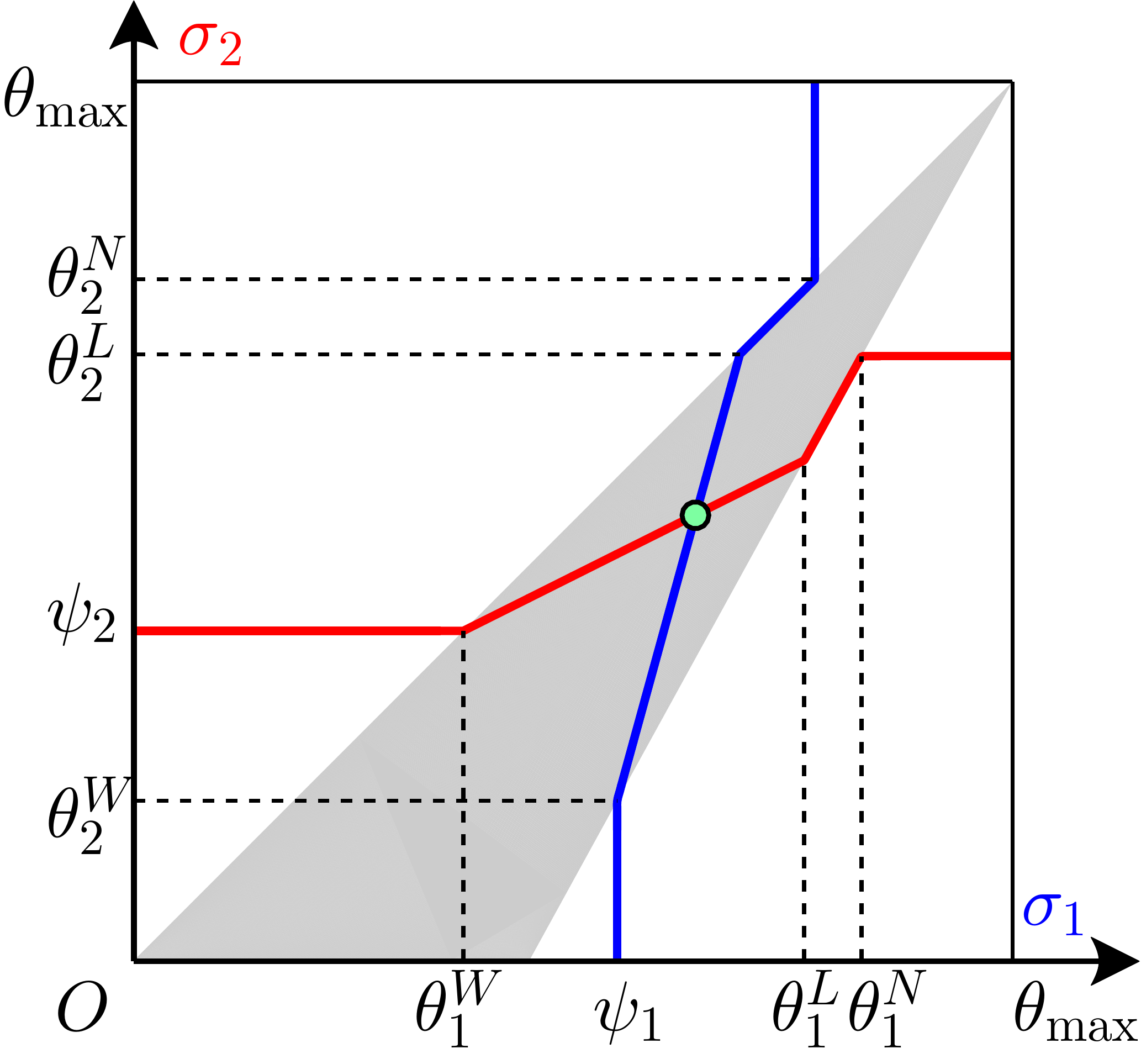}}
	\subfigure[$(\costqos_1,\costqos_2)\in\Psi_2^\text{WM}$]{\label{fig: R7}\includegraphics[height=0.182\linewidth]{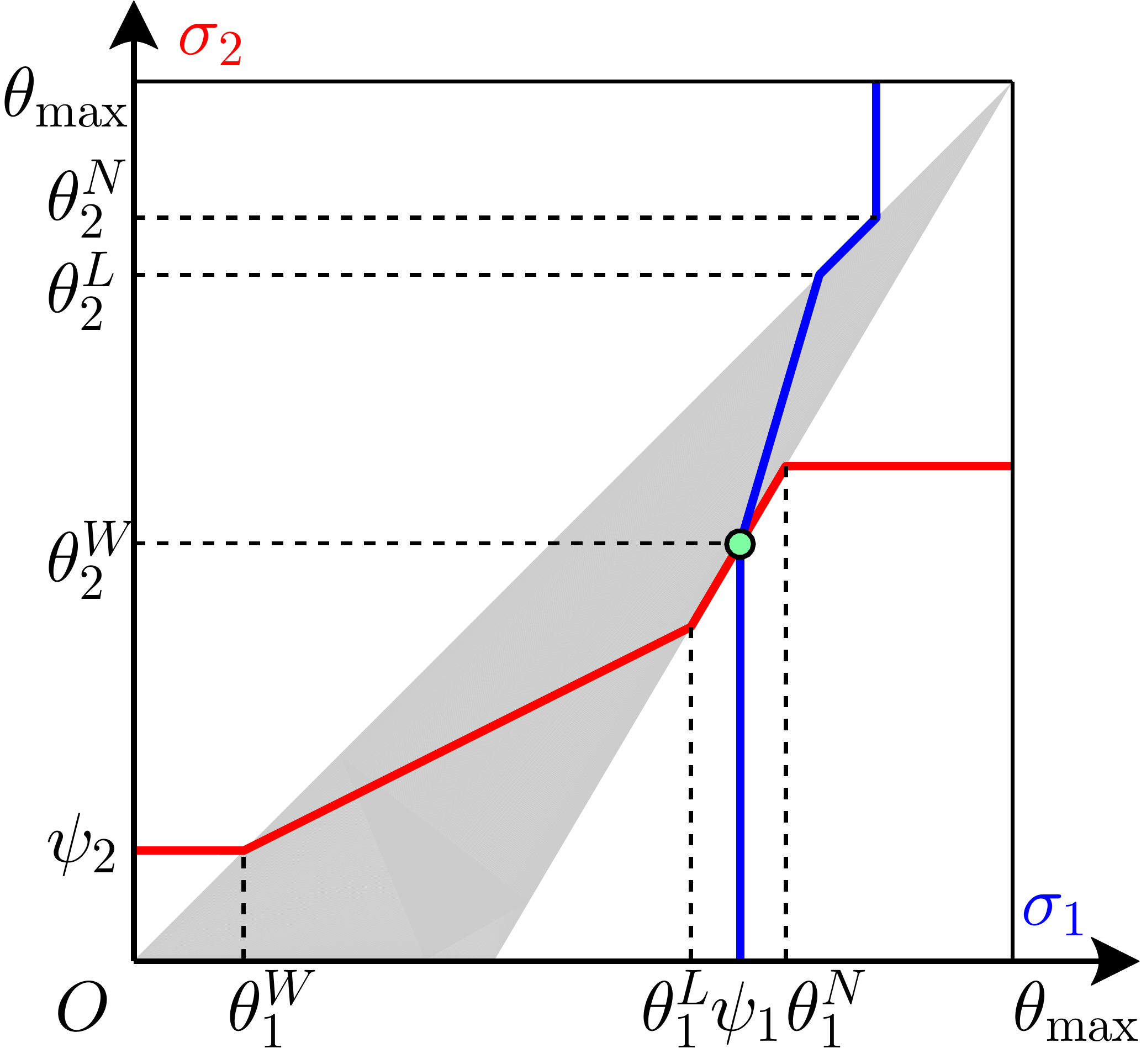}}
	\subfigure[$(\costqos_1,\costqos_2)\in\Psi_2^\text{SM}$]{\label{fig: R8}\includegraphics[height=0.182\linewidth]{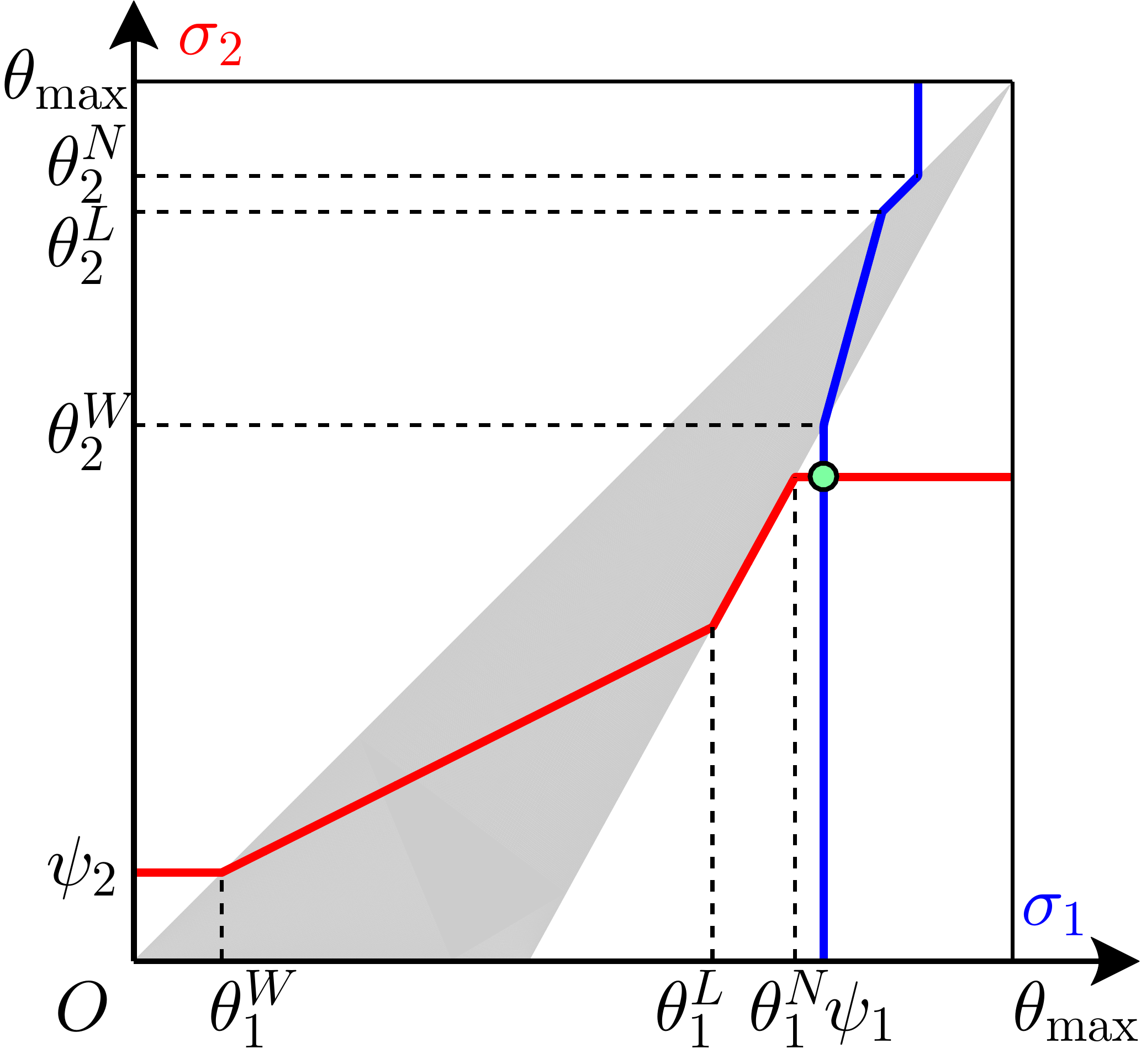}}
	\caption{The illustration of Game \ref{Game: threshold} equilibriums for the five $(\costqos_1,\costqos_2)$ regimes shown in Fig. \ref{fig: Structure}. }
	\label{fig: structure}
\end{figure*}

\subsection{Best Response Analysis}\label{Subsection: Best Response Analysis}
Since we have assumed that  MNO-1 is stronger and MNO-2 is weaker, the two MNOs'  best responses will be different. 
We will investigate the best response of MNO-2 and MNO-1 in Section \ref{Subsubsection: Best Response of Weaker MNO-2} and Section \ref{Subsubsection: Best Response of Stronger MNO-1}, respectively.

\subsubsection{Best Response of the Weaker MNO-2}\label{Subsubsection: Best Response of Weaker MNO-2}
To facilitate the analysis of the weaker MNO-2's best response to the stronger MNO-1's threshold user choice, we first introduce  the MNO-1's \textit{winning} threshold $\val_1^W\eq \costqos_2$, \textit{no-influence} threshold $\val_1^N\eq\xi(\Mechanism)\valthr_2^\text{MP}(\mechanism_2) + (1-\xi(\Mechanism))\valmax$, and \textit{losing} threshold $\val_1^L$ that satisfies
\begin{equation}\label{Equ: three points of MNO-1}
\textstyle 
\frac{\val_1^L-(1-\xi(\Mechanism))\valmax}{\xi(\Mechanism)} 
-\frac{ 1 - H\left( \frac{\val_1^L-(1-\xi(\Mechanism))\valmax}{\xi(\Mechanism)} \right)}{ h\left( \frac{\val_1^L-(1-\xi(\Mechanism))\valmax}{\xi(\Mechanism)}\right) } = \costqos_2,
\end{equation}
where $\val_1^L$ is unique for an arbitrary $\val$ distribution with the increasing failure rate (IFR).\footnote{The IFR condition means that $h(\val)/[1 - H(\val)]$ increases in $\val$. Many commonly used distributions (e.g., uniform distribution, normal distribution, and gamma distribution) satisfy the IFR condition.}

Lemma \ref{Lemma: theta_2^*} presents the the weaker MNO-2's best response.
The proof is given in Appendix \ref{Appendix: Partition BRs}.
\begin{lemma}[Best Response of MNO-2] \label{Lemma: theta_2^*}
	Given MNO-1's threshold user type $\valthr_1$,  MNO-2 maximizes its profit $\profit_2(\Mechanism,\Threshold)$ by choosing a  threshold user type $\valthr_2^{*}(\Mechanism,\valthr_1) $ as follows: 
	\begin{equation}
		\valthr_2^{*}(\Mechanism,\valthr_1)=
		\begin{cases} 
		\costqos_2, & \text{if } \valthr_1\in [0,\val_1^W) , \\
		\hat{\valthr}_2			     ,  & \text{if } \valthr_1\in [\val_1^W,\val_1^L),\\
		\frac{\valthr_1 + (\xi(\Mechanism)-1)\valmax}{\xi(\Mechanism)},	& \text{if } \valthr_1 \in [\val_1^L,\val_1^N),\\
		\valthr_2^\text{MP}(\mechanism_2), 		& \text{if }\valthr_1 \in [\val_1^N,\valmax].
		\end{cases}
	\end{equation} 
	Here $\hat{\valthr}_2$ solves $\hat{\valthr}_2-\frac{ H( \valeq(\valthr_1,\hat{\valthr}_2) ) - H(\hat{\valthr}_2)}{\frac{\xi(\Mechanism)}{1-\xi(\Mechanism)}h( \valeq(\valthr_1,\hat{\valthr}_2) ) +h(\hat{\valthr}_2)} =\costqos_2$, where $h(\cdot)$ and $H(\cdot)$ represent the PDF and CDF of users' data valuation $\val$, respectively. 
\end{lemma}

Lemma \ref{Lemma: theta_2^*} is applicable to an arbitrary $\val$ distribution satisfying the IFR.
For the illustration purpose, Fig. \ref{fig: Best response of MNO-2} plots the MNO-2's best response under a uniform distribution of $\val$, i.e., $h(\val)=1/\valmax$. 
Specifically, the red line segments (i.e., BR2-a, BR2-b, BR2-c, and BR2-d) denote $\valthr_2^{*}(\Mechanism,\valthr_1)$. 
Next we discuss the physical meanings of the four different parts of best response in Lemma \ref{Lemma: theta_2^*}.

\begin{itemize}
	\item \textbf{BR2-a}: MNO-2 gives up the competition  and obtains a zero market share, i.e., $(\valthr_1,\valthr_2^{*}(\Mechanism,\valthr_1))\in\Sigma_2$, if MNO-1's  threshold user type is smaller than its \textit{winning} threshold, i.e., $\valthr_1< \val_1^W$.
	\item \textbf{BR2-b}: MNO-2 shares the market with MNO-1, i.e., $(\valthr_1,\valthr_2^{*}(\Mechanism,\valthr_1))\in\Sigma_3$, if MNO-1's threshold user type is between its \textit{winning} and \textit{losing} thresholds, i.e., $\val_1^W\le\valthr_1<\val_1^L$.
	\item \textbf{BR2-c}: MNO-2 leaves MNO-1 a zero market share, i.e., $(\valthr_1,\valthr_2^{*}(\Mechanism,\valthr_1))\in\Sigma_1$, if MNO-1's threshold user type is between its \textit{losing} and \textit{no-influence} thresholds, i.e., $\val_1^L\le\valthr_1<\val_1^N$.
	\item \textbf{BR2-d}: MNO-2 becomes a monopoly in the market (deciding  its threshold user type without considering the existence of MNO-1), i.e., $\valthr_2^{*}(\Mechanism,\valthr_1)=\valthr_2^\text{MP}(\mechanism_2)$, if MNO-1's threshold user type is no smaller than its \textit{no-influence} threshold, i.e., $\val_1^N\le\valthr_1\le\valmax$.
\end{itemize}

\subsubsection{Best Response of the Stronger MNO-1}\label{Subsubsection: Best Response of Stronger MNO-1}
Now we consider the stronger MNO-1's best response to the weaker MNO-2.
Similarly, we first define MNO-2's \textit{winning} threshold $\val_2^W\eq  \frac{\costqos_1 + (\xi(\Mechanism)-1)\valmax}{\xi(\Mechanism)} $, \textit{no-influence} threshold $\val_2^N \eq \valthr_1^\text{MP}(\mechanism_1)$, and \textit{losing} threshold $\val_2^L$ that satisfies 
\begin{equation}
\textstyle \val_2^L-\frac{ 1-H\left( \val_2^L \right) }{\frac{1}{1-\xi(\Mechanism)}h\left( \val_2^L \right)} = \costqos_1 ,
\end{equation}
where $\val_2^L$ is unique for an arbitrary $\val$ distribution with the IFR.

\begin{lemma}[Best Response of MNO-1] \label{Lemma: theta_1^*}
	Given MNO-2's threshold user type $\valthr_2$,  MNO-1 maximizes its profit $\profit_1(\Mechanism,\Threshold)$ by the threshold user type $\valthr_1^{*}(\Mechanism,\valthr_2)$, which satisfies
	\begin{equation}
	\valthr_1^{*}(\Mechanism,\valthr_2) = 
	\begin{cases}
	\costqos_1,		&\text{if }\valthr_2\in[0,\val_2^W),\\
	\hat{\valthr}_1,		&\text{if }\valthr_2 \in[\val_2^W,\val_2^L),\\
	\valthr_2,				&\text{if }\valthr_2 \in[\val_2^L,\val_2^N),\\
	\valthr_1^\text{MP}(\mechanism_1),	&\text{if }\valthr_2 \in[\val_2^N,\valmax],
	\end{cases}
	\end{equation}  
	where $\hat{\valthr}_1$ solves $\hat{\valthr}_1-\frac{ 1-H( \valeq(\hat{\valthr}_1,\valthr_2) ) }{\frac{1}{1-\xi(\Mechanism)}h( \valeq(\hat{\valthr}_1,\valthr_2) )} = \costqos_1$.
\end{lemma}

Lemma \ref{Lemma: theta_1^*} applies to an arbitrary $\val$ distribution satisfying the IFR. Fig. \ref{fig: Best response of MNO-1} illustrates the results in Lemma \ref{Lemma: theta_1^*} under a uniform distribution. 
For an easy comparison with Fig. \ref{fig: Best response of MNO-2}, in Fig. \ref{fig: Best response of MNO-1} we plot the best response $\valthr_1^{*}(\Mechanism,\valthr_2)$ on the horizontal axis and the variable $\valthr_2$ on the vertical axis.
Moreover, Fig. \ref{fig: Best response of MNO-1} contains two cases (sub-figures): 
\begin{itemize}
	\item Fig. \ref{fig: Best response of MNO-1 2}: When MNO-1 has a relatively large cost-QoS ratio, i.e., $\costqos_1>(1-\xi(\Mechanism))\valmax$, the corresponding   insights are similar to Fig. \ref{fig: Best response of MNO-2}.
	\item Fig. \ref{fig: Best response of MNO-1 1}: When MNO-1 has a small cost-QoS ratio, i.e., $\costqos_1\le(1-\xi(\Mechanism))\valmax$,  MNO-2's \textit{winning} threshold $\val_2^W$ is always negative.
	This means that no matter how small the MNO-2's threshold user type $\valthr_2$ is, MNO-1 can always get a positive market share, i.e., $(\valthr_1^{*}(\Mechanism,\valthr_2),\valthr_2)\notin\Sigma_1$ for all $\valthr_2$.
	This is possible as MNO-2 is the weaker one. Fig. \ref{fig: Best response of MNO-1 2} represents a degenerated case of Fig. \ref{fig: Best response of MNO-1 1}. 
\end{itemize}

\subsection{Equilibrium Analysis}\label{Subsection: Equilibrium Analysis}
Based on the above best responses, i.e., $\valthr_1^{*}(\Mechanism,\valthr_2)$ and $\valthr_2^{*}(\Mechanism,\valthr_1)$ characterized in Lemmas \ref{Lemma: theta_2^*} and \ref{Lemma: theta_1^*}, next we will study the threshold equilibrium of \textit{Game \ref{Game: threshold}}, denoted by $\Threshold^*=\{\valthr_1^*,\valthr_2^*\}$.

Theorem \ref{Theorem: Equilibrium Regime} characterizes the different outcomes of \textit{Game \ref{Game: threshold}} (based on the $(\costqos_1,\costqos_2)$ plane) and the corresponding equilibrium $\Threshold^*=\{\valthr_1^*,\valthr_2^*\}$ of each outcome.
Moreover, Fig. \ref{fig: Structure} and Fig. \ref{fig: structure} illustrate the general results of Theorem \ref{Theorem: Equilibrium Regime} assuming a uniform distribution of $\val$.
In Fig. \ref{fig: Structure}, the horizontal and vertical axises correspond to $\costqos_1$ and $\costqos_2$, respectively.
The  five $(\costqos_1,\costqos_2)$ regions correspond to the different outcomes of \textit{Game \ref{Game: threshold}}.
In Fig. \ref{fig: structure}, each sub-figure represents the two MNOs' best responses under one of the five outcomes in Fig. \ref{fig: Structure}. 

\begin{theorem} \label{Theorem: Equilibrium Regime}
	Game \ref{Game: threshold} has five different types of equilibrium based on the values of $(\costqos_1,\costqos_2)$, as illustrated Fig. \ref{fig: Structure}.
	\begin{enumerate}
		
		\item MNO-1's strong monopoly regime $(\costqos_1,\costqos_2)\in\Psi_1^\text{SM}=\{(\costqos_1,\costqos_2): \costqos_2 > \val_2^N(\costqos_1,\xi) \}$: 
			MNO-2 obtains a  zero market share.
			The equilibrium $\Threshold^*(\Mechanism)$ is
			\begin{equation}
			\Threshold^*(\Mechanism)=\big\{\valthr_1^\text{MP}(\mechanism_1),\costqos_2\big\},
			\end{equation}
			which is illustrated by the green circle in Fig. \ref{fig: R4}.
			MNO-2 gives up the competition, thus MNO-1 can decide its threshold user type \textit{without} considering the impact of MNO-2.
		
		\item MNO-1's weak monopoly regime $(\costqos_1,\costqos_2)\in\Psi_1^\text{WM}=\{(\costqos_1,\costqos_2): \val_2^L(\costqos_1,\xi) < \costqos_2 \le \val_2^N(\costqos_1,\xi) \}$:
			MNO-2 obtains a zero market share.
			The equilibrium $\Threshold^*(\Mechanism)$ is 
			\begin{equation}
				\Threshold^*=\big\{\costqos_2,\costqos_2\big\},
			\end{equation}
			which is illustrated by the green circle in Fig. \ref{fig: R5}.
			MNO-2 still tries to compete for market share, but in vain.
			However, MNO-1 has to decide its threshold user type \textit{considering} the impact of MNO-2.

		\item Coexistence regime $(\costqos_1,\costqos_2)\in\Psi^\text{C}=\{(\costqos_1,\costqos_2): \costqos_2\le\val_2^L(\costqos_1,\xi), \costqos_1 \le \val_1^L(\costqos_2,\xi) \}$.
			Both MNOs share the market. 
			The equilibrium $\Threshold^*(\Mechanism)$ solves
				\begin{equation}\label{Equ: Coexistence}
				\left\{
				\begin{aligned}
				&\textstyle 
				\valthr_1^*-\frac{ 1-H\left( \valeq(\valthr_1^*,\valthr_2^*) \right) }{\frac{1}{1-\xi(\Mechanism)}h\left( \valeq(\valthr_1^*,\valthr_2^*) \right)} =\psi_1 ,    \\
				&\textstyle
				\valthr_2^*-\frac{ H\left( \valeq(\valthr_1^*,\valthr_2^*) \right) - H\left(\valthr_2^*\right)}{\frac{\xi(\Mechanism)}{1-\xi(\Mechanism)} h\left( \valeq(\valthr_1^*,\valthr_2^*) \right) +h\left(\valthr_2^*\right)} =\psi_2,\\
				\end{aligned}
				\right.	
				\end{equation}
			which is illustrated by the green circle in Fig. \ref{fig: R6}.
			Both MNOs get strictly positive market share.

		\item MNO-2's weak monopoly regime $(\costqos_1,\costqos_2)\in\Psi_2^\text{WM}=\{(\costqos_1,\costqos_2): \val_1^L(\costqos_2,\xi) < \costqos_1 \le \val_1^N(\costqos_2,\xi) \}$.
			MNO-1 obtains a  zero market share.
			The equilibrium $\Threshold^*(\Mechanism)$ is 
			\begin{equation}
				\Threshold^*(\Mechanism)=\Big\{\costqos_1, \textstyle \frac{\costqos_1 + \left(\xi(\Mechanism)-1\right)\valmax}{\xi(\Mechanism)}  \Big\},
			\end{equation}
			which is illustrated by the green circle in Fig. \ref{fig: R7}.
			MNO-1 still tries to compete for market share, but in vain.
			However, MNO-2 has to decide its threshold user type \textit{considering} the impact of MNO-1.
		
		\item MNO-2's strong monopoly regime $(\costqos_1,\costqos_2)\in\Psi_2^\text{SM}=\{(\costqos_1,\costqos_2): \costqos_1 > \val_1^N(\costqos_2,\xi) \}$.
			MNO-2 obtains a  zero market.
			The equilibrium $\Threshold^*(\Mechanism)$ is
			\begin{equation}
				\Threshold^*(\Mechanism)=\big\{\costqos_1,\valthr_2^\text{MP}(\mechanism_2)\big\}.
			\end{equation}
			which is illustrated by the green circle in Fig. \ref{fig: R8}.
			MNO-1 gives up the competition, thus MNO-2 decides its threshold user type \textit{without} considering the impact of MNO-1.
	\end{enumerate}
\end{theorem}

Note that the threshold users type $\sigma_n$ corresponds to different combinations of the subscription fee and the per-unit fee (according to Definition 1).
Therefore, the uniqueness of the equilibrium in Game 1 does not necessarily imply  the unique pricing equilibrium in terms of the subscription fee and the per-unit fee.

So far we have characterized the equilibrium $\Threshold^*(\Mechanism)$ in Stage II under the data mechanism $\Mechanism$.
Next we move on to the MNOs' data mechanism selection in Stage I.

\section{MNOs' Data Mechanism Selection in Stage I\label{Section: MNO's Data Mechanism Selection in Stage I}}
In Stage I, the two MNOs will decide their data mechanisms $\Mechanism=\{\mechanism_1,\mechanism_2\}$, considering the responses from Stages II and III. 
Notice that we no longer needs Assumption \ref{Assumption: xi}, which is only used to facilitate the analysis in Stages II and III (without loss of generality). 
Moreover, we make Assumption \ref{Assumption: QoS} in this section.
\begin{assumption}\label{Assumption: QoS}
	MNO-1's average QoS is no worse than that of MNO-2, i.e., $\QoS_1\ge\QoS_2$.
\end{assumption}

Therefore, we will refer to MNO-1 and MNO-2 as the \textit{high-QoS MNO} and the \textit{low-QoS MNO}, respectively.
Note that \textit{Assumption \ref{Assumption: QoS} is not a technical assumption that limits our contributions}, since we can always switch the indices of the two MNOs if $\QoS_1<\QoS_2$.

Note that the two MNOs' costs ($c_1$ and $c_2$) will also play important roles in the equilibrium analysis. 
We will capture how the QoS values and costs interact with each other in Theorem \ref{Theorem: mechanism equilibrium} and Fig. \ref{fig: Mechanism structure} at the end of Section \ref{Section: MNO's Data Mechanism Selection in Stage I}. 

We model the two MNOs' data mechanism selection as the following game.
\begin{game}[Data Mechanism Selection in Stage I]\label{Game: Mechanism}
	The two MNOs' data mechanism selection in Stage I can be modeled as the following game:
	\begin{itemize} 
		\item Players: MNO-$n$ for both $n=1,2$.
		\item Strategies: Each MNO-$n$ decides its data mechanism $\mechanism_n$ from $\{\trad,\roll\}$.
		\item Preferences: Each MNO-$n$ obtains a profit $\profit_n(\Mechanism)$.
	\end{itemize}
\end{game}

In Game \ref{Game: Mechanism}, each MNO-$n$ decides its data mechanism $\mechanism_n$ to maximize its own profit $\profit_n(\Mechanism)$, considering the threshold equilibrium $\Threshold^*(\Mechanism)$ (derived in Theorem \ref{Theorem: Equilibrium Regime}) in Stage II.
To analyze Game \ref{Game: Mechanism},  we need to specify the two-by-two profit matrix as shown in Table \ref{table: MNOs' profits in data mechanisms selection}, where the strategy of each MNO is the data mechanism $\mechanism_n\in\{\trad,\roll\}$.
Since it is a two-by-two profit matrix, we can directly go through each of  the four possible outcomes to check whether it is an equilibrium.
Recall that our earlier work \cite{Zhiyuan2017pricing,Zhiyuan2018TMC} showed that a \textit{monopoly} MNO should always choose the rollover mechanism $\roll$ to maximize its profit.
For the \textit{duopoly} market considered here, however, we will show that this is not always the case, i.e.,  $\Mechanism=\{\roll,\roll\}$ is not always the equilibrium of Game \ref{Game: Mechanism}. 

\begin{table} 
	\setlength{\abovecaptionskip}{0pt}
	\setlength{\belowcaptionskip}{2pt}
	\renewcommand{\arraystretch}{1.2}		
	\caption{MNOs' profits in Game \ref{Game: Mechanism}.} 
	\label{table: MNOs' profits in data mechanisms selection}
	\centering
	\begin{tabular}{|c|c|c|c|c|c|c|c|c|c|}
		\hline
		\backslashbox{$\mechanism_1$}{$\mechanism_2$}		& $\trad$ 		& $\roll$ 	 					 \\
		\hline
		$\trad$		& $\profit_1(\trad,\trad)$, $\profit_2(\trad,\trad)$ 	& $\profit_1(\trad,\roll)$, $\profit_2(\trad,\roll)$	\\
		\hline
		$\roll$		& $\profit_1(\roll,\trad)$, $\profit_2(\roll,\trad)$	& $\profit_1(\roll,\roll)$, $\profit_2(\roll,\roll)$ \\
		\hline
	\end{tabular} 
\end{table}

Next we first introduce the potential market partitions at Game \ref{Game: Mechanism} equilibrium in Section \ref{Subsection: Market Partition at Game 2 Equilibrium}, then we present the equilibrium $\Mechanism^*=\{\mechanism_1^*,\mechanism_2^*\}$ in Section \ref{Subsection: Single-MNO-Surviving} and Section \ref{Subsection: Coexistence}.
Finally, we graphically illustrate the equilibrium $\Mechanism^*$ in Section \ref{Subsection: Equilibrium Illustration}.

\subsection{Market Partition at Game \ref{Game: Mechanism} Equilibrium} \label{Subsection: Market Partition at Game 2 Equilibrium}
We characterize the three possible market partitions at a Game \ref{Game: Mechanism} equilibrium in Theorem \ref{Theorem: Mechanism market partition}, depending on the two MNOs' average QoS (i.e., $\QoS_1$ and $\QoS_2$) and costs (i.e., $c_1$ and $c_2$).
The proof is given in Appendix \ref{Appendix: partition at equilibrium}.
\begin{theorem}\label{Theorem: Mechanism market partition}
	There exist two threshold costs $C^\textit{Single}_1(\QoS_1,\QoS_2,c_2)$ and $C^\textit{Single}_2(\QoS_1,\QoS_2,c_1)$, such that the market partitions at a Game \ref{Game: Mechanism} equilibrium has three possibilities.
	\begin{enumerate}
		\item MNO-1 Surviving: If MNO-1 experiences an extremely small cost, i.e., $c_1<C^\textit{Single}_1(\QoS_1,\QoS_2,c_2)$, then MNO-2 obtains zero market share.
		\item MNO-2 Surviving: If MNO-2 experiences an extremely small cost, i.e., $c_2<C^\textit{Single}_2(\QoS_1,\QoS_2,c_1)$, then MNO-1 obtains zero market share.
		\item Coexistence: If the two MNOs' costs are comparable, i.e., $c_1\ge C^\textit{Single}_1(\QoS_1,\QoS_2,c_2)$ and $c_2\ge C^\textit{Single}_2(\QoS_1,\QoS_2,c_1)$, then the two MNOs share the market.
	\end{enumerate}
	
	Moreover, the two threshold costs are given by
	\begin{equation}
	C^\textit{Single}_1(\QoS_1,\QoS_2,c_2)\eq\textstyle
	\QoS_1\left[ \frac{c_2}{\QoS_2} - \left(1-\frac{\QoS_2\usage_{\trad}}{\QoS_1\usage_{\roll}}\right)\cdot\frac{1-H\left(\frac{c_2}{\QoS_2}\right)}{h\left(\frac{c_2}{\QoS_2}\right)} \right],
	\end{equation}
	
	\begin{equation}
	\begin{aligned}
	&C^\textit{Single}_2(\QoS_1,\QoS_2,c_1)\eq\textstyle
	\max\bigg\{ 
	\textstyle \QoS_2\left[ \frac{c_1}{\QoS_1} - \left(1-\frac{\QoS_1\usage_{\trad}}{\QoS_2\usage_{\roll}}\right)\cdot\frac{1-H\left(\frac{c_1}{\QoS_1}\right)}{h\left(\frac{c_1}{\QoS_1}\right)} \right], \\
	&\qquad\qquad\textstyle c_1-(\QoS_1-\QoS_2)\valmax - \QoS_2\cdot\frac{ 1-H\left( \frac{c_1-(\QoS_1-\QoS_2)\valmax}{\QoS_2} \right)  }{ h\left( \frac{c_1-(\QoS_1-\QoS_2)\valmax}{\QoS_2} \right) }		\bigg\}.
	\end{aligned}
	\end{equation}	
\end{theorem}

\begin{figure*}
	\centering
	\begin{minipage}{0.26\textwidth}
		\setlength{\abovecaptionskip}{13pt}
		\centering
		\includegraphics[width=0.93\linewidth]{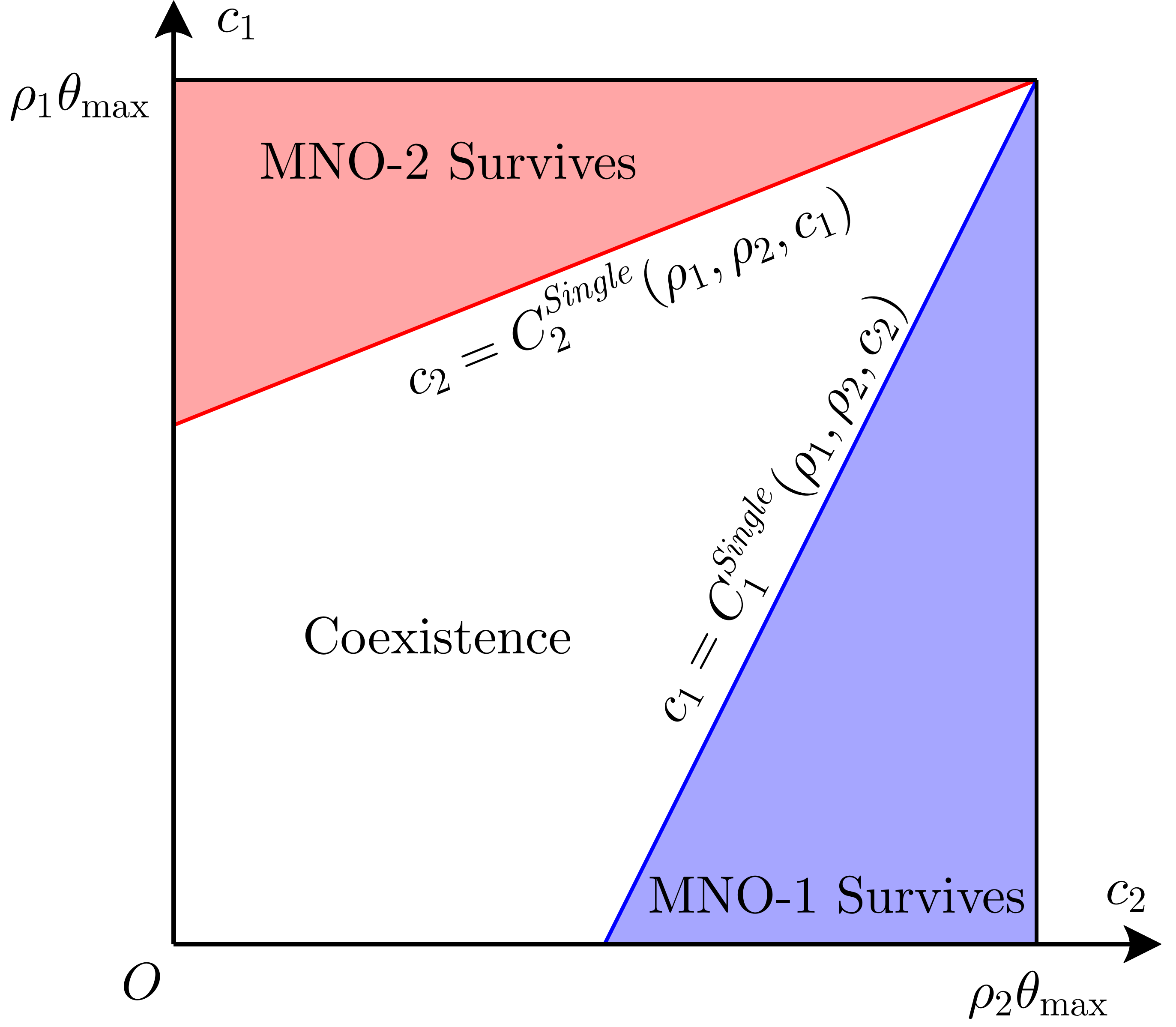}
		\caption{Illustration of Theorem \ref{Theorem: Mechanism market partition}.}
		\label{fig: MechanismEQ_0}
	\end{minipage} 
	\begin{minipage}{0.73\textwidth}
		\setlength{\abovecaptionskip}{-1pt}
		\setlength{\belowcaptionskip}{0pt}
		\subfigure[ $0<\QoS_2\le\hat{\QoS}$ ]{\label{fig: MechanismEQ_1}\includegraphics[height=0.29\linewidth]{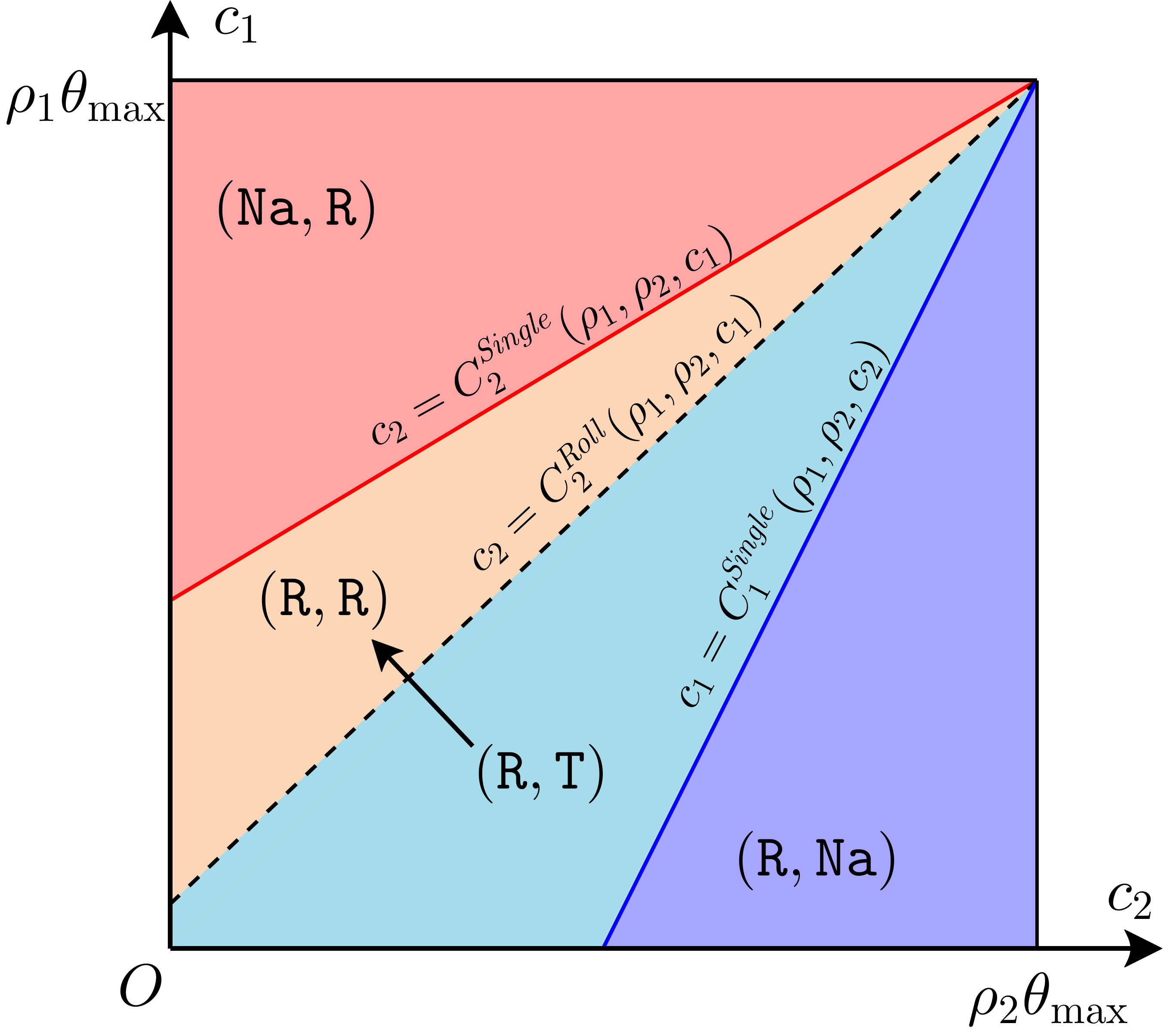}} 
		\subfigure[ $\hat{\QoS}<\QoS_2\le\tilde{\QoS}$ ]{\label{fig: MechanismEQ_2}\includegraphics[height=0.29\linewidth]{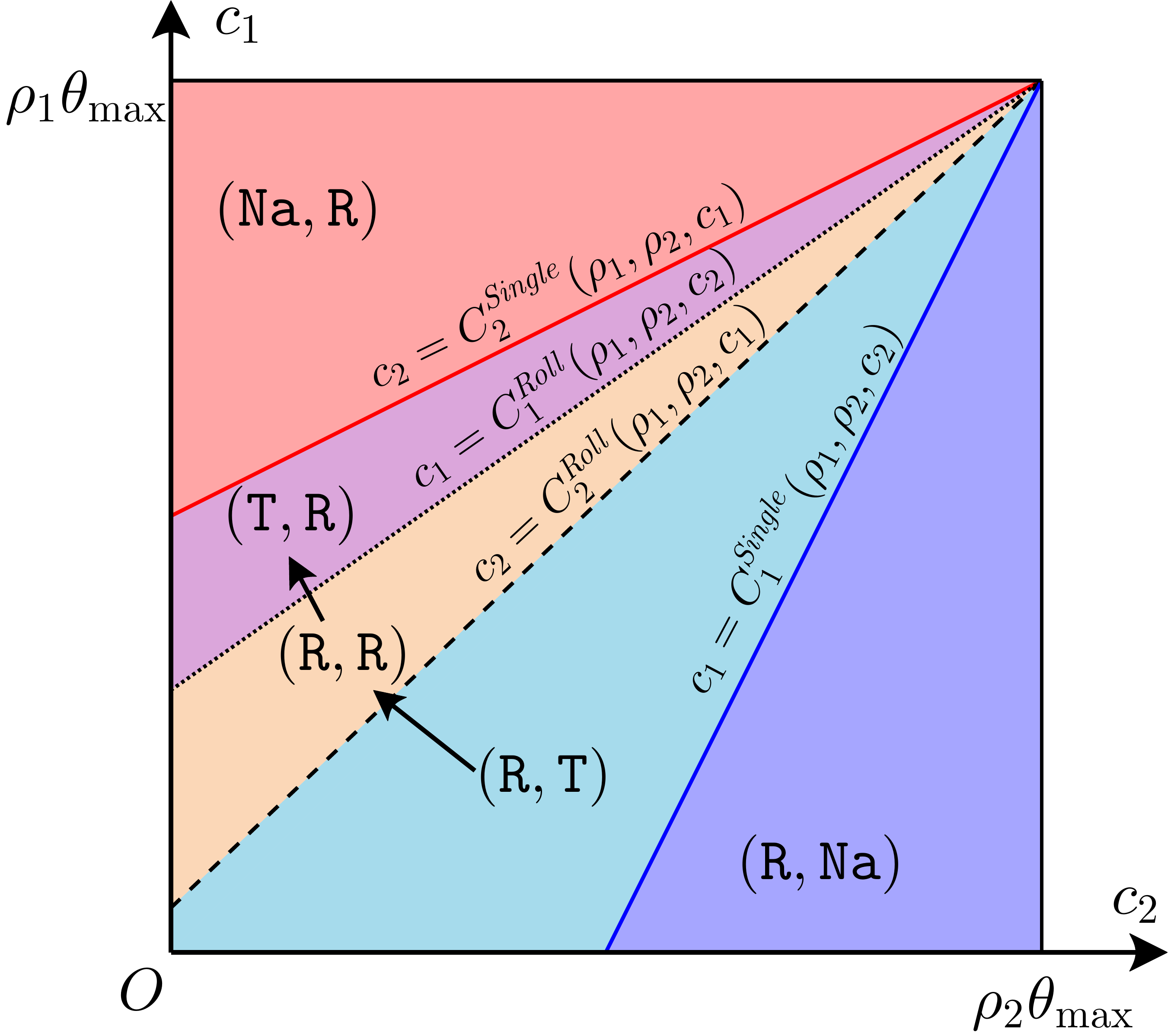}} 
		\subfigure[ $\tilde{\QoS}<\QoS_2\le\QoS_1$ ]{\label{fig: MechanismEQ_3}\includegraphics[height=0.29\linewidth]{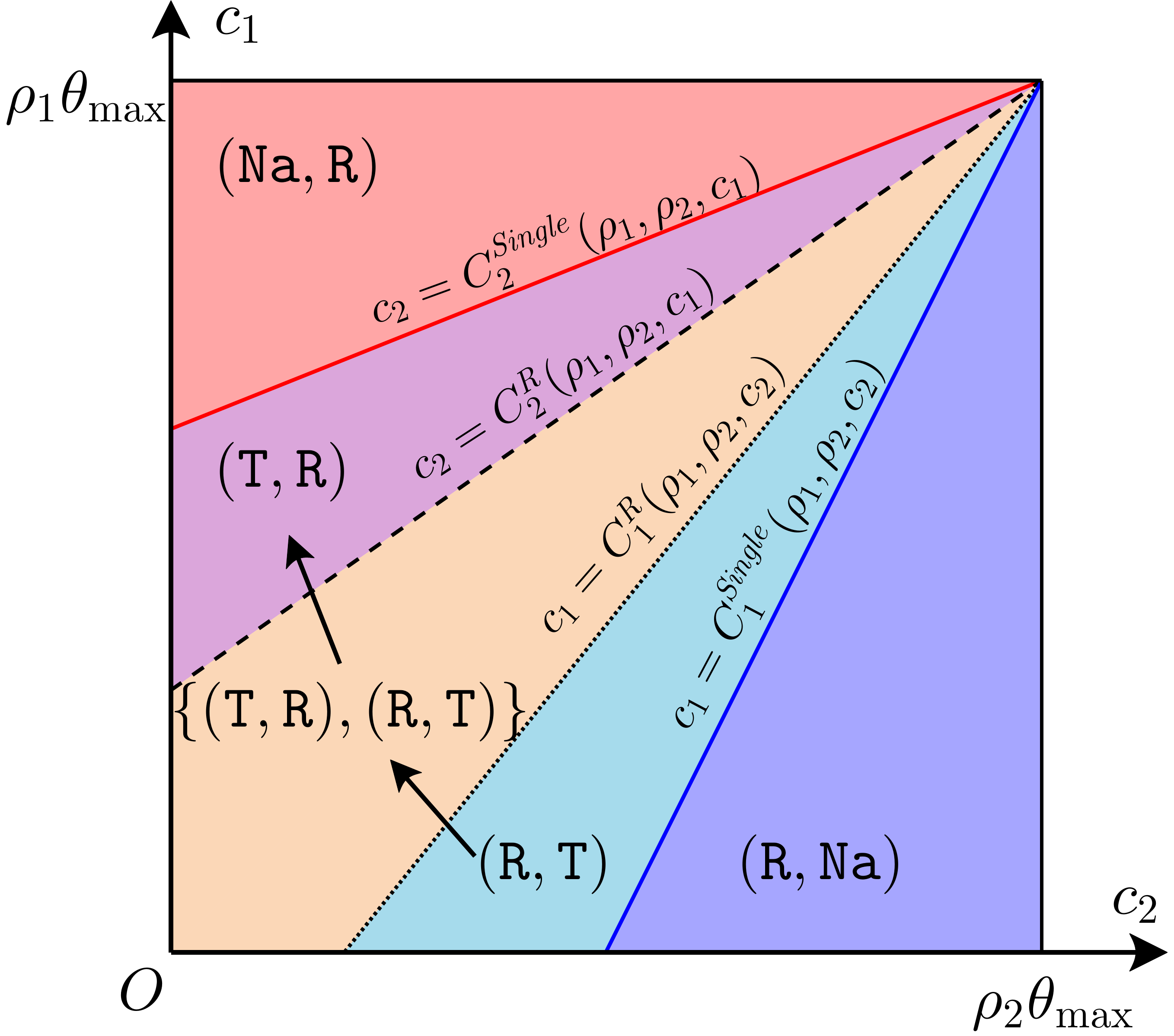}}
		\caption{Illustration of Theorem \ref{Theorem: mechanism equilibrium}.}
		\label{fig: Mechanism structure}
	\end{minipage} 
\end{figure*}

We illustrate Theorem \ref{Theorem: Mechanism market partition} in Fig. \ref{fig: MechanismEQ_0} under a uniform distribution of $\val$.
Specifically, the horizontal and vertical axis corresponds to MNO-2's cost $c_2$ and MNO-1's cost $c_1$, respectively.
The two shaded (i.e., red and blue) areas represent that only one  MNO obtains a positive market share.
The white area represents that the two MNOs share the market.
In Section \ref{Subsection: Single-MNO-Surviving} and Section \ref{Subsection: Coexistence}, we will present the detailed equilibrium $\Mechanism^*$ for each area in Fig. \ref{fig: MechanismEQ_0}.
Before that, let us introduce  $\na$ in Definition \ref{Definition: Na} for notation simplicity.
\begin{definition} \label{Definition: Na}
	The notation $\na$ is used to describe the equilibriums in the following cases:
	\begin{itemize}
		\item The equilibrium  $\Mechanism^*=(\mechanism_1^*,\na)$ represents that both $(\mechanism_1^*,\trad)$ and $(\mechanism_1^*,\roll)$ are the equilibriums of Game \ref{Game: Mechanism}.
		\item The equilibrium  $\Mechanism^*=(\na,\mechanism_2^*)$ represents that both $(\trad,\mechanism_2^*)$ and $(\roll,\mechanism_2^*)$ are the equilibriums of Game \ref{Game: Mechanism}.
	\end{itemize}
	
\end{definition}

\subsection{Single-MNO-Surviving}\label{Subsection: Single-MNO-Surviving}
We introduce the data mechanism equilibrium $\Mechanism^*$ of the single-MNO-surviving case in Lemma \ref{Lemma: Mechanism equilibrium MNO-1} and Lemma \ref{Lemma: Mechanism equilibrium MNO-2}.
The proofs are given in Appendix \ref{Appendix: partition at equilibrium}.
\begin{lemma}\label{Lemma: Mechanism equilibrium MNO-1}
	If $c_1<C^\textit{Single}_1(\QoS_1,\QoS_2,c_2)$, then the high-QoS MNO-1 obtains a positive market share, the low-QoS MNO-2 obtains a zero market share independent of his choice and the data mechanism equilibrium  is
		\begin{equation}\label{Equ: Mechanim Equilibrium R Na}
				\Mechanism^{*}=(\roll,\na).
		\end{equation}
		Furthermore, we have 
		\begin{equation}\label{Equ: Profit Effect R Na}
			\profit_1(\roll,\trad)=\profit_1(\roll,\roll).
		\end{equation}
\end{lemma}

\begin{lemma}\label{Lemma: Mechanism equilibrium MNO-2}
	If $c_2<C^\textit{Single}_2(\QoS_1,\QoS_2,c_1)$, then the low-QoS MNO-2 obtains a positive market share, the high-QoS MNO-1 obtains a zero market share independent of his choice and the data mechanism equilibrium is
		\begin{equation}\label{Equ: Mechanim Equilibrium Na R}
		\Mechanism^{*}=(\na,\roll).
		\end{equation}
	Furthermore, we have 
		\begin{equation}\label{Equ: Profit Effect Na R}
		\profit_2(\roll,\roll)\le\profit_2(\trad,\roll).
		\end{equation}
\end{lemma}

Lemma \ref{Lemma: Mechanism equilibrium MNO-1} and Lemma \ref{Lemma: Mechanism equilibrium MNO-2} reveal the impact of MNO-1's QoS advantage (i.e., $\QoS_1\ge\QoS_2$):
\begin{itemize}
	\item When the low-QoS MNO-2 obtains a zero market share (i.e., $c_1<C_1^\textit{Single}(\QoS_1,\QoS_2,c_2)$) as in Lemma \ref{Lemma: Mechanism equilibrium MNO-1}, its choice between $\trad$ and $\roll$ has no effect on the high-QoS MNO-1, i.e., $\profit_1(\roll,\trad)=\profit_1(\roll,\roll)$.
	
	\item When the high-QoS MNO-1 obtains a zero market share (i.e., $c_2<C^\textit{Single}_2(\QoS_1,\QoS_2,c_1)$) as in Lemma \ref{Lemma: Mechanism equilibrium MNO-2}, it can reduce the low-QoS MNO-2's profit by choosing the rollover  mechanism $\roll$, i.e., $\profit_2(\roll,\roll)\le\profit_2(\trad,\roll)$.
\end{itemize}


\subsection{Coexistence}\label{Subsection: Coexistence}
Now we consider the case where both the MNOs obtain positive market shares at Game \ref{Game: Mechanism} equilibrium, i.e., $c_1\ge C^\textit{Single}_1(\QoS_1,\QoS_2,c_2)$ and $c_2\ge C^\textit{Single}_2(\QoS_1,\QoS_2,c_1)$.

To facilitate our later discussion, we first introduce the QoS-Flip phenomenon in Section \ref{Subsubsection: QoS-Flip Phenomenon}.
\subsubsection{QoS-Flip Phenomenon}\label{Subsubsection: QoS-Flip Phenomenon}
A counter-intuitive result that we will elaborate is  that the high-QoS MNO-1 will not always attract the high valuation users.
According to Theorem \ref{Theorem: Market Partition}, MNO-$1$ attracts the high valuation users if $\QoS_1\usage_{\mechanism_1}>\QoS_2\usage_{\mechanism_2}$, otherwise, MNO-$2$  attracts the high valuation users if $\QoS_1\usage_{\mechanism_1}<\QoS_2\usage_{\mechanism_2}$.
Recall that $\usage_{\mechanism_n}$ defined in (\ref{Equ: alpha}) depends on the data mechanism $\mechanism_n\in\{\trad,\roll\}$.
The inequality (\ref{Equ: alpha flexibility}) indicates 
\begin{itemize}
	\item The low-QoS MNO-2 can attract the high valuation users under the data mechanism $\Mechanism=\{\trad,\roll\}$ if $\QoS_2>\QoS_1\usage_{\trad}/\usage_{\roll}$.
	\item The high-QoS MNO-1 can attract the high valuation users under the data mechanism $\Mechanism\in\{(\trad,\trad),(\roll,\trad),(\roll,\roll)\}$.
\end{itemize}

Therefore, the rollover mechanism may reverse the MNO-2's QoS disadvantage.
We refer to this phenomenon as \textit{QoS-flip}, defined as follows:
\begin{definition}[QoS-flip]
	The QoS-flip happens if the low-QoS MNO-2 attracts the high valuation users, i.e., $\QoS_2>\QoS_1\usage_{\trad}/\usage_{\roll}$, under the data mechanism $\Mechanism=\{\trad,\roll\}$.
\end{definition}

In the following analysis for the equilibrium of the coexistence case, we will explain  when QoS-flip happens.

\subsubsection{Equilibrium of Coexistence Case}

To present the equilibrium in the coexistence case clearly, we need to use two cost thresholds (i.e., ${C}^\textit{Roll}_1(\QoS_1,\QoS_2,c_2)$ and ${C}^\textit{Roll}_2(\QoS_1,\QoS_2,c_1)$) and two QoS thresholds (i.e., $\tilde{\QoS}$ and $\hat{\QoS}$), which depend on both MNOs' QoS (i.e., $\QoS_1$ and $\QoS_2$) and costs (i.e., $c_1$ and $c_2$).
Due to the complexity of the MNOs' two-dimensional heterogeneity in QoS $\QoS_n$ and cost $c_n$, as well as the users' heterogeneity in the data valuation $\val$, there is no closed-form expression for ${C}^\textit{Roll}_1$, ${C}^\textit{Roll}_2$, $\tilde{\QoS}$, and $\hat{\QoS}$.
Nevertheless, we will explain how to compute them numerically in Appendix \ref{Appendix: numerical compute}.

Theorem \ref{Theorem: mechanism equilibrium} presents the data mechanism equilibrium of the coexistence case.
\begin{theorem}\label{Theorem: mechanism equilibrium}
	Under the coexistence case of Game \ref{Game: Mechanism}, there exist two QoS thresholds $\tilde{\QoS}>\hat{\QoS}$ such that the equilibrium $\Mechanism^{*}$ has three different possibilities: 
	\begin{enumerate}
		\item When MNO-1 has a \textbf{large} QoS advantage over MNO-2, i.e., $0<\QoS_2\le\hat{\QoS}$, the equilibrium $\Mechanism^{*}$ (as shown in Fig. \ref{fig: MechanismEQ_1}) is
			\begin{itemize}
				\setlength{\itemsep}{3pt}
				\item $(\roll,\trad)$ if MNO-2 experiences a  large cost, i.e., $ c_2>{C}^\textit{Roll}_2(\QoS_1,\QoS_2,c_1) $; 
				\item $(\roll,\roll)$ if MNO-2 experiences a  small cost, i.e., $ c_2\le {C}^\textit{Roll}_2(\QoS_1,\QoS_2,c_1) $.
			\end{itemize}
		\item When MNO-1 has a \textbf{small} QoS advantage over MNO-2, i.e., $\hat{\QoS}<\QoS_2\le\tilde{\QoS}$, then the equilibrium $\Mechanism^{*}$ (as shown in Fig. \ref{fig: MechanismEQ_2}) is
			\begin{itemize}
				\setlength{\itemsep}{3pt}
				\item $(\roll,\trad)$ if MNO-2 experiences a large cost, i.e., $ c_2>{C}^\textit{Roll}_2(\QoS_1,\QoS_2,c_1) $;
				\item $(\roll,\roll)$ if both MNO-1 and MNO-2 experience small and comparable costs, i.e., $ c_1\le {C}^\textit{Roll}_1(\QoS_1,\QoS_2,c_2) $ and $ c_2\le {C}^\textit{Roll}_2(\QoS_1,\QoS_2,c_1) $; 
				\item $(\trad,\roll)$ if MNO-1 experiences a large cost, i.e., $ c_1 > {C}^\textit{Roll}_1(\QoS_1,\QoS_2,c_2) $.
			\end{itemize}
		\item When MNO-1 has a \textbf{negligible} QoS advantage over MNO-2, i.e., $\tilde{\QoS}<\QoS_2\le\QoS_1$, then the equilibrium $\Mechanism^{*}$ (as shown in Fig. \ref{fig: MechanismEQ_3}) is
			\begin{itemize}
				\setlength{\itemsep}{3pt}
				\item $(\roll,\trad)$ if MNO-1 experiences a small cost, i.e., $ c_1 < {C}^\textit{Roll}_1(\QoS_1,\QoS_2,c_2) $;
				\item $\{(\roll,\trad),(\trad,\roll)\}$ if both MNO-1 and MNO-2 experience large and comparable costs, i.e., $ c_1\ge {C}^\textit{Roll}_1(\QoS_1,\QoS_2,c_2) $ and $ c_2\ge {C}^\textit{Roll}_2(\QoS_1,\QoS_2,c_1) $; 
				\item $(\trad,\roll)$ if MNO-2 experiences a small cost, i.e., $ c_2<{C}^\textit{Roll}_2(\QoS_1,\QoS_2,c_1) $.
			\end{itemize}
	\end{enumerate}		
\end{theorem}

So far we have completely characterized the data mechanism equilibrium $\Mechanism^{*}$ of Game \ref{Game: Mechanism} in Theorems \ref{Theorem: Mechanism market partition} and \ref{Theorem: mechanism equilibrium}.
Next we will use the uniformly distributed market to illustrate the equilibrium $\Mechanism^*$.

\subsection{Equilibrium Illustration}\label{Subsection: Equilibrium Illustration}
Fig. \ref{fig: Mechanism structure} visualizes the equilibrium $\Mechanism^*$ under a uniform distribution of $\val$. 
Specifically, the three sub-figures in Fig. \ref{fig: Mechanism structure} corresponds to different levels of MNO-1's QoS advantage, i.e., large for Fig. \ref{fig: MechanismEQ_1}, small for Fig. \ref{fig: MechanismEQ_2}, and negligible for Fig. \ref{fig: MechanismEQ_3}.
In each sub-figure, the horizontal and vertical axises correspond to MNO-2's cost $c_2$ and MNO-1's cost $c_1$, respectively.
We label the corresponding equilibrium $\Mechanism^*$ on the $(c_2,c_1)$ plane in each sub-figure.

\textbf{MNO-1 Surviving:} 
In each sub-figure of Fig. \ref{fig: Mechanism structure}, the blue region marked by $(\roll,\na)$ represents that MNO-1 becomes a monopoly in the market and leaves MNO-2 a zero market share no matter which data mechanism MNO-2 adopts.
This is because that MNO-1 has enough cost advantage, i.e., $c_1<C^\textit{Single}_1(\QoS_1,\QoS_2,c_2)$.
In this case, there are two equilibriums, i.e., $(\roll,\trad)$ and $(\roll,\roll)$.

\textbf{MNO-2 Surviving:} 
In each sub-figure of Fig. \ref{fig: Mechanism structure}, the red region marked by $(\na,\roll)$ represents that MNO-2 becomes a monopoly in the market and leaves MNO-1 zero market share no matter which data mechanism MNO-1 adopts, since MNO-2 has a large cost advantage, i.e., $c_2<C^\textit{Single}_2(\QoS_1,\QoS_2,c_1)$.
Similarly, there are two equilibriums, i.e., $(\trad,\roll)$ and $(\roll,\roll)$.

\textbf{Coexistence:}
When two MNOs' costs are comparable, i.e., $c_1\ge C^\textit{Single}_1(\QoS_1,\QoS_2,c_2)$ and $c_2\ge C^\textit{Single}_2(\QoS_1,\QoS_2,c_1)$, they will share the market.
In this case, the equilibrium structure also depends on their QoS difference:
\begin{itemize}
	\item Fig. \ref{fig: MechanismEQ_1}: MNO-1 has a \textbf{large} QoS advantage over MNO-2, i.e., $0<\QoS_2\le\hat{\QoS}$.
			According to the arrow in Fig. \ref{fig: MechanismEQ_1}, the equilibrium of Game \ref{Game: Mechanism} gradually changes from $(\roll,\trad)$ to $(\roll,\roll)$ as MNO-1's cost $c_1$ increases or MNO-2's cost $c_2$ decreases.
			That is, the large QoS advantage enables MNO-1 to adopt rollover mechanism $\roll$ all the time, but the low-QoS MNO-2 has the opportunity of upgrading to the rollover mechanism $\roll$ as its cost advantage increases which compensates its QoS disadvantage.
	\item Fig. \ref{fig: MechanismEQ_2}: MNO-1 has a \textbf{small} QoS advantage over MNO-2, i.e., $\hat{\QoS}<\QoS_2\le\tilde{\QoS}$.
			Based on the arrows in Fig. \ref{fig: MechanismEQ_2}, as MNO-1's cost $c_1$ increases or MNO-2's cost $c_2$ decreases, the corresponding equilibrium changes according to the order: $(\roll,\trad)$ $\rightarrow$ $(\roll,\roll)$ $\rightarrow$ $(\trad,\roll)$.
			Different from the case of Fig. \ref{fig: MechanismEQ_1}, in Fig. \ref{fig: MechanismEQ_2} the QoS-flip phenomenon happens at the equilibrium $(\trad,\roll)$, since MNO-1's QoS advantage is not large enough and experiences a large cost. 
	\item Fig. \ref{fig: MechanismEQ_3}: MNO-1 has a \textbf{negligible} QoS advantage over MNO-2, i.e., $\tilde{\QoS}<\QoS_2\le\QoS_1$.	 
			As MNO-1's cost $c_1$ increases or MNO-2's cost $c_2$ decreases, i.e., the arrows in Fig. \ref{fig: MechanismEQ_3}, the corresponding equilibrium changes according to the order: $(\roll,\trad)$ $\rightarrow$ $\{(\roll,\trad),(\trad,\roll)\}$ $\rightarrow$ $(\trad,\roll)$.
			We note that the symmetric equilibriums $\{(\roll,\trad),(\trad,\roll)\}$ arise between $(\roll,\trad)$ and $(\trad,\roll)$, instead of the $(\roll,\roll)$ in Fig. \ref{fig: MechanismEQ_2}.
			This is because that the negligible QoS advantage makes the two MNOs more homogeneous and the head-to-head market competition will reduce the profits of both MNOs.
			Similar to the anti-coordination game (e.g., the hawk-dove game), no matter who adopts the rollover mechanism $\roll$, it is a best choice for the competitor to choose the traditional mechanism $\trad$.
\end{itemize}


So far, we have finished the analysis for the three-stage competition model, and revealed that $\Mechanism=\{\roll,\roll\}$ is \textbf{not} always the equilibrium in a competitive market, which is consistent with our practical observations mentioned in Section  \ref{Subsection: Background and Motivation}.

\section{Numerical Results\label{Section: Numerical Results}}
Next we will simulate the MNOs' data mechanism equilibriums based on empirical data in Section \ref{Subsubsection: Data Mechanism Selection in Stage I} and evaluate the effects of rollover mechanism and market competition in Section \ref{Subsection: Impact of Rollover Mechanism}.
Before that, let us first introduce our simulation setting for mobile users and MNOs.

For mobile users, we adopt the empirical results from the previous literatures to model the monthly data demand, data valuation, and network substitutability. 
Specifically, we follow the data analysis results in \cite{lambrecht2007does} and assume that users' monthly data demand follows a truncated log-normal distribution with mean $\dmean=10^3$ on the interval $[0,10^4]$, i.e., the mean value is $\dmean=1$GB and the maximal usage is $D=10$GB.
Moreover, we adopt the empirical study in \cite{Zhiyuan2018TMC} by assuming that $\val$ follows a Gamma distribution with the shape parameter $k=4.5$ and the scale parameter $r=0.11$, and assuming  $\cut=0.8$.

For the MNOs, we assume  that both MNOs offer a 1GB data plan, i.e., $\dcap_1=\dcap_2=1$GB, and MNO-1 provides a better QoS than MNO-2.
For simplicity, we normalize MNO-1's QoS, i.e., $\QoS_1=1$, and consider three cases where MNO-2 provides different QoS, i.e., $\QoS_2\in\{0.91,0.95,0.99\}$.
Mathematically, $\QoS_2=0.91$ corresponds to the large QoS advantage case (mentioned in Theorem \ref{Theorem: mechanism equilibrium}).
The choices of $\QoS_2=0.95$ and $\QoS_2=0.99$ correspond to the small QoS advantage case and the negligible QoS advantage case, respectively.
Due to space limit, we will only show the results of  $\QoS_2=0.91$ in the main paper and refer interested readers to Appendix \ref{Appendix: Data Mechanism Equilibrium} for the results of $\QoS_2=\{0.95,0.99\}$.

\subsection{Data Mechanism Equilibrium \label{Subsubsection: Data Mechanism Selection in Stage I}}
We illustrate the data mechanism equilibrium $\Mechanism^*$ in Game \ref{Game: Mechanism} by varying MNO-1's cost $c_1$, given MNO-2's cost $c_2=40$ RMB/GB.
\begin{figure}  
	\centering
	\setlength{\abovecaptionskip}{0pt}
	\subfigure[Data mechanism equilibrium $\Mechanism^{*}$.]{\label{fig: EQ_mechanism_3}{\includegraphics[width=1\linewidth]{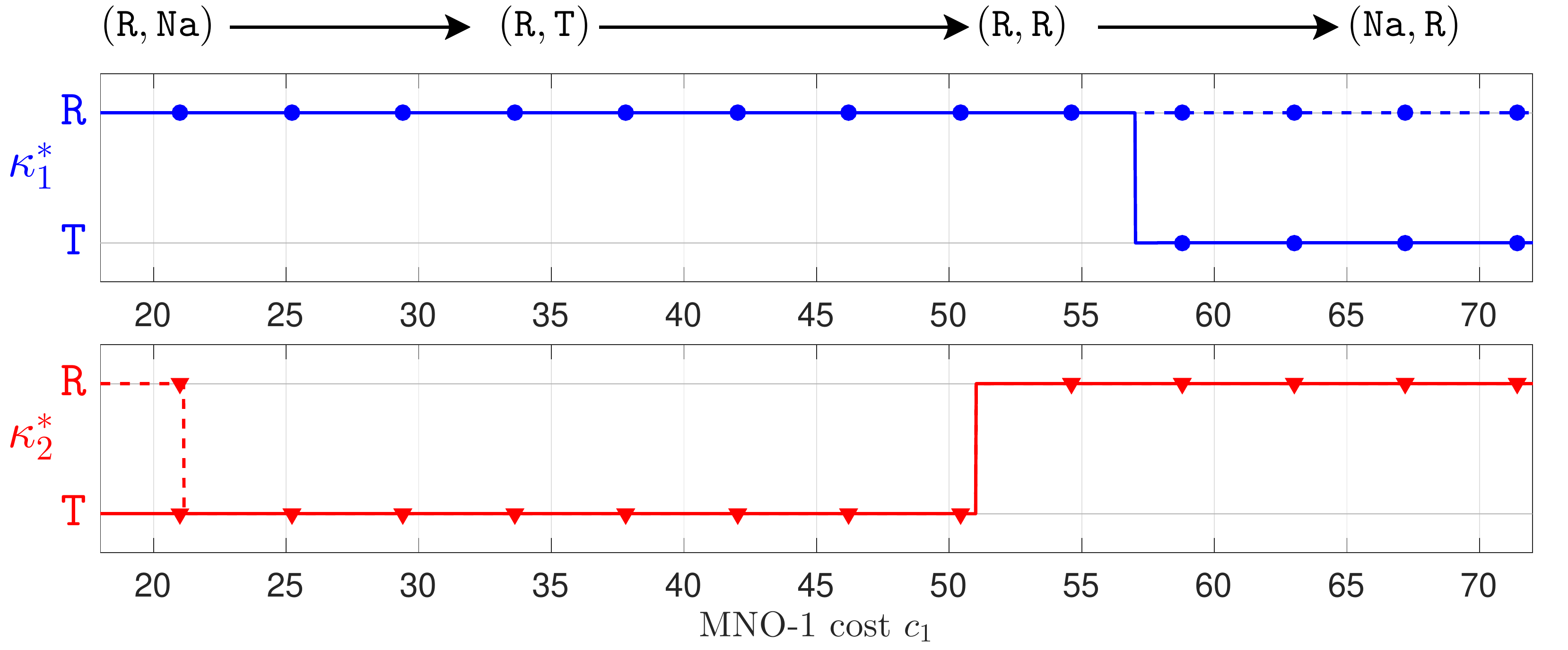}}}
	\subfigure[MNOs' profits under $\Mechanism^{*}$.]{\label{fig: EQ_mechanism_3_Profit}{\includegraphics[width=1\linewidth]{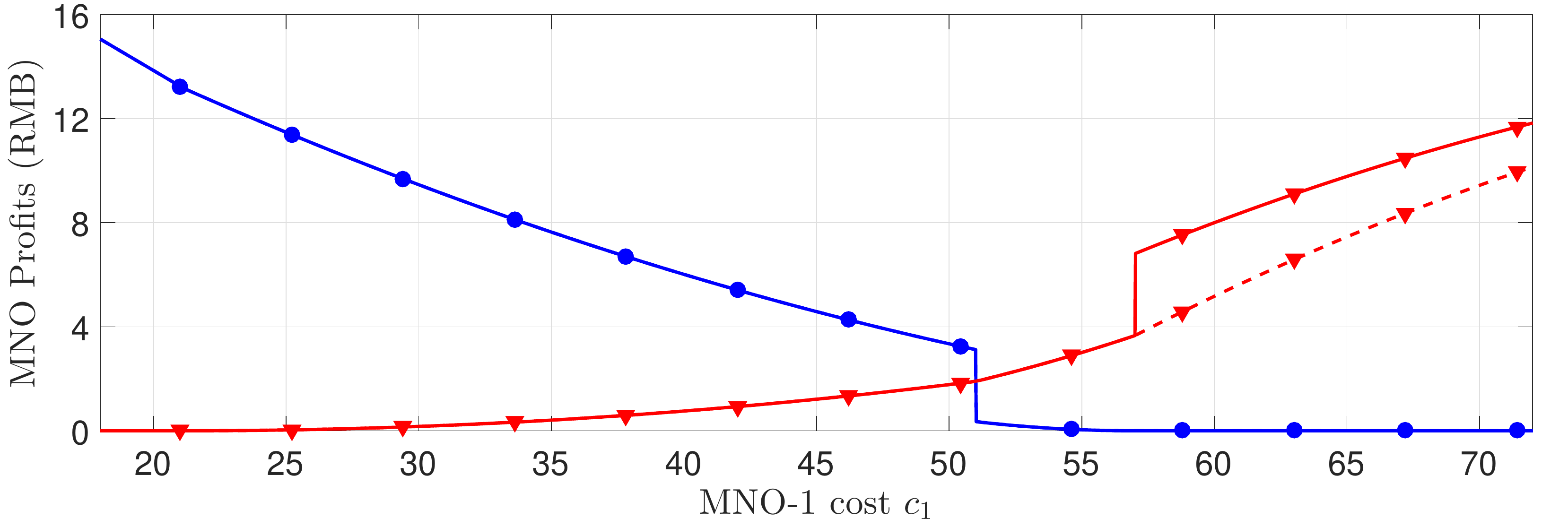}}} 
	\caption{Game \ref{Game: Mechanism} equilibrium when $\QoS_2=0.91$.}
	\label{fig: EQ rho=0.91}
\end{figure}

In Fig. \ref{fig: EQ_mechanism_3}, we plot the data mechanism equilibrium $\Mechanism^{*}$ versus MNO-1's cost $c_1$, and label $\Mechanism^{*}$ on the top of the figure.
Specifically, the blue circle lines represent MNO-1's data mechanism $\mechanism_1^{*}$, the red triangle lines represents MNO-2's data mechanism $\mechanism_2^{*}$.
Moreover, we use two line styles (i.e., solid and dash) when there are two equilibriums.

Fig. \ref{fig: EQ_mechanism_3_Profit} shows the two MNOs' profits under the equilibrium $\Mechanism^{*}$.
Specifically, the blue circle curves and red triangle curves represent MNO-1's profit $\profit_1(\Mechanism^{*})$ and MNO-2's profit $\profit_2(\Mechanism^{*})$, respectively.
The solid (dash, respectively) curves  in Fig. \ref{fig: EQ_mechanism_3_Profit} correspond to the equilibriums plotted by the solid (dash, respectively) lines in Fig. \ref{fig: EQ_mechanism_3}.
Overall, as MNO-1's cost $c_1$ increases, its profit (i.e., the blue circle curve) eventually decreases to zero, while MNO-2's profit (i.e., the red circle curves) will increase.
Moreover, MNO-1 experiences a significant profit drop when the equilibrium changes from $(\roll,\trad)$ to $(\roll,\roll)$ at $c_1=51$ RMB/GB.
In addition, the following observations validate Lemma \ref{Lemma: Mechanism equilibrium MNO-1} and Lemma \ref{Lemma: Mechanism equilibrium MNO-2}:
\begin{itemize}
	\item When MNO-1 obtains zero market share ($c_1>53$ RMB/GB), the two equilibriums $(\trad,\roll)$ and $(\roll,\roll)$ lead to different profits for MNO-2, i.e., $\profit_2(\roll,\roll)\le\profit_2(\trad,\roll)$, which shows that the high-QoS MNO-1 may reduce MNO-2's profit by choosing the rollover mechanism $\roll$, even though it obtains zero market share. 
	\item When MNO-2 obtains zero market share ($c_1<23$ RMB/GB), the blue solid curve and the blue dash curve overlap, which means that the two equilibriums $(\roll,\trad)$ and $(\roll,\roll)$ lead to the same profit for MNO-1, i.e., $\profit_1(\roll,\trad) = \profit_1(\roll,\roll)$.
\end{itemize}

\begin{figure}  
	\centering
		\setlength{\abovecaptionskip}{0pt}
		\setlength{\belowcaptionskip}{0pt}
	\subfigure[MNO-1]{\label{fig: ProfitCompare_3_MNO1}{\includegraphics[width=0.49\linewidth]{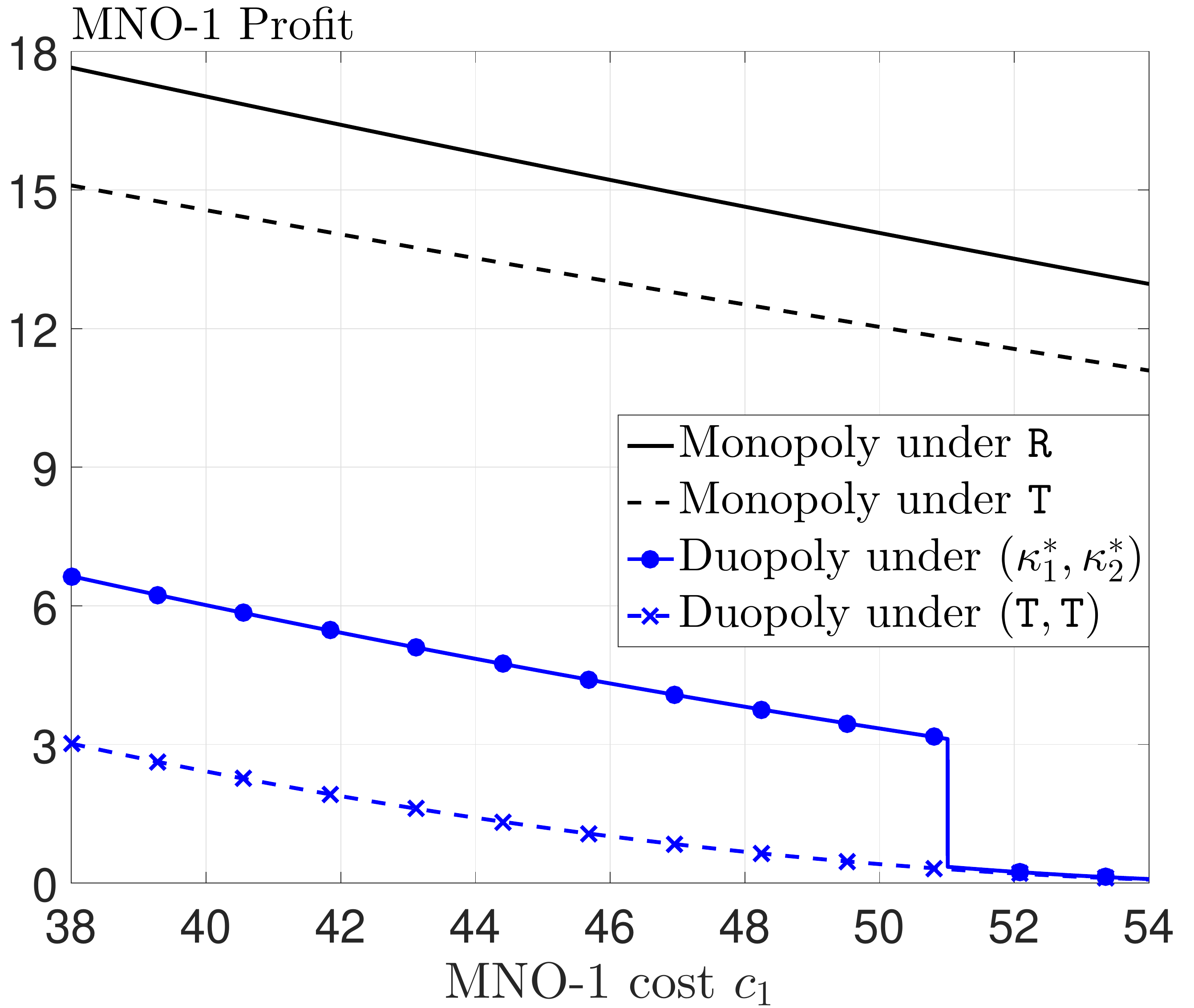}}}
	\subfigure[MNO-2]{\label{fig: ProfitCompare_3_MNO2}{\includegraphics[width=0.49\linewidth]{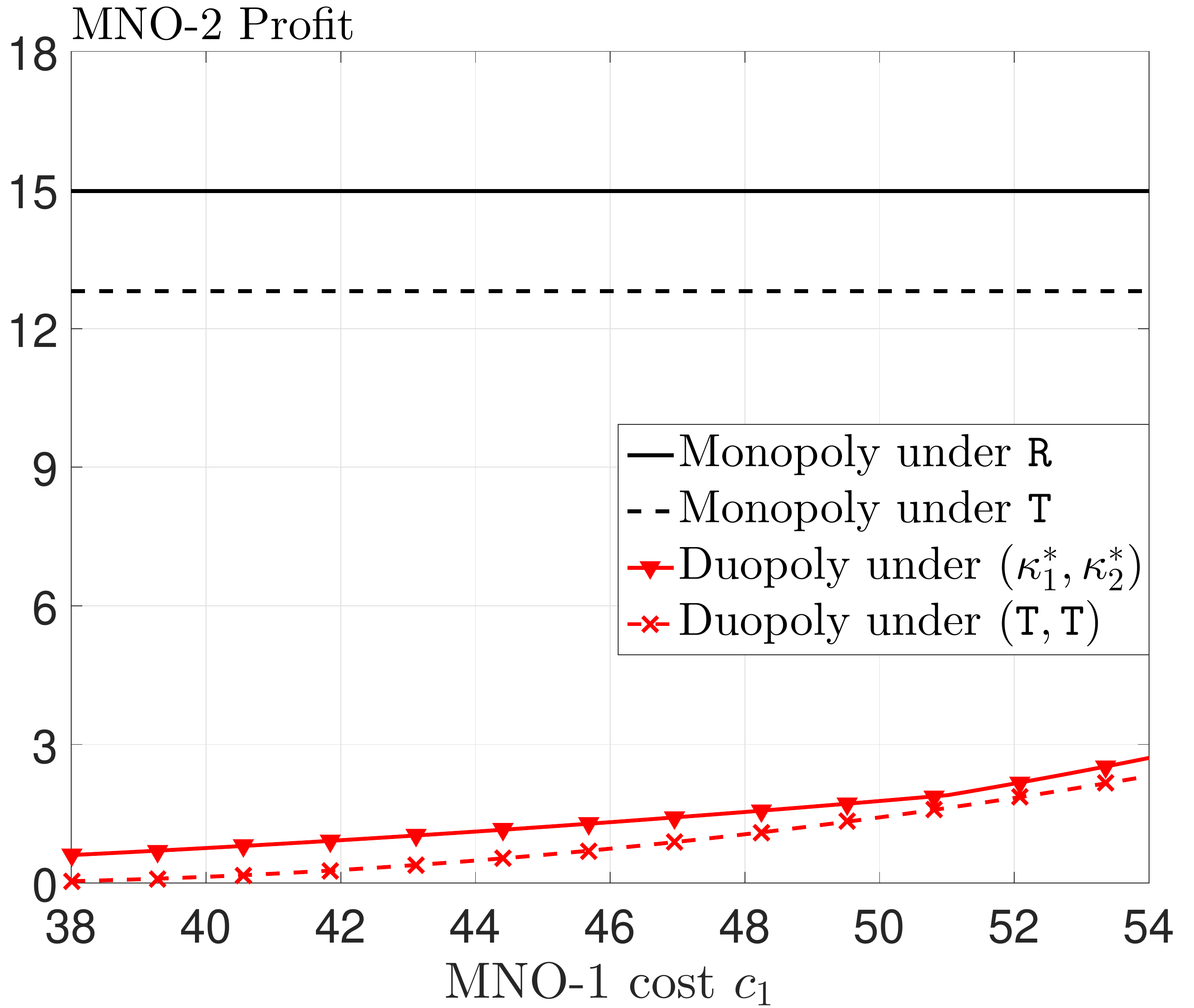}}} 
	\caption{Impact of rollover mechanism on duopoly market.}
	\label{fig: Impact of rollover mechanism}
\end{figure}

\subsection{Impact of Rollover Mechanism and Competition}\label{Subsection: Impact of Rollover Mechanism}
Now we evaluate the impact of the rollover mechanism and the market competition on the MNOs' profits.

Fig. \ref{fig: ProfitCompare_3_MNO1} plots MNO-1's profit versus its cost $c_1$ in four scenarios.
Specifically, the two black curves without markers correspond to MNO-1's monopoly market under data mechanism $\trad$ and $\roll$.
The blue circle curve represents MNO-1's profit in the duopoly market under the equilibrium $\Mechanism^*$ shown in Fig. \ref{fig: EQ_mechanism_3}.
Essentially, the blue circle curve is the same as that in Fig. \ref{fig: EQ_mechanism_3_Profit}.
The blue cross curve corresponds to MNO-1's profit in the duopoly market under fixed data mechanisms $(\trad,\trad)$.
In this case, the MNOs only compete on price (but not on the data mechanism choice).
By comparing the black curves with the blue curves, we find that the market competition significantly reduces MNO-1's profit.
By comparing the two blue curves with markers, we find that the rollover mechanism significantly  increases MNO-1's profit $274\%$ (on average) in the market competition.

Fig. \ref{fig: ProfitCompare_3_MNO2} plots MNO-2's profit versus the cost $c_1$ in similar scenarios.
Specifically, the two black curves without markers correspond to MNO-2's monopoly market.
The two red curves  correspond to the cases in the duopoly market. 
The red triangle curve is the same as that in Fig. \ref{fig: EQ_mechanism_3_Profit}.
The red cross curve corresponds to MNO-2's profit in the duopoly market under fixed data mechanism $(\trad,\trad)$.
By comparing the black curves with the red curves, we find that the market competition significantly reduces MNO-2's profit.
The profit decrement of MNO-2 is larger than that of MNO-1, since MNO-1 has the QoS advantage and attracts high-valuation users.
By comparing the two red curves with markers, we find that the rollover mechanism increases MNO-1's profit $188\%$ (on average) in the market competition.

The above findings also hold for the small QoS advantage and negligible QoS advantage cases.
Please refer to Appendix \ref{Appendix: Data Mechanism Equilibrium} for more details.

\section{Conclusions and Future Work\label{Section: Conclusions and Future Work}}
In this paper, we studied the duopoly competition in the telecommunication market in terms of the MNOs' rollover mechanism adoption and the pricing decisions.
Different from the monopoly market, where the MNO always increase its profit by adopting the rollover mechanism, the data mechanism equilibrium in the duopoly market is much more complicated. 
Roughly speaking, the high-QoS MNO would gradually abandon the rollover mechanism as its QoS advantage diminishes (due to its increasing cost or the competitor's decreasing cost).

In the future, we will extend the results of this paper in the following aspects.
First, we would like to collaborate with MNOs and extend the current analysis by investigating real world data.
Second, we will consider a more realistic oligopoly market and analyze the competition among multiple MNOs.
We provide some preliminary results along this direction in Appendix \ref{Appendix: Oligopoly Market}.
Third, we will study the MNOs' sequential data mechanism adoptions by following the example in \cite{duan2015economic}.

\bibliographystyle{IEEEtran}
\bibliography{ref}

\vspace{-30pt}
\begin{IEEEbiography}[{	\includegraphics[width=1in,height=1.25in,clip,keepaspectratio]{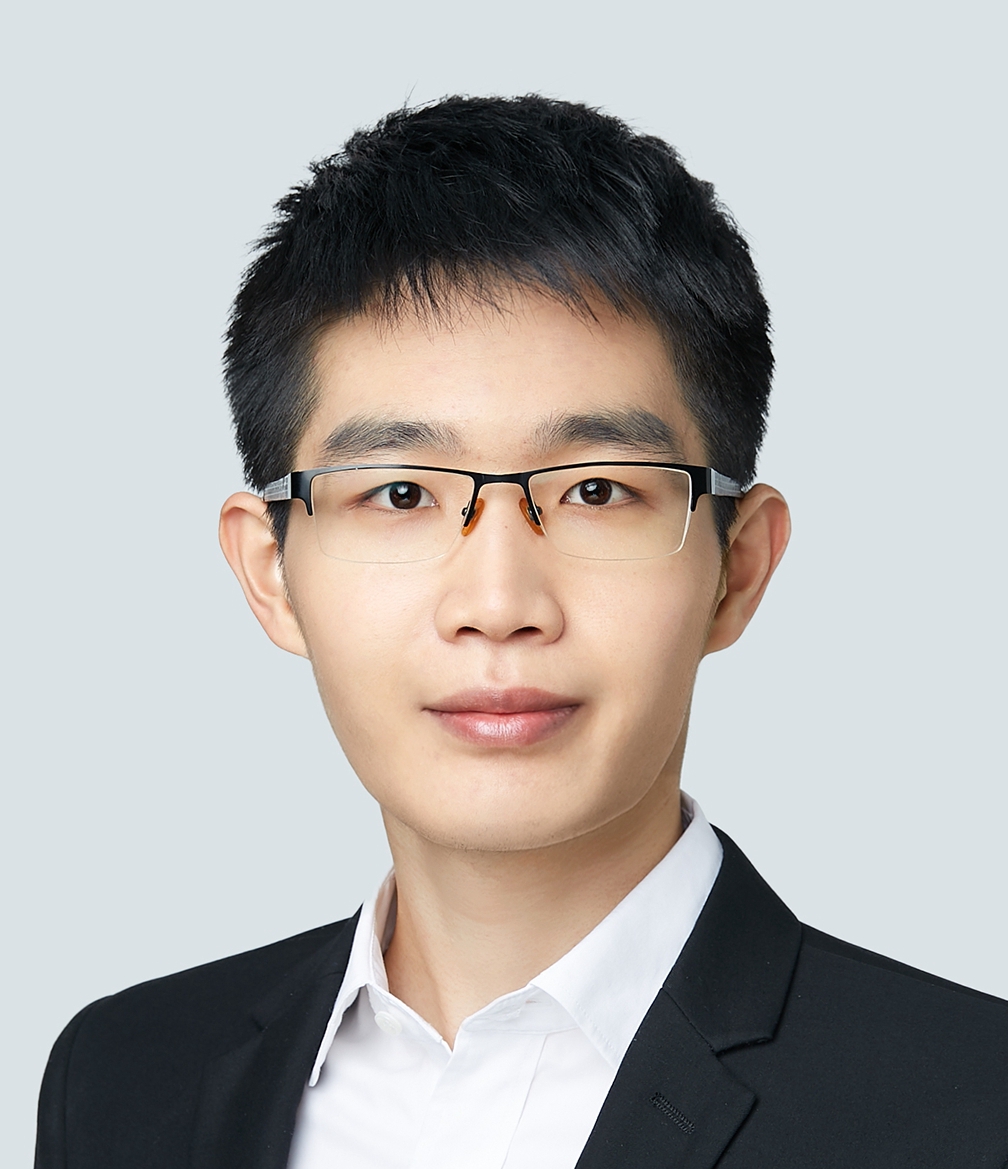}	}]{Zhiyuan Wang}
	received the B.S. degree from Southeast University, Nanjing, China, in 2016.
	He is currently working toward the Ph.D. degree with the Department of Information Engineering, The Chinese University of Hong Kong, Shatin, Hong Kong.  
	His research interests include the field of network economics and game theory, with current emphasis on smart data pricing and fog computing. 
	He is the recipient of the Hong Kong PhD Fellowship.
\end{IEEEbiography}

\vspace{-30pt}
\begin{IEEEbiography}[{	\includegraphics[width=1in,height=1.25in,clip,keepaspectratio]{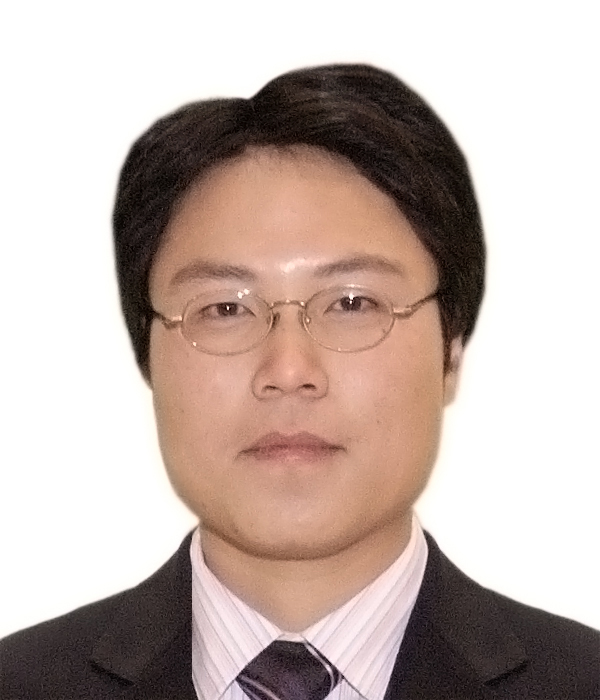}	}]	{Lin Gao}
	(S'08-M'10-SM'16) is an Associate Professor with the School of Electronic and Information Engineering, Harbin Institute of Technology, 	Shenzhen, China. He received the Ph.D. degree in Electronic Engineering from Shanghai Jiao Tong University in 2010. His main research
	interests are in the area of network economics 	and games, with applications in wireless communications 	and networking. 
	He received the	IEEE ComSoc Asia-Pacific Outstanding Young 	Researcher Award in 2016.
\end{IEEEbiography}

\vspace{-25pt}
\begin{IEEEbiography}
	[{	\includegraphics[width=1in,height=1.25in,clip,keepaspectratio]{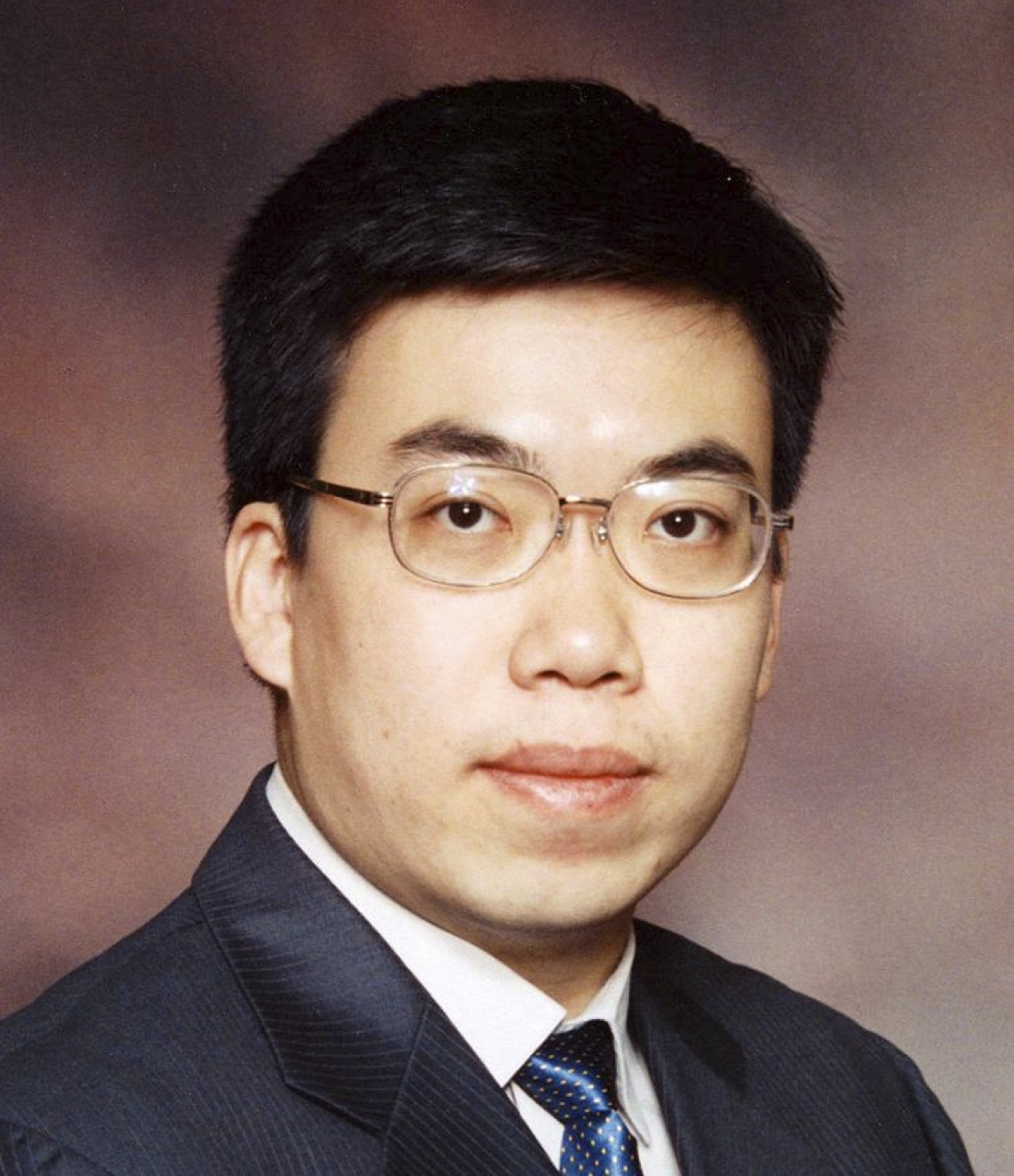}	}]	{Jianwei Huang}
	(F'16) is a Presidential Chair Professor and Associate Dean of the School of Science and Engineering, The Chinese University of Hong Kong, Shenzhen. 
	He is also a Professor in the 	Department of Information Engineering at The 	Chinese University of Hong Kong. 
	He is the
	co-author of 9 Best Paper Awards, including 	IEEE Marconi Prize Paper Award in Wireless
	Communications 2011. He has co-authored six 	books, including the textbook on ``Wireless Network
	Pricing''. He has served as the Chair of 	IEEE Technical Committee on Cognitive Networks and Technical Committee on Multimedia Communications. 
	He has been  an IEEE ComSoc 	Distinguished Lecturer and a Thomson Reuters 	Highly Cited Researcher.
\end{IEEEbiography}

\newpage
\appendices

\section{Oligopoly Market\label{Appendix: Oligopoly Market}}
In this section, we consider the competitive market with $N$ MNOs, denoted by $\mathcal{N}=\{1,2,...,N\}$.
Each user makes his subscription decision given the $N$ MNOs' pricing strategies $\pricing=\{s_n, \forall n\in \mathcal{N}\}$ and the data mechanisms $\Mechanism=\{\mechanism_n, \forall n\in \mathcal{N}\}$.
Without loss of generality, we suppose that $\QoS_1\usage_{\mechanism_1}>\QoS_2\usage_{\mechanism_2}>...>\QoS_N\usage_N$.
The type-$\val$ user will subscribe to MNO-$n$, denoted by $\val\in\subscription_n$, if and only if MNO-$n$ brings him an non-negative and the largest payoff among the $N$ MNOs, i.e.,
\begin{equation}
\left\{
\begin{aligned}
& \payoffexp_n(\plan_n,\val)\ge0, \\
& \payoffexp_n(\plan_n,\val)\ge\payoffexp_m(\plan_m,\val),\forall\ m\ne n.
\end{aligned}
\right.
\end{equation}

To facilitate later discussion on the market partition among the $N$ MNOs, we follow (\ref{Equ: xi}) and define $\xi_n^m$ as
\begin{equation}
\xi_n^m\eq\frac{\QoS_m\usage_{\mechanism_m}}{\QoS_n\usage_{\mechanism_n}},\ \forall\ n,m\in\mathcal{N}.
\end{equation}
We further express the neutral user type between MNO-$n$ and MNO-$m$, denoted by $\valeq_n^m$, as follows 
\begin{equation} \label{Equ: valeq n m}
\valeq_n^m(\valthr_n,\valthr_m)=\frac{\valthr_n-\xi_n^m\cdot\valthr_m}{1-\xi_n^m}.
\end{equation}

The market partition for $N$ competitive MNOs are much more complicated compared with the duopoly case.
Recall that there three market partitions in the duopoly market as discussed in Theorem \ref{Theorem: Market Partition}.
In the oligopoly case, however, there are much more competition outcomes in terms of which MNO would obtain a zero market share.
Therefore, we cannot enumerate all of the outcomes one by one.
Nevertheless, we summarize the coexistence outcome in Theorem \ref{Theorem: Partition N}.

\begin{theorem}[Coexistence of $N$ MNOs]\label{Theorem: Partition N}
	Consider  the $N$ MNOs' threshold user types $\Threshold=\{\valthr_n, \forall n\in \mathcal{N}\}$ under the data mechanism selection  $\Mechanism=\{\mechanism_n,\forall n\in\mathcal{N}\}$ and the pricing strategy $\pricing=\{s_n,\forall n\in\mathcal{N}\}$. 
	All of the MNOs obtain strictly positive market share as following
	\begin{equation} 
	\left\{
	\begin{aligned}
	& \subscription_1=[\valeq_1^2,\valmax],\\
	& \subscription_n=[\valeq_n^{n+1},\valeq_{n-1}^n], \forall\ n=2,3,...,N-1,\\
	& \subscription_N=[\valthr_N,\valeq_{N-1}^N],
	\end{aligned}
	\right.
	\end{equation}
	if and only if the corresponding threshold user types $\valthr_n$ ($n\in\mathcal{N}$) satisfies
	{\small \begin{equation} 
		\left\{
		\begin{aligned}
		& 0\le \valthr_N<\valthr_{N-1}, \\
		& (1-\xi_{n-1}^{n+1})\valthr_n
		< (1-\xi_{n}^{n+1})\valthr_{n-1} + (\xi_n^{n+1}-\xi_{n-1}^{n+1})\valthr_{n+1},\\
		&  \qquad\qquad\qquad\qquad\qquad\qquad\qquad\quad \forall\ n=2,3,...,N-1,\\
		& (1-\xi_1^2)\valmax+\xi_1^2\valthr_2 > \valthr_1 .
		\end{aligned}
		\right.
		\end{equation}}
\end{theorem}

\begin{figure} 
	\centering
	\includegraphics[width=0.7\linewidth]{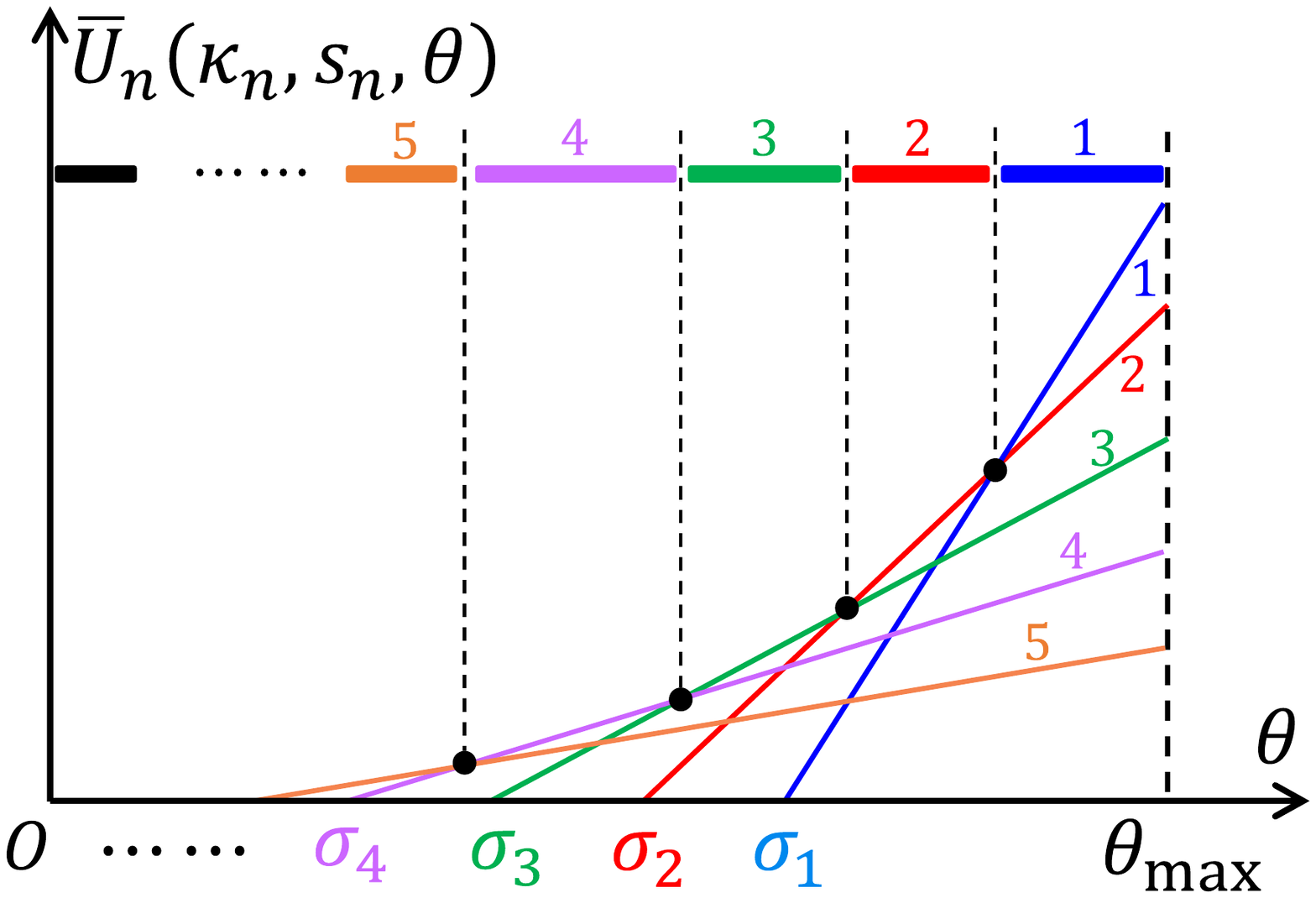}
	\caption{The $N$ MNOs coexist.}
	\label{fig: Competition_N}
\end{figure}

Fig. \ref{fig: Competition_N} illustrates the coexistence outcome discussed in Theorem \ref{Theorem: Partition N}.
Recall that $\QoS_n\usage_{\mechanism_n}$ represents a subscriber's marginal utility change for one unit data valuation increment.
The MNO that provides a better QoS (i.e., $\QoS_n$) and better time flexibility (i.e., $\mechanism_n$), can obtain the higher valuation subscribers and charge higher prices, hence obtains more revenue.

Under such a market partition, the corresponding profit of each MNO is given by
\begin{equation}
\begin{aligned}
&\profit_n(\Mechanism,\Threshold)= \\
& \left\{
\begin{aligned}
& \QoS_1\usage_{\mechanism_1} \left[ \valthr_1 -\costqos_1 \right]  \Big[ 1-H\left(\valeq_1^2 \right) \Big],	\qquad\qquad\qquad\qquad n=1, \\
& \QoS_n\usage_{\mechanism_n} \left[ \valthr_n -\costqos_n \right]  \Big[ H\left(\valeq_{n-1}^{n} \right) - H\left(\valeq_{n}^{n+1} \right) \Big], \ \  1< n<N,\\
& \QoS_N\usage_{\mechanism_N} \left[ \valthr_N -\costqos_N \right]  \Big[ H\left(\valeq_{N-1}^{N} \right) - H\left(\valthr_N \right) \Big], \qquad  n=N,
\end{aligned}\right.
\end{aligned}
\end{equation}
where $\valeq_{n}^{n+1}$ depends on $\valthr_n$ and $\valthr_{n+1}$, given by 
\begin{equation}
\valeq_{n}^{n+1}=\frac{\valthr_n-\xi_n^{n+1} \cdot \valthr_{n+1}}{1-\xi_n^{n+1}}.
\end{equation}


\section{Monopoly Market as Benchmark\label{Appendix: Monopoly Market as Benchmark}}
In this section, we introduce the MNO's optimal decisions on the threshold user type and the data mechanism in the monopoly market.
The results in this section is partially based on our previous results in \cite{Zhiyuan2018TMC}, since we assume the homogeneity in the network substitutability $\cut$.

Without loss of generality, let's consider the monopoly market of MNO-$n$ under the data mechanism $\mechanism_n$ and the pricing strategy $s_n=\{\pcap_n,\adfee_n\}$.
According to Definition \ref{Definition: Threshold User Type}, we denote $\valthr_n(\mechanism_n,s_n)$ the threshold user type and the market share of MNO-$n$ is $\subscription_n=[\valthr_n(\mechanism_n,s_n),\valmax]$.
Therefore, the MNO-$n$'s expected profit under the  data mechanism $\mechanism_n$ and the pricing strategy $s_n$ is 
\begin{equation} \label{Equ: W_n MP}
\begin{aligned}
	&\profit_n(\mechanism_n,s_n)	\\
	=& \QoS_n\usage_{\mechanism_n} \left[\valthr_n(\mechanism_n,s_n) - \frac{c_n}{\QoS_n}\right] \Big[ 1-H\big(\valthr_n(\mechanism_n,s_n) \big) \Big],
\end{aligned}
\end{equation}
	where $H(\cdot)$ is the CDF of users' data valuation $\val$.
	Note that MNO-$n$ experiences a negative profit if its corresponding threshold user type $\valthr_n(\mechanism_n,s_n) < {c_n}/{\QoS_n}$, which is a trivial case.
	Moreover, MNO-$n$ will have no subscriber if its cost-QoS ratio ${c_n}/{\QoS_n}$ is greater than the users' highest data valuation $\valmax$.
	With this observation, we will focus on the case where  $\valthr_n(\mechanism_n,s_n) \in[{c_n}/{\QoS_n},\valmax]$, where we assume that ${c_n}/{\QoS_n}<\valmax$  to avoid the trivial case. 
	
	We characterize MNO-$n$'s profit-maximizing pricing strategy $s_n^\text{MP}$ and data mechanism $\mechanism_n^\text{MP}$ in {Lemma \ref{Lemma: Monopoly Optimal Pricing}} and {Lemma \ref{Lemma: Monopoly Optimal Data Mechanism}}, respectively.
	Here the superscript ``MP'' means ``monopoly''.
	\begin{lemma}\label{Lemma: Monopoly Optimal Pricing}
		For monopoly MNO-$n$, given the data mechanism $\mechanism_n$, it maximizes its profit through a pricing strategy $s_n^\text{MP}$ such that its threshold user type  $\valthr_n(\mechanism_n,s_n^\text{MP})=\valthr_n^\text{MP}(\mechanism_n)$, which is the solution to the following equation:
		\begin{equation}\label{Equ: sigma^{MP}}
			\valthr_n^\text{MP}(\mechanism_n)-\frac{1-H(\valthr_n^\text{MP}(\mechanism_n))}{h(\valthr_n^\text{MP}(\mechanism_n))}=\frac{c_n}{\QoS_n},
		\end{equation}	
		where $\valthr_n^\text{MP}(\mechanism_n)$ is unique for an arbitrary $\val$ distribution with  the increasing failure rate (IFR).  
	\end{lemma}
	\begin{proof}[\bf {Proof of Lemma \ref{Lemma: Monopoly Optimal Pricing}}]
		We prove Lemma \ref{Lemma: Monopoly Optimal Pricing} by deriving the MNO's profit-maximizing threshold user type $\valthr_n^\text{MP}$.
		
		Recall that the MNO-$n$'s expected monthly profit under the  data mechanism $\mechanism_n$ and the pricing strategy $s_n$ is 
		\begin{equation} \label{Equ: W_n MP proof}
		\begin{aligned}
		&\profit_n(\mechanism_n,s_n)	\\
		=& \QoS_n\usage_{\mechanism_n} \left[\valthr_n(\mechanism_n,s_n) - \frac{c_n}{\QoS_n} \right] \big[ 1-H(\valthr_n(\mechanism_n,s_n) ) \big].
		\end{aligned}
		\end{equation}
		
		Given the data mechanism $\mechanism_n$, the MNO's expected profit can be expressed as a function of the threshold user type $\valthr_n$, as follows:
		\begin{equation}\label{Proof Equ: profit n}
		\begin{aligned}
		&\profit_n(\mechanism_n,\valthr_n)	=  \QoS_n\usage_{\mechanism_n} \left[\valthr_n - \frac{c_n}{\QoS_n}\right] \big[ 1-H(\valthr_n ) \big].
		\end{aligned}
		\end{equation}
		To compute the maximum value of (\ref{Proof Equ: profit n}), we take the derivative of (\ref{Proof Equ: profit n}) with respect to $\valthr_n$ and obtain 
		\begin{equation}\label{Equ: sigma^MP}
		\begin{aligned}
		\frac{\partial W_n(\mechanism_n,\valthr_n)}{\partial \valthr_n}
		&=\textstyle \QoS_n\usage_{\mechanism_n}h(\valthr_n)\cdot g(\valthr_n),
		\end{aligned}
		\end{equation}
		where $h(\cdot)$ is the PDF of the data valuation $\val$ and $g(\cdot)$ is given by
		\begin{equation}\label{Equ: sigma^MP 1}
		 g(\valthr_n)\eq\frac{1-H(\valthr_n)}{h(\valthr)}-\valthr_n+\frac{c_n}{\QoS_n}.
		\end{equation}
		Note that $g(\cdot)$ is a monotonically decreasing function if the distribution of $\val$ satisfies the IFR.
		In addition, we can show that 
		\begin{equation}
		\left\{
		\begin{aligned}
		& g(0)=+\infty>0,\\
		&\textstyle g(\valmax)=\frac{c_n}{\QoS_n}-\valmax<0,
		\end{aligned}
		\right.
		\end{equation}
		which implies that there exists a \textit{unique} $\valthr_n^\text{MP}$ satisfying 
		\begin{equation}\label{Equ: sigma^MP 2}
		\left\{
		\begin{aligned}
		& \textstyle g\left(\valthr_n^\text{MP}\right)=0,\\
		& \textstyle g\left(\valthr_n\right)>0,	&\text{if } \valthr_n < \valthr_n^\text{MP},	\\
		& \textstyle g\left(\valthr_n\right)<0,	&\text{if } \valthr_n > \valthr_n^\text{MP}.	\\		
		\end{aligned}
		\right.
		\end{equation}	
		
		Combining (\ref{Equ: sigma^MP}) and (\ref{Equ: sigma^MP 2}), we know that 
		\begin{equation}
		\left\{
		\begin{aligned}
		& \frac{\partial W_n(\mechanism_n,\valthr_n)}{\partial \valthr_n}=0, &\text{if } & \valthr_n=\valthr_n^\text{MP},	\\
		& \frac{\partial W_n(\mechanism_n,\valthr_n)}{\partial \valthr_n}>0, &\text{if } & \valthr_n<\valthr_n^\text{MP},\\
		& \frac{\partial W_n(\mechanism_n,\valthr_n)}{\partial \valthr_n}<0, &\text{if } & \valthr_n>\valthr_n^\text{MP},
		\end{aligned}
		\right.
		\end{equation}
		which indicates that the MNO-$n$ maximizes its profit if its threshold user type is $\valthr_n^\text{MP}$, i.e., 
		\begin{equation}\label{Equ: sigma^{MP}}
		\valthr_n^\text{MP}-\frac{1-H(\valthr_n^\text{MP})}{h(\valthr_n^\text{MP})}=\frac{c_n}{\QoS_n}.
		\end{equation}
	\end{proof}
	
	{Lemma \ref{Lemma: Monopoly Optimal Pricing}} reveals the trade-off between the subscription fee and the per-unit fee, i.e., the profit-maximizing subscription fee $\pcap_n^\text{MP}$ and per-unit fee $\adfee_n^\text{MP}$ need to satisfy
	\begin{equation}\label{Equ: MP pcap adfee}
		\textstyle \adfee_n^\text{MP}\left(\cut^{-1}-1\right) \left(\dmean-\usage_{\mechanism_n}\right) + \pcap_n^\text{MP} = \QoS_n\usage_{\mechanism_n} \valthr_n^\text{MP}(\mechanism_n).
	\end{equation}
	
	A larger $\adfee_n^\text{MP}$ would lead to a smaller $\pcap_n^\text{MP}$, and vice versa. 
	
	

	\begin{lemma}\label{Lemma: Monopoly Optimal Data Mechanism}
		Under the optimal pricing strategy in {Lemma \ref{Lemma: Monopoly Optimal Pricing}}, a monopoly MNO-$n$ obtains a higher profit under  the rollover mechanism (than the traditional mechanism $\trad$), i.e., $\mechanism_n^\text{MP}=\roll$.
	\end{lemma}
	\begin{proof}[\bf {Proof of Lemma \ref{Lemma: Monopoly Optimal Data Mechanism} }]
		Under the optimal pricing strategy specified in Lemma \ref{Lemma: Monopoly Optimal Pricing}, the MNO's profit is given by
		\begin{equation}\label{Proof Equ: profit n kappa}
		\begin{aligned}
		&\profit_n(\mechanism_n,\valthr_n^\text{MP})	= \textstyle \QoS_n\usage_{\mechanism_n} \left[\valthr_n^\text{MP} - \frac{c_n}{\QoS_n}\right] \Big[ 1-H\left(\valthr_n^\text{MP} \right) \Big].
		\end{aligned}
		\end{equation}
		
		Since $\usage_{\roll}>\usage_{\trad}$, we know that $\profit_n(\roll,\valthr_n^\text{MP})>\profit_n(\trad,\valthr_n^\text{MP})$, which indicates that $\mechanism_n^\text{MP}=\roll$.
	\end{proof}
	
	{Lemma \ref{Lemma: Monopoly Optimal Data Mechanism}} shows that the monopoly MNO should select the rollover mechanism $\roll$ to maximize its profit.
	This conclusion still holds under users' two-dimensional heterogeneity on data valuation and network substitutability.
	We refer interested readers to our previous works in \cite{Zhiyuan2018TMC}.
	

\section{Bertrand Competition\label{Appendix: Bertrand}}
In the main paper, we make Assumption \ref{Assumption: xi} in Section \ref{Section: User Subscription in Stage III} and Section \ref{Section: MNOs' Pricing Competition in Stage II}.
Now we consider the case of $\xi(\Mechanism)=1$.
More specifically, we focus on showing that our analysis for $\xi(\Mechanism)\ne1$ (in Sections \ref{Section: User Subscription in Stage III} and \ref{Section: MNOs' Pricing Competition in Stage II} of the main paper) is also applicable to the case of $\xi(\Mechanism)=1$ in terms of the market partition and the MNOs' best responses.

\subsection{Market Partition}
We summarize the user subscription equilibrium when $\xi(\Mechanism)=1$ in Theorem \ref{Theorem: Market Partition same}.
\begin{theorem}[Market Partition for $\xi(\Mechanism)=1$]\label{Theorem: Market Partition same}
	Consider  MNOs' threshold user types $\valthr_1$ and $\valthr_2$ under the data mechanism $\Mechanism=\{\mechanism_1,\mechanism_2\}$ and the pricing strategy $\pricing=\{s_1,s_2\}$.
	There are two market competition results in the $(\valthr_1,\valthr_2)$ plane shown in Fig. \ref{fig: Partition structure same}.
	\begin{enumerate}
		\item $\Sigma_1$: MNO-1 has a larger threshold user type than MNO-2, i.e., $(\valthr_1,\valthr_2)\in\Sigma_1$ where $\Sigma_1$ is
		\begin{equation}
		(\valthr_1,\valthr_2)\in\Sigma_1\eq\{(\valthr_1,\valthr_2):\valthr_1>\valthr_2 \}.
		\end{equation}
		In this case, 	MNO-2's market share corresponds to the users with $\val$ in  $\subscription_2=[\valthr_2,\valmax]$, while MNO-1 has a  zero market share $\subscription_1=\varnothing$, as shown in Fig. \ref{fig: Duopoly_same_1}.
		\item $\Sigma_2$: MNO-1 has a smaller threshold user type than MNO-2, i.e., $(\valthr_1,\valthr_2)\in\Sigma_2$ where $\Sigma_2$ is
		\begin{equation}
		(\valthr_1,\valthr_2)\in\Sigma_2\eq\{(\valthr_1,\valthr_2):\valthr_1 <\valthr_2\}.
		\end{equation}
		In this case, MNO-1 has a market share of $\subscription_1=[\valthr_1,\valmax]$, while MNO-2 has a zero market share of $\subscription_2=\varnothing$, as shown in Fig. \ref{fig: Duopoly_same_2}.
	\end{enumerate}
\end{theorem}

By comparing Theorem \ref{Theorem: Market Partition same} with Theorem \ref{Theorem: Market Partition}, we find that Theorem \ref{Theorem: Market Partition same} is a special case of Theorem \ref{Theorem: Market Partition} if $\xi(\Mechanism)=1$.
That is, as $\xi(\Mechanism)$ is approaching to 1, the gray region in Fig. \ref{fig: Partition structure} will diminish and eventually becomes Fig. \ref{fig: Partition structure same}.

Next we study the best response of each MNO based on the subscription equilibrium in Theorem \ref{Theorem: Market Partition same}.

\begin{figure}
	\centering
	\includegraphics[width=0.5\linewidth]{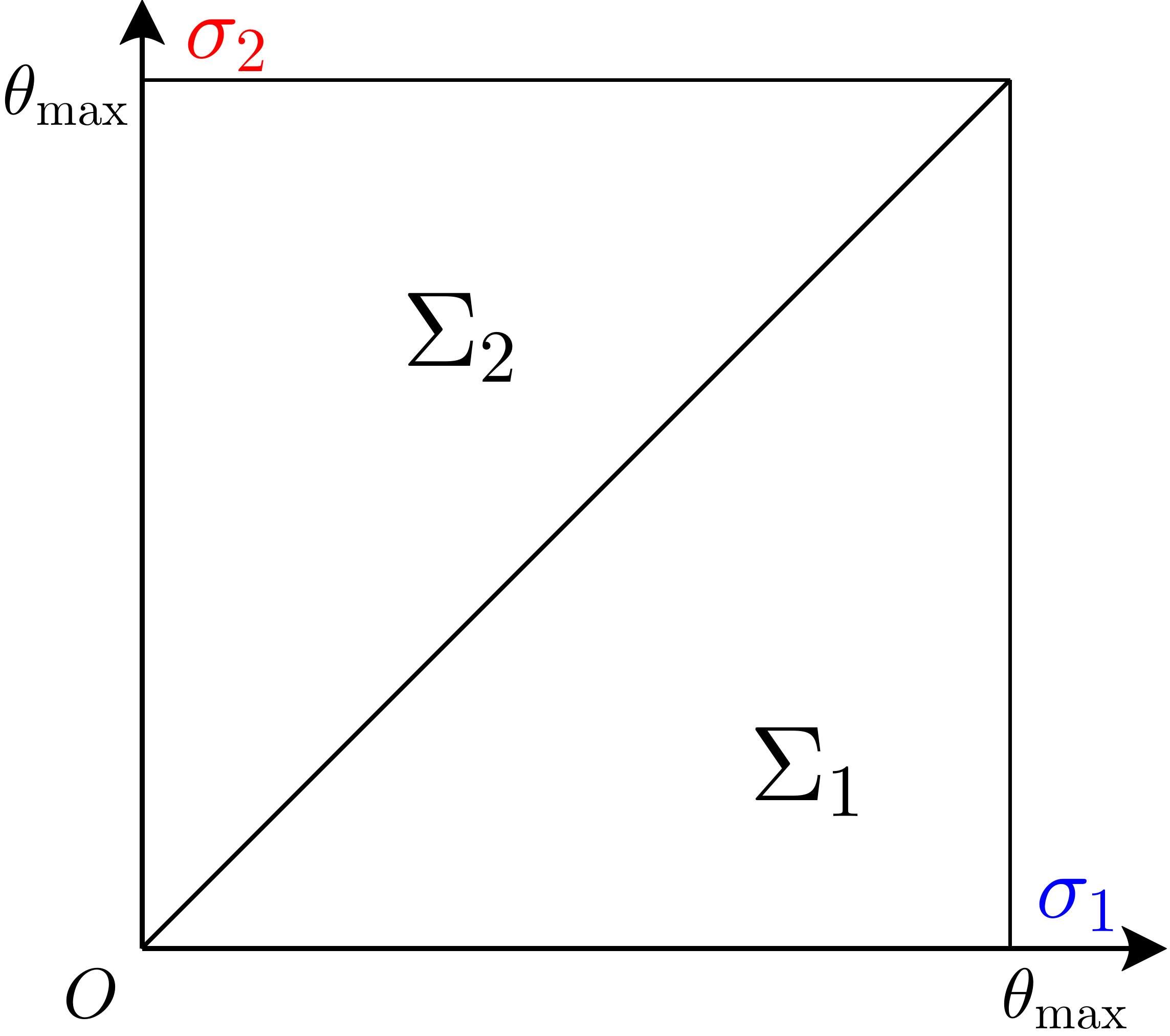}
	\caption{Partition structure for $\xi(\Mechanism)=1$.}
	\label{fig: Partition structure same}
\end{figure}

\begin{figure}
	\centering
	\subfigure[MNO-1 surviving ($\Sigma_2$)]{\label{fig: Duopoly_same_1}{\includegraphics[width=0.48\linewidth]{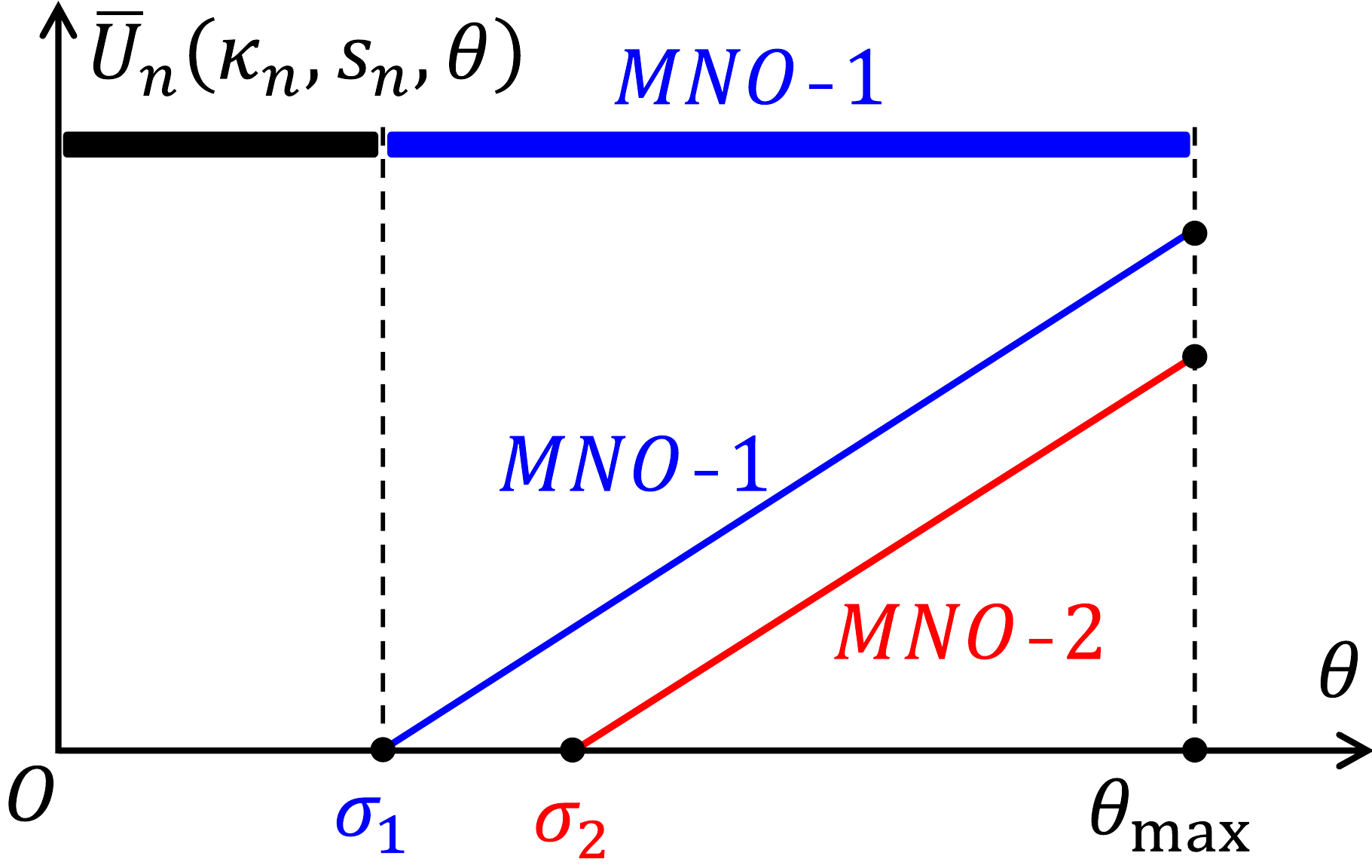}}} 
	\subfigure[MNO-2 surviving ($\Sigma_1$)]{\label{fig: Duopoly_same_2}{\includegraphics[width=0.48\linewidth]{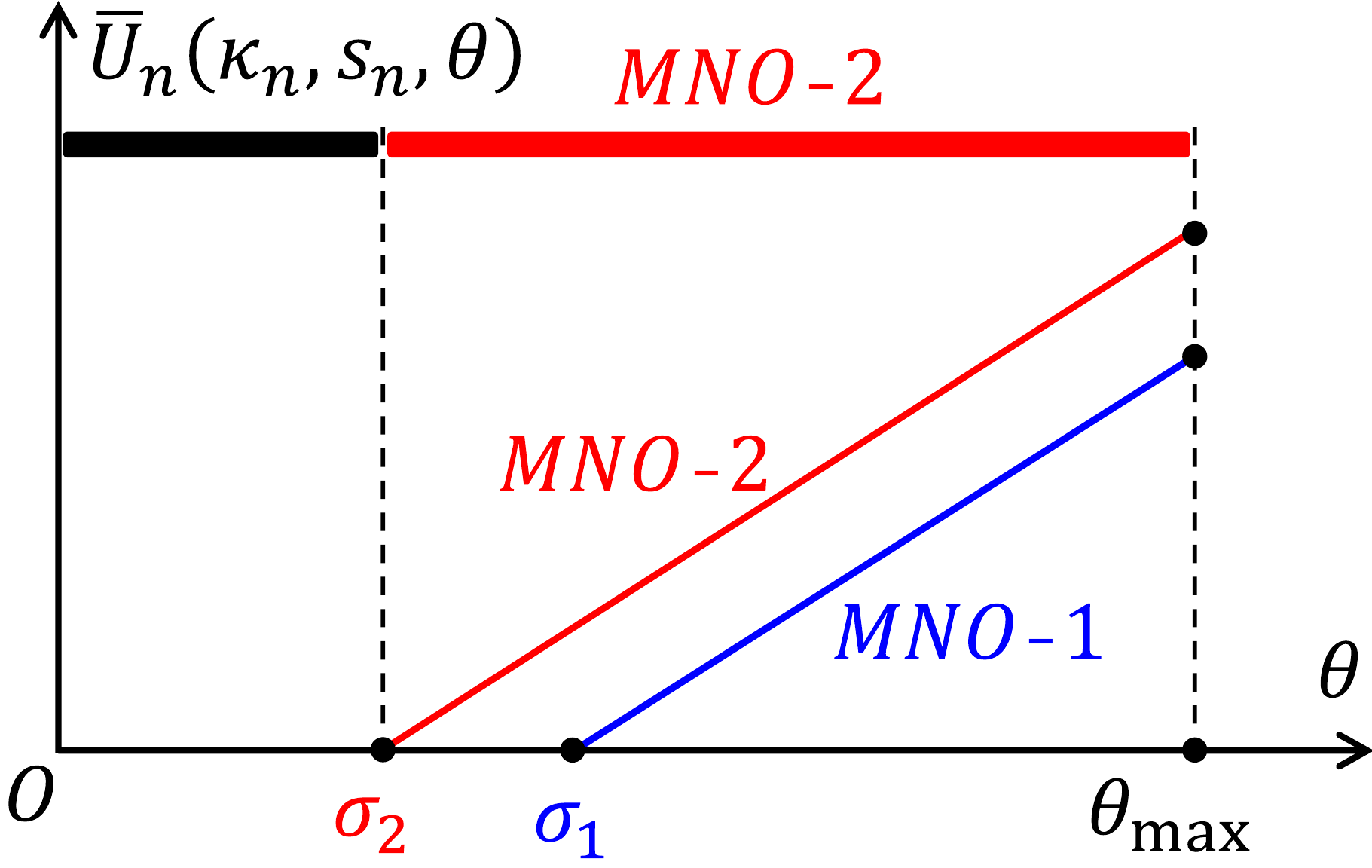}}}
	\caption{Two market partition modes when $\xi(\Mechanism)=1$. }
	\label{fig: Partition same}
\end{figure} 

\subsection{Best Response}

For $\xi(\Mechanism)=1$, the two MNOs are symmetric, i.e., $\QoS_1\usage_{\mechanism_1}=\QoS_2\usage_{\mechanism_2}$.
In this case, their best responses are the same.
Therefore, we will take MNO-2 as example by presenting its best response to MNO-1 in Lemma \ref{Lemma: theta_2^* same}
\begin{lemma}[Best Response of MNO-2 for $\xi(\Mechanism)=1$] \label{Lemma: theta_2^* same}
	Given MNO-1's threshold user type $\valthr_1$,  MNO-2 maximizes its profit $\profit_2(\Mechanism,\Threshold)$ by choosing a  threshold user type $\valthr_2^{*}(\Mechanism,\valthr_1) $ as follows: 
	\begin{equation}
	\valthr_2^{*}(\Mechanism,\valthr_1)=
	\begin{cases} 
	\costqos_2, 						& \text{if } \valthr_1\in [0,\val_1^W) , \\
	[\valthr_1]^-,  					& \text{if } \valthr_1\in [\val_1^W,\val_1^N),\\
	\valthr_2^\text{MP}(\mechanism_2), 	& \text{if }\valthr_1 \in [\val_1^N,\valmax],
	\end{cases}
	\end{equation} 
	where $[x]^-$ denotes the value slightly lower than $x$.	
\end{lemma}

Fig. \ref{fig: BR2_same} illustrates the MNO-2's best response specified in Lemma \ref{Lemma: theta_2^* same}. 
Specifically, the red line segments (i.e., BR2-a, BR2-b, and BR2-c) denote $\valthr_2^{*}(\Mechanism,\valthr_1)$. 
Next we discuss the physical meanings of the three different parts of the best response in more details.
\begin{itemize}
	\item \textbf{BR2-a}: MNO-2 gives up the competition and obtains a zero market share, i.e., $(\valthr_1,\valthr_2^{*}(\valthr_1))\in\Sigma_2$, if MNO-1 chooses a threshold user type smaller than its \textit{winning} threshold, i.e., $\valthr_1\le\val_1^W$.
	\item \textbf{BR2-b}: MNO-2 leaves a zero market to MNO-1, i.e., $(\valthr_1,\valthr_2^{*}(\valthr_1))\in\Sigma_1$, by choosing a threshold user type $\valthr_2^{*}(\Mechanism,\valthr_1)$ slightly smaller than $\valthr_1$, if MNO-1 chooses a threshold user type between its \textit{winning} and \textit{no-influence} thresholds, i.e., $\val_1^W<\valthr_1<\val_1^N$.
	\item \textbf{BR2-c}: MNO-2 leaves a zero market to MNO-1 and decides its threshold user type as in Appendix \ref{Appendix: Monopoly Market as Benchmark} without considering the existence of MNO-1, i.e., $\valthr_2^{*}(\Mechanism,\valthr_1)=\valthr_2^\text{MP}$ as defined in (\ref{Equ: sigma^{MP}}), if MNO-1 chooses a threshold user type larger than its \textit{no-influence} threshold, i.e., $\valthr_1\ge\val_1^N$.
\end{itemize}

By comparing Lemma \ref{Lemma: theta_2^* same} with Lemma \ref{Lemma: theta_2^*}, we find that the MNO-2's best response under $\xi(\Mechanism)=1$ is a special case of that under $\xi(\Mechanism)\ne1$.
That is, the illustration in Fig. \ref{fig: Best response of MNO-2} will become Fig. \ref{fig: BR2_same} as $\xi(\Mechanism)$ is approaching to 1.

Furthermore, the best response of MNO-1 is similar to Lemma \ref{Lemma: theta_2^* same}, and we illustrate it in Fig. \ref{fig: BR1_same}.
For an easy comparison with Fig. \ref{fig: BR2_same}, in Fig. \ref{fig: BR1_same} we plot the best response $\valthr_1^{*}(\valthr_2)$ on the horizontal axis and the variable $\valthr_2$ on the vertical axis.

\begin{figure} 
	\centering
	\subfigure[Illustration of $\valthr_2^{*}(\Mechanism,\valthr_1)$.]{\label{fig: BR2_same}{\includegraphics[height=0.43\linewidth]{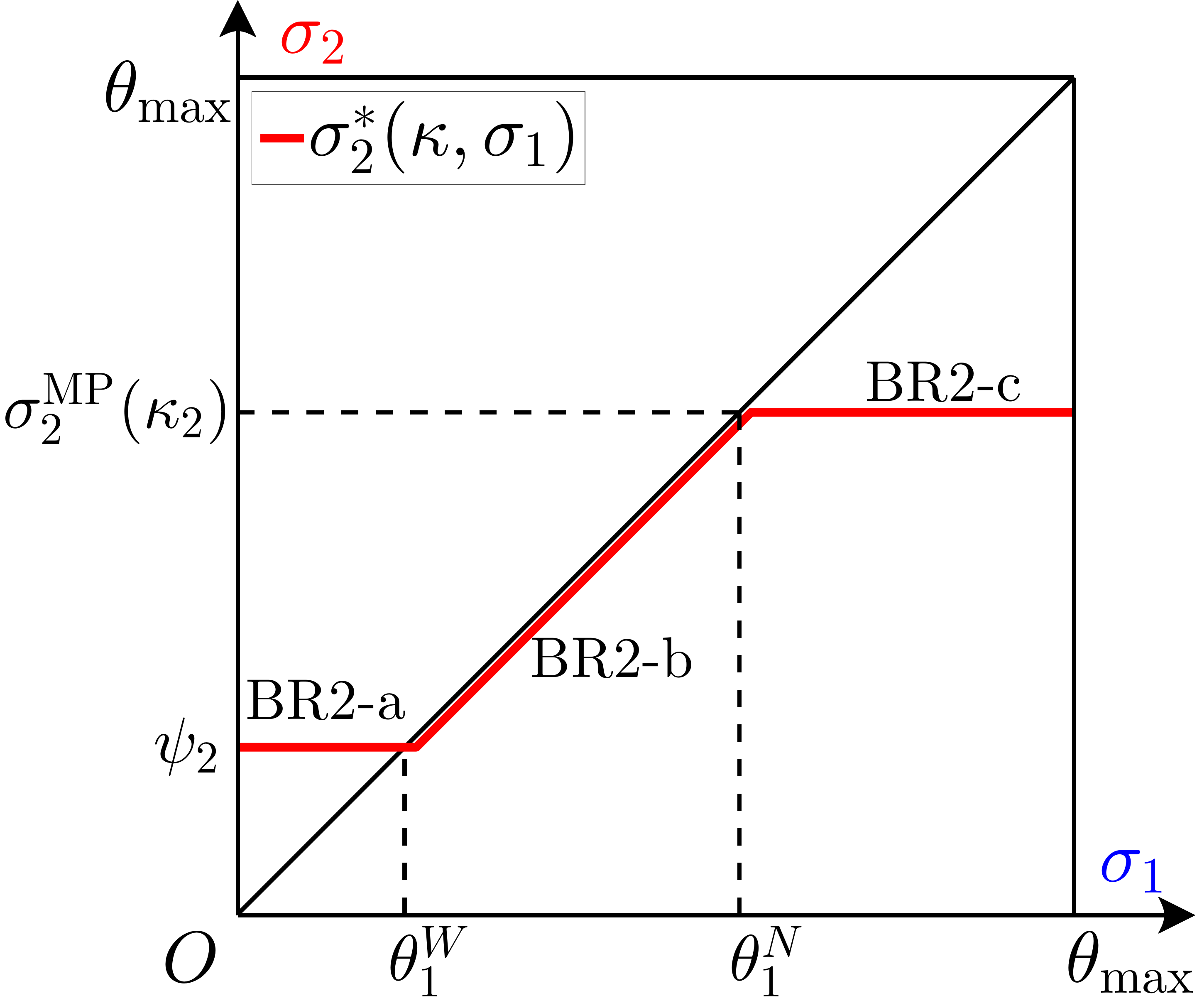}}}
	\subfigure[Illustration of $\valthr_2^{*}(\Mechanism,\valthr_1)$.]{\label{fig: BR1_same}{\includegraphics[height=0.43\linewidth]{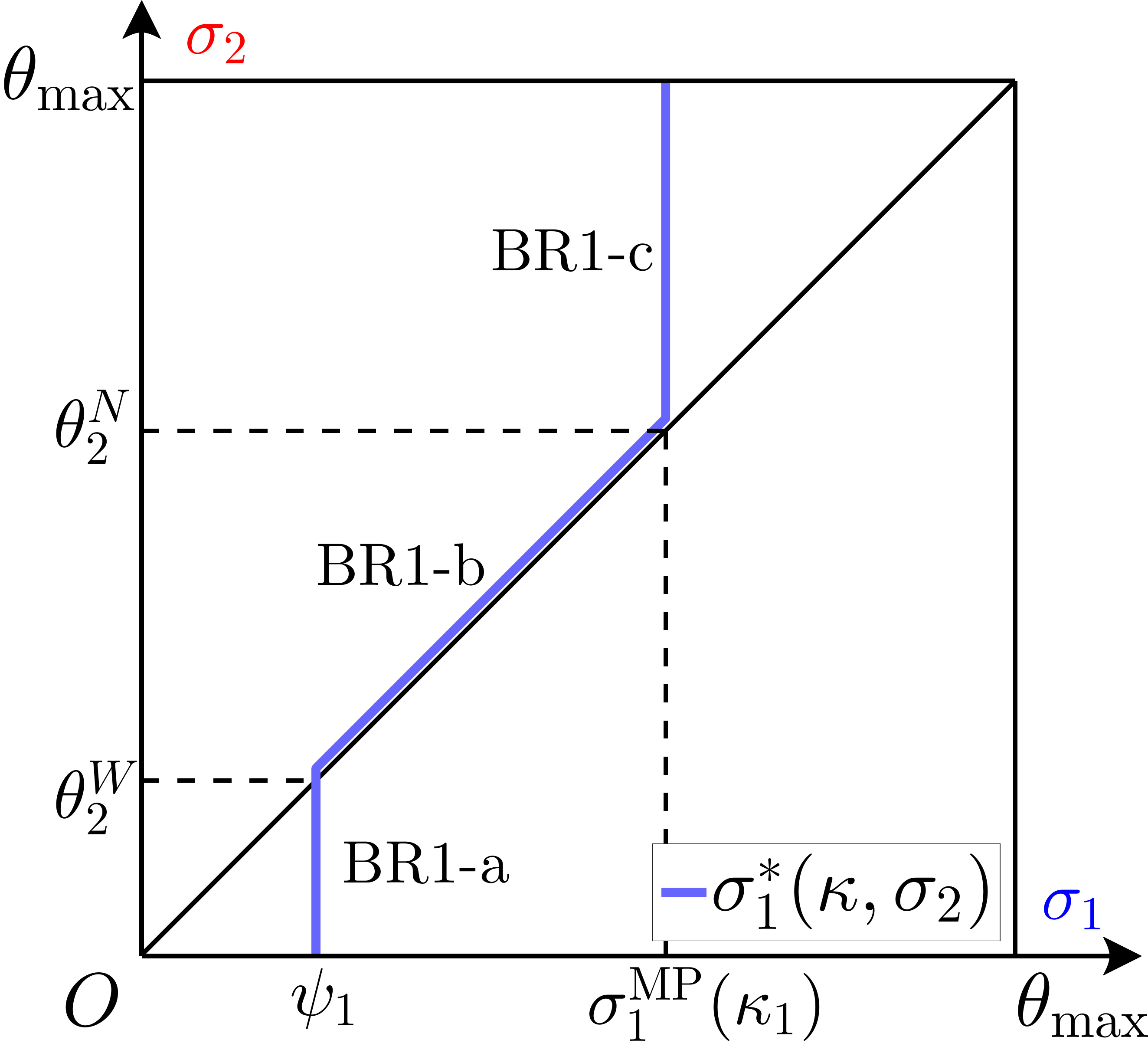}}} 
	\caption{Illustration of the best responses for $\xi(\Mechanism)=1$.}
	\label{fig: Best response same}
\end{figure}

\section{}\label{Appendix: Partition BRs}

\begin{proof}[\bf {Proof of Theorem \ref{Theorem: Market Partition} }]
	We prove Theorem \ref{Theorem: Market Partition} by characterizing the conditions of the three market partitions.
	Recall that the neutral user type $\valeq$ is
	\begin{equation}\label{Proof Equ: valeq}
		\valeq(\valthr_1,\valthr_2) = \frac{ \valthr_1 - \xi(\Mechanism) \cdot \valthr_2  }{ 1 - \xi(\Mechanism) }.
	\end{equation}
	
	According to the illustration in Fig. \ref{fig: Duopoly_1}, MNO-2 leaves a zero market share to MNO-1 if and only if \textit{the neutral user type is no smaller than the maximal data valuation of the user group}, i.e., $\valeq(\valthr_1,\valthr_2)\ge\valmax$.
	It is mathematically equivalent to
	\begin{equation}
	\valthr_1-\valthr_2 \ge (1-\xi(\Mechanism))(\valmax-\valthr_2),
	\end{equation}
	which corresponds to the set $\Sigma_1$ in Theorem \ref{Theorem: Market Partition}.

	According to the illustration in Fig. \ref{fig: Duopoly_3}, MNO-1 leaves a zero market share to MNO-2 if and only if \textit{the threshold user type of MNO-1 is no larger than that of MNO-2}, i.e., $\valthr_1\le\valthr_2$.
	It corresponds to the set $\Sigma_2$  in Theorem \ref{Theorem: Market Partition}.

	According to the illustration in Fig. \ref{fig: Duopoly_2}, both the MNOs obtain positive market shares if and only if \textit{the neutral user type is smaller than the maximal data valuation of the user group} (i.e., $\valeq<\valmax$) and \textit{the threshold user type of MNO-1 is larger than that of MNO-2} (i.e., $\valthr_1>\valthr_2$).
	They are equivalent to
	\begin{equation}
	0 < \valthr_1 -\valthr_2 < (1-\xi(\Mechanism))(\valmax-\valthr_2),
	\end{equation} 
	which corresponds to the set $\Sigma_3$ in Theorem \ref{Theorem: Market Partition}.	
\end{proof}

\begin{proof}[\bf {Proof of Lemma \ref{Lemma: theta_1^*}}]
	The proof of the two MNOs' best responses is similar, here we take MNO-1 as an example.
	We prove Lemma \ref{Lemma: theta_1^*} by characterizing the best response of MNO-1 to MNO-2.
	
	Based on Theorem \ref{Theorem: Market Partition}, both the MNOs obtain positive market shares if and only if
	\begin{equation}\label{Proof Equ: share condition}
	0 < \valthr_1 -\valthr_2 < (1-\xi(\Mechanism))(\valmax-\valthr_2).
	\end{equation}
	The corresponding profit of MNO-1 is
	\begin{equation}\label{Proof Equ: W1}
	\profit_1(\Mechanism,\Threshold) 
	= \QoS_1\usage_{\mechanism_1} \left[ \valthr_1 -\costqos_1 \right]  \Big[ 1-H\left(\valeq(\valthr_1,\valthr_2)  \right) \Big].
	\end{equation}
	
	From (\ref{Proof Equ: share condition}), we can derive the following two inequalities for the coexistence case regarding to $\valthr_1$
	\begin{subequations}\label{Proof Equ: coexisten condition}
	\begin{align}
	\valthr_1 &> \valthr_2,\\
	\valthr_1 &< \valthr_2 + (1-\xi(\Mechanism))(\valmax-\valthr_2).
	\end{align}
	\end{subequations}	
	
	Based on (\ref{Proof Equ: W1}), we take the first order derivative of $\profit_1(\Mechanism,\Threshold)$ with respect to $\valthr_1$ and obtain
	\begin{equation}\label{Proof Equ: W1 derivative}
	\begin{aligned}
	\frac{\partial \profit_1(\Mechanism,\Threshold)}{\partial \valthr_1} 
	=& \QoS_1\usage_{\mechanism_1} \Bigg[ 1-H\left(\frac{\valthr_1-\xi(\Mechanism)\valthr_2}{1-\xi(\Mechanism)}\right)  \\ 
	& \quad\quad -\frac{\valthr_1-\costqos_1}{1-\xi(\Mechanism)}\cdot h\left(\frac{\valthr_1-\xi(\Mechanism)\valthr_2}{1-\xi(\Mechanism)}\right)  \Bigg] . \\
	\end{aligned}
	\end{equation}
	From (\ref{Proof Equ: W1 derivative}), we can see that $ \frac{\partial \profit_1}{\partial \valthr_1}=0 $ is equivalent to 
	\begin{equation}
	g(\valthr_1,\valthr_2)=0,
	\end{equation}
	where $g(\valthr_1,\valthr_2)$ is 	
	\begin{equation}
	g(\valthr_1,\valthr_2)=\frac{1-H\left(\frac{\valthr_1-\xi(\Mechanism)\valthr_2}{1-\xi(\Mechanism)}\right)}{h\left(\frac{\valthr_1-\xi(\Mechanism)\valthr_2}{1-\xi(\Mechanism)}\right)} - \frac{\valthr_1-\costqos_1}{1-\xi(\Mechanism)}  .
	\end{equation}
	
	Furthermore, under the IFR condition, we can show that $g(\valthr_1,\valthr_2)$ is decreasing in $\valthr_1$.
	Therefore, when both the MNOs obtain positive market shares (i.e., the two inequalities in (\ref{Proof Equ: coexisten condition}) hold), we have
	\begin{equation}
	\underbrace{g(\valthr_2 + (1-\xi(\Mechanism))(\valmax-\valthr_2),\valthr_2)}_{J(\valthr_2)} < g(\valthr_1,\valthr_2) < \underbrace{g(\valthr_2,\valthr_2)}_{K(\valthr_2)}.
	\end{equation}
	Specifically, $J(\valthr_2)$ and $K(\valthr_2)$ are given by
	\begin{equation}
		J(\valthr_2)=- \frac{ \xi(\Mechanism)\valthr_2 -\costqos_1}{1-\xi(\Mechanism)}-\valmax,
	\end{equation}
	\begin{equation}
		K(\valthr_2)=\frac{1-H\left(\valthr_2\right)}{h\left(\valthr_2\right)} - \frac{\valthr_2-\costqos_1}{1-\xi(\Mechanism)}.
	\end{equation}
	Note that both $J(\valthr_2)$ and $K(\valthr_2)$ decrease in $\valthr_2$ under the IFR condition.
		
	Now we define $\val_2^W$, $\val_2^L$, and $\val_2^N$ as follows
	\begin{subequations}
	\begin{align}
	& J(\val_2^W)=0,	\\
	& K(\val_2^L)=0,	\\
	& \val_2^N = \valthr_1^\text{MP}(\mechanism_1),
	\end{align}
	\end{subequations}
	and derive the best response of MNO-1 based on $\val_2^W$, $\val_2^L$, and $\val_2^N$ as follows
	\begin{itemize}
		\item If $\valthr_2<\val_2^W$, then $g(\valthr_1,\valthr_2)>0$ for any $\valthr_1\in[\valthr_2,\valthr_2 + (1-\xi(\Mechanism))(\valmax-\valthr_2)]$, which means that MNO-1 should give up the market competition and obtain zero market share, i.e., $\valthr_1^*=\costqos_1$.
		
		\item If $\val_2^W\le\valthr_2<\val_2^L$, then there exists a unique $\hat{\valthr}_1\in[\valthr_2,\valthr_2 + (1-\xi(\Mechanism))(\valmax-\valthr_2)]$ such that $g(\hat{\valthr}_1,\valthr_2)=0$, which means that the both MNOs obtain positive market shares, i.e., $\valthr_1^*=\hat{\valthr}_1$.
		
		\item If $\val_2^L\le\valthr_2<\val_2^N$, then $g(\valthr_1,\valthr_2)<0$ for any $\valthr_1\in[\valthr_2,\valthr_2 + (1-\xi(\Mechanism))(\valmax-\valthr_2)]$, which means that MNO-1 can leave a zero market competition to MNO-2, i.e., $\valthr_1^*=\valthr_2$.
		
		\item If $\valthr_2\ge\val_2^N$, then MNO-1 can decide its threshold user type as in Lemma \ref{Lemma: Monopoly Optimal Pricing} without considering the existence of MNO-2, i.e., $\valthr_1^*=\valthr_2^\text{MP}(\mechanism_1)$.
	\end{itemize}
\end{proof}

\section{}\label{Appendix: partition at equilibrium}
\begin{proof}[\bf Proof of Theorem \ref{Theorem: Mechanism market partition}]
	We prove Theorem \ref{Theorem: Mechanism market partition} by deriving the conditions under which one of the MNO obtains a zero market share.
	
	First, we derive the condition (i.e., the upper bound of $c_1$) under which MNO-2 cannot obtain a positive market share no matter what data mechanism outcome $\Mechanism=\{\mechanism_1,\mechanism_2\}$.
	According to Theorem \ref{Theorem: Market Partition}, MNO-2 just obtains a zero market share if $\valthr_1=\valthr_2$.
	Therefore, we substitute  $\valthr_1=\valthr_2=\valthr$ into the threshold equilibrium condition (\ref{Equ: Coexistence}) and obtain
	\begin{equation}\label{Proof Equ: MNO-2 just zero}
	\left\{
	\begin{aligned}
	& H(\valthr)+\frac{\valthr-\costqos_1}{1-\xi(\Mechanism)}\cdot h(\valthr)=1, \\
	& H(\valthr)=H(\valthr)+(\valthr-\costqos_2) \cdot \frac{h(\valthr)}{1-\xi(\Mechanism)}.
	\end{aligned}
	\right.
	\end{equation}
	
	After solving (\ref{Proof Equ: MNO-2 just zero}), we obtain
	\begin{equation}\label{Proof Equ: c_1}
	c_1=\QoS_1\left[ \frac{c_2}{\QoS_2} - \left(1-\xi(\Mechanism)\right)\cdot\frac{1-H\left(\frac{c_2}{\QoS_2}\right)}{h\left(\frac{c_2}{\QoS_2}\right)} \right],
	\end{equation}
	which means that MNO-2 cannot obtain a positive market share under the data mechanism outcome $\Mechanism=\{\mechanism_1,\mechanism_2\}$ if 
	\begin{equation}\label{Proof Equ: c_1 inequality}
	c_1<\QoS_1\left[ \frac{c_2}{\QoS_2} - \left(1-\xi(\Mechanism)\right)\cdot\frac{1-H\left(\frac{c_2}{\QoS_2}\right)}{h\left(\frac{c_2}{\QoS_2}\right)} \right].
	\end{equation}
	
	Note that the right hand side of (\ref{Proof Equ: c_1 inequality}) increases in $\xi(\Mechanism)=\frac{\QoS_2\usage_{\mechanism_2}}{\QoS_1\usage_{\mechanism_1}}$.
	Therefore, we substitute the data mechanisms $\Mechanism=\{\roll,\trad\}$ into the right hand side of (\ref{Proof Equ: c_1 inequality}) to derive $C^\textit{Single}_1(\QoS_1,\QoS_2,c_2)$ as follows
	\begin{equation}
	C^\textit{Single}_1(\QoS_1,\QoS_2,c_2)\eq\textstyle
	\QoS_1\left[ \frac{c_2}{\QoS_2} - \left(1-\frac{\QoS_2\usage_{\trad}}{\QoS_1\usage_{\roll}}\right)\cdot\frac{1-H\left(\frac{c_2}{\QoS_2}\right)}{h\left(\frac{c_2}{\QoS_2}\right)} \right].
	\end{equation}
	
	In this case, the data mechanism equilibrium of Game \ref{Game: Mechanism} is $\Mechanism^*=\{\roll,\na\}$ if $c_1<C^\textit{Single}_1(\QoS_1,\QoS_2,c_2)$.

	Second, we derive the condition (i.e., the upper bound of $c_2$) under which MNO-1 cannot obtain a positive market share no matter what data mechanism outcome $\Mechanism=\{\mechanism_1,\mechanism_2\}$.
	Since MNO-1 has the advantage on the QoS, when MNO-1 just obtains zero market share, it is possible for MNO-1 to attract the high valuation users under $(\roll,\roll)$ or the low valuation users under $(\trad,\roll)$.
	Therefore, we have the following two critical conditions that lead to a  zero market share for MNO-1:
	\begin{itemize}
		\item $\valthr_1=\xi(\Mechanism)\valthr_2+(1-\xi(\Mechanism))\valmax$ under $\Mechanism=\{\roll,\roll\}$.
		\item $\valthr_1=\valthr_2$ under $\Mechanism=\{\trad,\roll\}$.
	\end{itemize}
	
	We substitute  $\valthr_1=\xi(\Mechanism)\valthr_2+(1-\xi(\Mechanism))\valmax$ and $\Mechanism=\{\roll,\roll\}$ into the threshold equilibrium condition (\ref{Equ: Coexistence}) and obtain
	\begin{equation}\label{Proof Equ: c_2 inequality 1}
	c_2 = \textstyle c_1-(\QoS_1-\QoS_2)\valmax - \QoS_2\cdot\frac{ 1-H\left( \frac{c_1-(\QoS_1-\QoS_2)\valmax}{\QoS_2} \right)  }{ h\left( \frac{c_1-(\QoS_1-\QoS_2)\valmax}{\QoS_2} \right) }.
	\end{equation}
	Similarly, we substitute $\valthr_1=\valthr_2$ and $\Mechanism=\{\trad,\roll\}$ into the threshold equilibrium condition (\ref{Equ: Coexistence}) and obtain
	\begin{equation}\label{Proof Equ: c_2 inequality 2}
	c_2 = \textstyle \QoS_2\left[ \frac{c_1}{\QoS_1} - \left(1-\frac{\QoS_1\usage_{\trad}}{\QoS_2\usage_{\roll}}\right)\cdot\frac{1-H\left(\frac{c_1}{\QoS_1}\right)}{h\left(\frac{c_1}{\QoS_1}\right)} \right].
	\end{equation}
	
	Combining (\ref{Proof Equ: c_2 inequality 1}) and (\ref{Proof Equ: c_2 inequality 2}), we can get the upper bound for MNO-2's cost  $C^\textit{Single}_2(\QoS_1,\QoS_2,c_1)$ as follows:
	\begin{equation}
	\begin{aligned}
	C^\textit{Single}_2(\QoS_1,\QoS_2,c_1)=
	\max\Bigg\{ 
	\textstyle \QoS_2\left[ \frac{c_1}{\QoS_1} - \left(1-\frac{\QoS_1\usage_{\trad}}{\QoS_2\usage_{\roll}}\right)\cdot\frac{1-H\left(\frac{c_1}{\QoS_1}\right)}{h\left(\frac{c_1}{\QoS_1}\right)} \right], \\
	\textstyle c_1-(\QoS_1-\QoS_2)\valmax - \QoS_2\cdot\frac{ 1-H\left( \frac{c_1-(\QoS_1-\QoS_2)\valmax}{\QoS_2} \right)  }{ h\left( \frac{c_1-(\QoS_1-\QoS_2)\valmax}{\QoS_2} \right) }
	  \Bigg\}
	\end{aligned}
	\end{equation}
	
	In this case, the data mechanism equilibrium of Game \ref{Game: Mechanism} is $\Mechanism^*=\{\na,\roll\}$ if $c_2<C^\textit{Single}_2(\QoS_1,\QoS_2,c_1)$.
\end{proof}

\begin{proof}[\bf Proof of Lemma \ref{Lemma: Mechanism equilibrium MNO-1} and Lemma \ref{Lemma: Mechanism equilibrium MNO-2}]
	In the proof of Theorem \ref{Theorem: Mechanism market partition}, we have explained the data mechanism equilibrium (\ref{Equ: Mechanim Equilibrium R Na}) of Lemma \ref{Lemma: Mechanism equilibrium MNO-1} and (\ref{Equ: Mechanim Equilibrium Na R}) of Lemma \ref{Lemma: Mechanism equilibrium MNO-2}.
	In the following we will prove Lemma \ref{Lemma: Mechanism equilibrium MNO-1} and Lemma \ref{Lemma: Mechanism equilibrium MNO-2} by showing the equality (\ref{Equ: Profit Effect R Na}) and inequality (\ref{Equ: Profit Effect Na R}), respectively.
	
	We first prove Lemma \ref{Lemma: Mechanism equilibrium MNO-1} by showing that MNO-2 cannot reduce MNO-1's profit no matter what data mechanism it adopts, i.e., $\profit_1(\roll,\trad)=\profit_1(\roll,\roll)$ if $c_1<C^\textit{Single}_1(\QoS_1,\QoS_2,c_2)$.
	
	From $c_1<C^\textit{Single}_1(\QoS_1,\QoS_2,c_2)$, we obtain 
	\begin{equation}\label{Proof Equ: condition Na}
	\textstyle
	\costqos_1<  \costqos_2 - \left(1-\frac{\QoS_2\usage_{\trad}}{\QoS_1\usage_{\roll}}\right)\cdot\frac{1-H\left(\costqos_2\right)}{h\left(\costqos_2\right)} .
	\end{equation}	
	Moreover, Lemma \ref{Lemma: Monopoly Optimal Pricing} indicates that MNO-1's optimal threshold user type $\valthr_1^\text{MP}(\roll)$ under the rollover mechanism $\mechanism_1=\roll$ satisfies 
	\begin{equation} \label{Proof Equ: val 1 MP Na}
	\textstyle
	\valthr_1^\text{MP}(\roll)-\frac{1-H(\valthr_1^\text{MP}(\roll))}{h(\valthr_1^\text{MP}(\roll))}=\costqos_1.
	\end{equation}	
	Combining (\ref{Proof Equ: condition Na}) and (\ref{Proof Equ: val 1 MP Na}) together, we obtain
	\begin{equation}
	\begin{aligned}
	\textstyle
	\valthr_1^\text{MP}(\roll)-\frac{1-H(\valthr_1^\text{MP}(\roll))}{h(\valthr_1^\text{MP}(\roll))} 
	&\textstyle 
	< \costqos_2 - \left(1-\frac{\QoS_2\usage_{\trad}}{\QoS_1\usage_{\roll}}\right)\cdot\frac{1-H\left(\costqos_2\right)}{h\left(\costqos_2\right)} \\
	&\textstyle
	< \costqos_2 - \frac{1-H\left(\costqos_2\right)}{h\left(\costqos_2\right)},
	\end{aligned}
	\end{equation}
	which implies $\costqos_2>\valthr_1^\text{MP}(\roll)$, since the data valuation $\val$ satisfies the IFR condition.
	
	Now we have shown that $\costqos_2>\valthr_1^\text{MP}(\roll)$ if $c_1<C^\textit{Single}_1(\QoS_1,\QoS_2,c_2)$.
	According to Theorem \ref{Theorem: Equilibrium Regime}, it corresponds to the MNO-1's strong monopoly regime, i.e., 	$(\costqos_1,\costqos_2)\in\Psi_1^\text{SM}$.
	That is, MNO-2 cannot affect MNO-1's profit, i.e., $\profit_1(\roll,\trad)=\profit_1(\roll,\roll)$.
	Hence we have proved Lemma \ref{Lemma: Mechanism equilibrium MNO-1}.
	
	As for the case of $c_2<C^\textit{Single}_2(\QoS_1,\QoS_2,c_1)$ in Lemma \ref{Lemma: Mechanism equilibrium MNO-2}, we find that the condition $c_2<C^\textit{Single}_2(\QoS_1,\QoS_2,c_1)$ cannot guarantee the MNO-2's strong monopoly regime (i.e., MNO-2's week monopoly regime is also possible), hence it is possible for MNO-1 to reduce MNO-2's profit by adopting the rollover mechanism, i.e., $\profit_2(\roll,\roll)\le\profit_2(\trad,\roll)$.
\end{proof}

\section{}\label{Appendix: numerical compute}
Next we introduce how to compute the threshold values in the equilibrium structure of Game \ref{Game: Mechanism}.
We start with introducing Lemma \ref{Proof Lemma: strategy}.

\begin{lemma}\label{Proof Lemma: strategy}
	In the coexistence case of Game \ref{Game: Mechanism}, i.e., $c_1\ge C^\text{Single}_1(\QoS_1,\QoS_2,c_2)$ and $c_2\ge C^\text{Single}_2(\QoS_1,\QoS_2,c_1)$, we have
		\begin{equation} 
		\begin{aligned}
		& \profit_1(\roll,\trad) > \profit_1(\trad,\trad).
		\end{aligned}
		\end{equation}	
\end{lemma}

\begin{proof}[\bf Proof of Lemma \ref{Proof Lemma: strategy}]

		Suppose the threshold equilibriums under the data mechanisms $\Mechanism=(\trad,\trad)$ and $\Mechanism'=(\roll,\trad)$ are $\Threshold=\{\valthr_1,\valthr_2\}$ and $\Threshold'=\{\valthr_1',\valthr_2'\}$, respectively.
		Under the data mechanisms $\Mechanism$ and $\Mechanism'$, MNO-1 always obtains the high valuation users.
		Therefore, the profits of MNO-$1$ in the two cases are
		\begin{equation}
		\begin{aligned}
		\profit_1(\Mechanism,\Threshold) &= \QoS_1\usage_{\trad} \left[ \valthr_1 -\costqos_1 \right]  \Big[ 1-H\big(\valeq(\Mechanism,\Threshold)  \big) \Big] ,
		\end{aligned}
		\end{equation}
		and 
		\begin{equation}
		\begin{aligned}
		\profit_1(\Mechanism',\Threshold') &= \QoS_1\usage_{\roll} \left[ \valthr_1' -\costqos_1 \right]  \Big[ 1-H\big(\valeq(\Mechanism',\Threshold')  \big) \Big] .
		\end{aligned}
		\end{equation}		
		
		In the following, we prove $\profit_1(\Mechanism,\Threshold)<\profit_1(\Mechanism',\Threshold')$.
		
		According to the definition of $\xi$ in (\ref{Equ: xi}), we have
		\begin{equation}\label{Proof Equ: xi xi'}
			\xi(\Mechanism)> \xi(\Mechanism').
		\end{equation}
		Based on the definition of the neutral user type in (\ref{Equ: theta_eq}), the inequality (\ref{Proof Equ: xi xi'}) implies that
		\begin{equation}\label{Proof Equ: mechanism mechanism'}
			\valeq(\Mechanism,\Threshold) > \valeq(\Mechanism',\Threshold).
		\end{equation}
		Therefore, we have
		\begin{equation}
		\begin{aligned}
		\profit_1(\Mechanism,\Threshold)
		&=\QoS_1\usage_{\trad} \left[ \valthr_1 -\costqos_1 \right]  \Big[ 1-H\big(\valeq(\Mechanism,\Threshold)  \big) \Big] \\
		&  <\QoS_1\usage_{\roll} \left[ \valthr_1 -\costqos_1 \right]  \Big[ 1-H\big(\valeq(\Mechanism',\Threshold)  \big) \Big]= \profit_1(\Mechanism',\Threshold),
		\end{aligned}		
		\end{equation}
		that it, 
		\begin{equation}\label{Proof Equ: W_1 1}
		\profit_1(\Mechanism,\Threshold) < \profit_1(\Mechanism',\Threshold).
		\end{equation}
		
		According to the best response analysis discussed in Lemma \ref{Lemma: theta_2^*} and Lemma \ref{Lemma: theta_1^*}, we know $\valthr_n<\valthr_n'$ for all $n\in\{1,2\}$.
		That is, both the MNOs can increase their threshold user types (by charging higher subscription fee and per-unit fee) if MNO-1 changes its data mechanism to rollover mechanism $\roll$.
		Accordingly, based on the definition of the neutral user type (\ref{Equ: theta_eq}), we have
		\begin{equation}\label{Proof Equ: W_1 2}
			\valeq(\Mechanism',\Threshold) > \valeq(\Mechanism',\valthr_1,\valthr_2'),
		\end{equation}
		which implies that 
		\begin{equation}
		\profit_1(\Mechanism',\Threshold) < \profit_1(\Mechanism',\valthr_1,\valthr_2').
		\end{equation}
		
		Recall that $\Threshold'=\{\valthr_1',\valthr_2'\}$ is the threshold equilibrium (i.e., the fixed point of the MNOs' best responses) under the data mechanism $\Mechanism'$.
		Hence we have 
		\begin{equation}\label{Proof Equ: W_1 3}
			\profit_1(\Mechanism',\valthr_1,\valthr_2') < \profit_1(\Mechanism',\Threshold').
		\end{equation}
		
		Combining (\ref{Proof Equ: W_1 1}), (\ref{Proof Equ: W_1 2}), and (\ref{Proof Equ: W_1 3}), we obtain
		\begin{equation}
			\profit_1(\Mechanism,\Threshold)<\profit_1(\Mechanism',\Threshold').
		\end{equation}	
\end{proof}
	
Lemma \ref{Proof Lemma: strategy} indicates that Game \ref{Game: Mechanism} never admits $(\trad,\trad)$ as an equilibrium.
The other three outcomes $\{(\trad,\roll),(\roll,\trad),(\roll,\roll)\}$ are possible at the equilibrium, which depends on the MNOs' QoS $\bm{\QoS}=\{\QoS_1,\QoS_2\}$ and costs $\bm{c}=\{c_1,c_2\}$.
Next we characterize the conditions under which one of the outcome in $\{(\trad,\roll),(\roll,\trad),(\roll,\roll)\}$ becomes the equilibrium of Game \ref{Game: Mechanism}.
To emphasize the dependence of the data mechanism equilibrium on the MNOs' QoS and costs, we write MNO-$n$'s profit as $\profit_n(\Mechanism,\bm{\QoS},\bm{c})$ for all $n\in\{1,2\}$.

We first characterize MNO-1's cost upper bound, denoted by ${C}^\textit{Roll}_1(\bm{\QoS},c_2)$, below which MNO-1 tends to adopt the rollover mechanism $\roll$.
Mathematically, ${C}^\textit{Roll}_1(\bm{\QoS},c_2)$ solves $\profit_1(\trad,\roll,\bm{\QoS},\bm{c})=\profit_1(\roll,\roll,\bm{\QoS},\bm{c})$ with respect to $c_1$.
Similarly, we can derive MNO-2's cost upper bound, denoted by ${C}^\textit{Roll}_2(\bm{\QoS},c_1)$, below which MNO-2 tends to adopt the rollover mechanism $\roll$.
Mathematically, $c_2={C}^\textit{Roll}_2$ solves $\profit_2(\roll,\trad,\bm{\QoS},\bm{c})=\profit_2(\roll,\roll,\bm{\QoS},\bm{c})$ with respect to $c_2$.
Therefore, the data mechanism equilibrium is $\Mechanism^*=\{\roll,\roll\}$ if
\begin{equation}\label{Proof Equ: condition (R,R)}
\left\{
\begin{aligned}
& c_1<{C}^\textit{Roll}_1(\bm{\QoS},c_2),\\
& c_2<{C}^\textit{Roll}_2(\bm{\QoS},c_1).
\end{aligned}
\right.
\end{equation}
Here the two inequalities in (\ref{Proof Equ: condition (R,R)}) hold simultaneously only if the two MNOs' has a relatively large QoS difference, i.e., $\QoS_2<\tilde{\QoS}$, where $\tilde{\QoS}$ solves $C^\textit{Roll}_1(\QoS_1,\tilde{\QoS},{C}^\textit{Roll}_2(\QoS_1,\tilde{\QoS},c_1))=c_1$.
Otherwise, if $\QoS_2\ge\tilde{\QoS}$, then there is no $\bm{\QoS}=\{\QoS_1,\QoS_2\}$ and $\bm{c}=\{c_1,c_2\}$ satisfying the two inequalities in (\ref{Proof Equ: condition (R,R)}).
In this case, Game \ref{Game: Mechanism} admits the symmetric equilibrium $\Mechanism^*=\{(\roll,\trad),(\trad,\roll)\}$ if 
\begin{equation}\label{Proof Equ: condition (R,T),(T,R)}
\left\{
\begin{aligned}
& c_1\ge{C}^\textit{Roll}_1(\bm{\QoS},c_2),\\
& c_2\ge{C}^\textit{Roll}_2(\bm{\QoS},c_1).
\end{aligned}
\right.
\end{equation}


\section{Data Mechanism Equilibrium}\label{Appendix: Data Mechanism Equilibrium}
Next we provide the numerical results for $\QoS_2\in\{0.95,0.99\}$.

\subsection{Small QoS Advantage}
MNO-1 has a small QoS advantage when $\QoS_2=0.95$ (i.e., $\QoS_1=1$ is a little larger than $\QoS_2=0.95$).

\begin{figure}  
	\centering
	\setlength{\abovecaptionskip}{0pt}
	\setlength{\belowcaptionskip}{0pt}
	\subfigure[Data mechanism equilibrium $\Mechanism^*$.]{\label{fig: EQ_mechanism_2}{\includegraphics[width=0.95\linewidth]{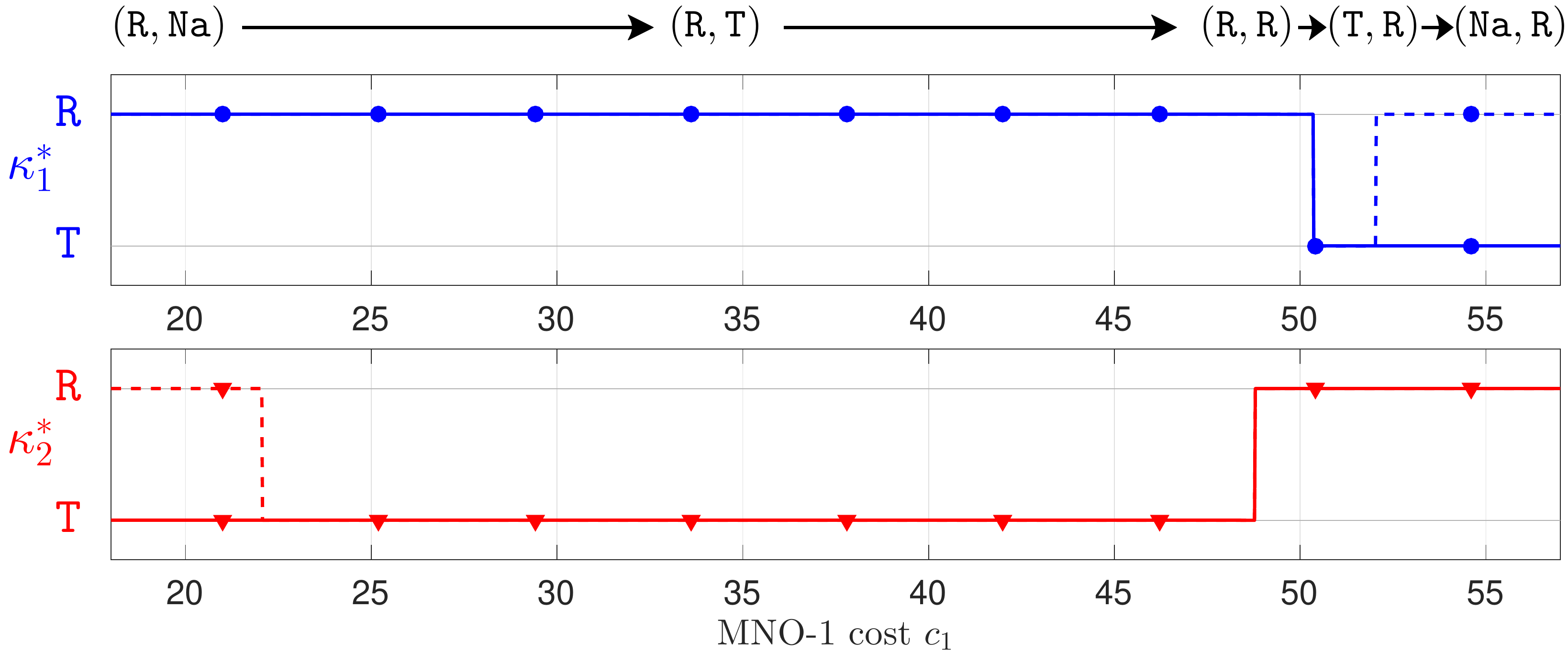}}}
	\subfigure[MNOs' profits under $\Mechanism^*$.]{\label{fig: EQ_mechanism_2_Profit}{\includegraphics[width=0.95\linewidth]{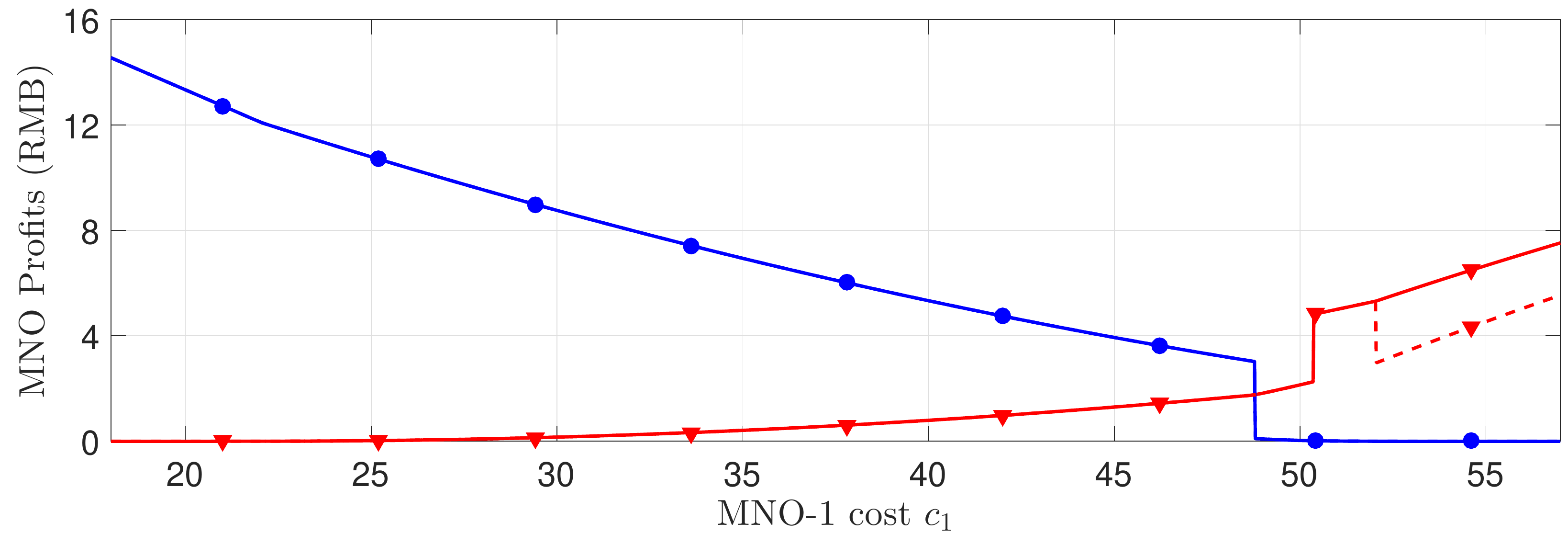}}} 
	\caption{Game \ref{Game: Mechanism} equilibrium when $\QoS_2=0.95$.}
	\label{fig: EQ rho=0.95}
\end{figure}

Similarly, we plot the data mechanism equilibrium $\Mechanism^*$ versus $c_1$ in Fig. \ref{fig: EQ_mechanism_2}, and label $\Mechanism^*$ on the top of the figure.
Fig. \ref{fig: EQ_mechanism_2_Profit} further simulates the corresponding profits under the equilibrium $\Mechanism^*$.

Different from Fig. \ref{fig: EQ_mechanism_3}, in Fig. \ref{fig: EQ_mechanism_2} we find that MNO-1 does \textbf{not always} adopt the rollover mechanism  $\roll$ all the time, which is the key difference from the case of $\QoS_2=0.91$ (shown in Fig. \ref{fig: EQ_mechanism_3}).
More specifically, as $c_1$ increases, the equilibrium varies  from $(\roll,\trad)$ to $(\trad,\roll)$ through $(\roll,\roll)$, which implies that:
as the high-QoS MNO-1's QoS advantage diminishes owning to its increasing cost,
(i) it would discard the rollover mechanism $\roll$, since its QoS advantage is not large enough;
(ii) the competitor (MNO-2) will upgrade from $\trad$ to $\roll$ to make more profit;
(iii) in this progress, it is possible for the two competitive MNOs to offer the rollover mechanism simultaneously, 

Similar to Fig. \ref{fig: EQ_mechanism_3_Profit}, in Fig. \ref{fig: EQ_mechanism_2_Profit} we find that MNO-1 experiences a profit drop when MNO-2 changes its data mechanism from $\trad$ to $\roll$, i.e., $c_1=48$ RMB/GB.
MNO-2 experiences a profit jump when MNO-1 changes its data mechanism from $\roll$ to $\trad$, i.e., $c_1=50.5$ RMB/GB.

Furthermore, Lemma \ref{Lemma: Mechanism equilibrium MNO-1} and Lemma \ref{Lemma: Mechanism equilibrium MNO-2} are also reflected in Fig. \ref{fig: EQ_mechanism_2_Profit} when $c_1<22$ RMB/GB or $c_1>52$ RMB/GB.
The key insights are similar to that we discussed for Fig. \ref{fig: EQ_mechanism_3_Profit}.

\begin{figure}  
	\centering
	\setlength{\abovecaptionskip}{0pt}
	\setlength{\belowcaptionskip}{0pt}
	\subfigure[MNO-1]{\label{fig: ProfitCompare_2_MNO1}{\includegraphics[width=0.49\linewidth]{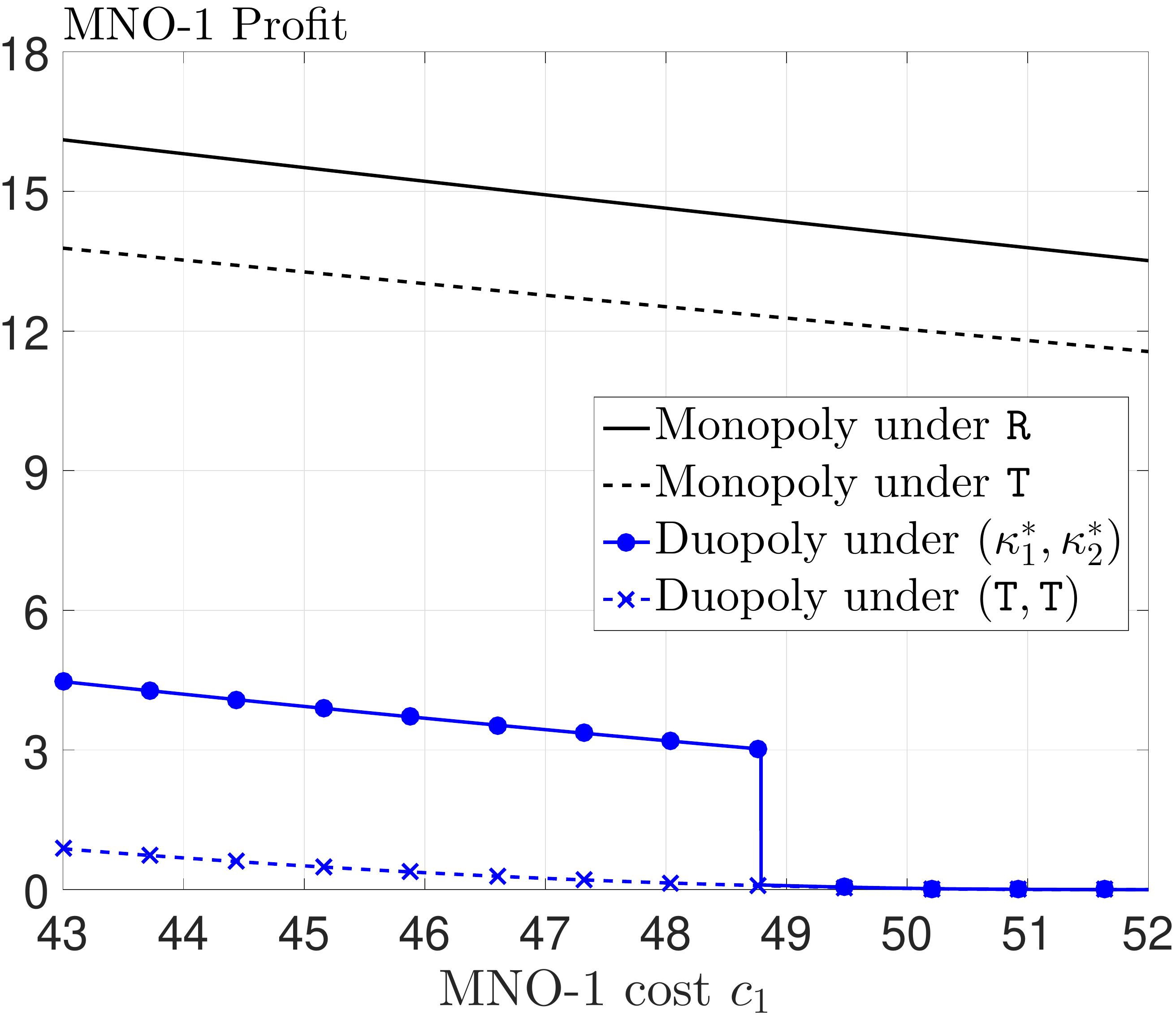}}}
	\subfigure[MNO-2]{\label{fig: ProfitCompare_2_MNO2}{\includegraphics[width=0.49\linewidth]{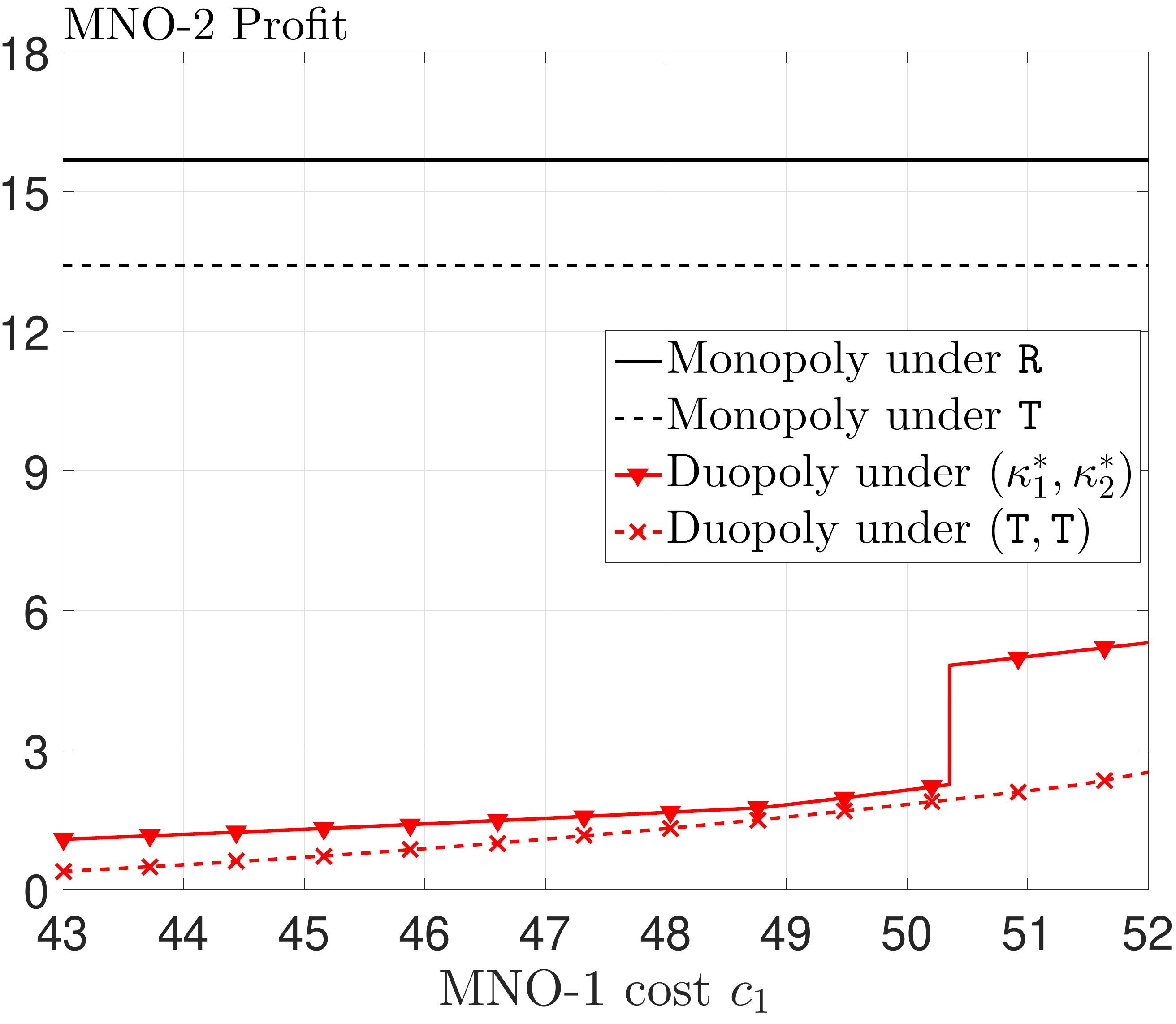}}} 
	\caption{Impact of rollover mechanism on duopoly market.}
	\label{fig: Impact of rollover mechanism 2}
\end{figure}

Fig. \ref{fig: ProfitCompare_2_MNO1} plots MNO-1's profit versus its cost $c_1$ in four scenarios.
The two black curves without markers correspond to MNO-1's monopoly market under data mechanism $\trad$ and $\roll$.
The blue circle curve corresponds to MNO-1's profit in the duopoly market under the equilibrium $\Mechanism^*$ shown in Fig. \ref{fig: EQ_mechanism_2}.
Essentially, the blue circle curve is the same as that in Fig. \ref{fig: EQ_mechanism_2_Profit}.
The blue cross curve corresponds to MNO-1's profit in the duopoly market under fixed data mechanisms $(\trad,\trad)$.
In this case, the MNOs only compete on price (but not on the data mechanism choice).
By comparing the black curves with the blue curves, we find that the market competition significantly reduces MNO-1's profit.
By comparing the two blue curves with markers, we note that the rollover mechanism significantly increases MNO-1's profit $342\%$ (on average) in the market competition.

Fig. \ref{fig: ProfitCompare_2_MNO2} plots MNO-2's profit versus the cost $c_1$ in similar scenarios.
Specifically, the two black curves without markers correspond to MNO-2's monopoly market.
The two red curves  correspond to the cases in the duopoly market. 
The red triangle curve is the same as that in Fig. \ref{fig: EQ_mechanism_2_Profit}.
The red cross curve corresponds to MNO-2's profit in the duopoly market under fixed data mechanisms $(\trad,\trad)$.
By comparing the black curves with the red curves, we find that the market competition significantly reduces MNO-2's profit.
The profit decrement of MNO-2 is larger than that of MNO-1, since MNO-1 has the QoS advantage and attracts high-valuation users.
By comparing the two red curves with markers, we note that the rollover mechanism increases MNO-1's profit $248\%$ (on average) in the market competition.

\subsection{Negligible QoS Advantage}
MNO-1 has a negligible QoS advantage when $\QoS_2=0.99$ (i.e., $\QoS_1=1$ is almost the same as $\QoS_2=0.99$). 

\begin{figure}  
	\centering
	\setlength{\abovecaptionskip}{0pt}
	\setlength{\belowcaptionskip}{0pt}
	\subfigure[Data mechanism equilibrium $\Mechanism^*$.]{\label{fig: EQ_mechanism_1}{\includegraphics[width=0.95\linewidth]{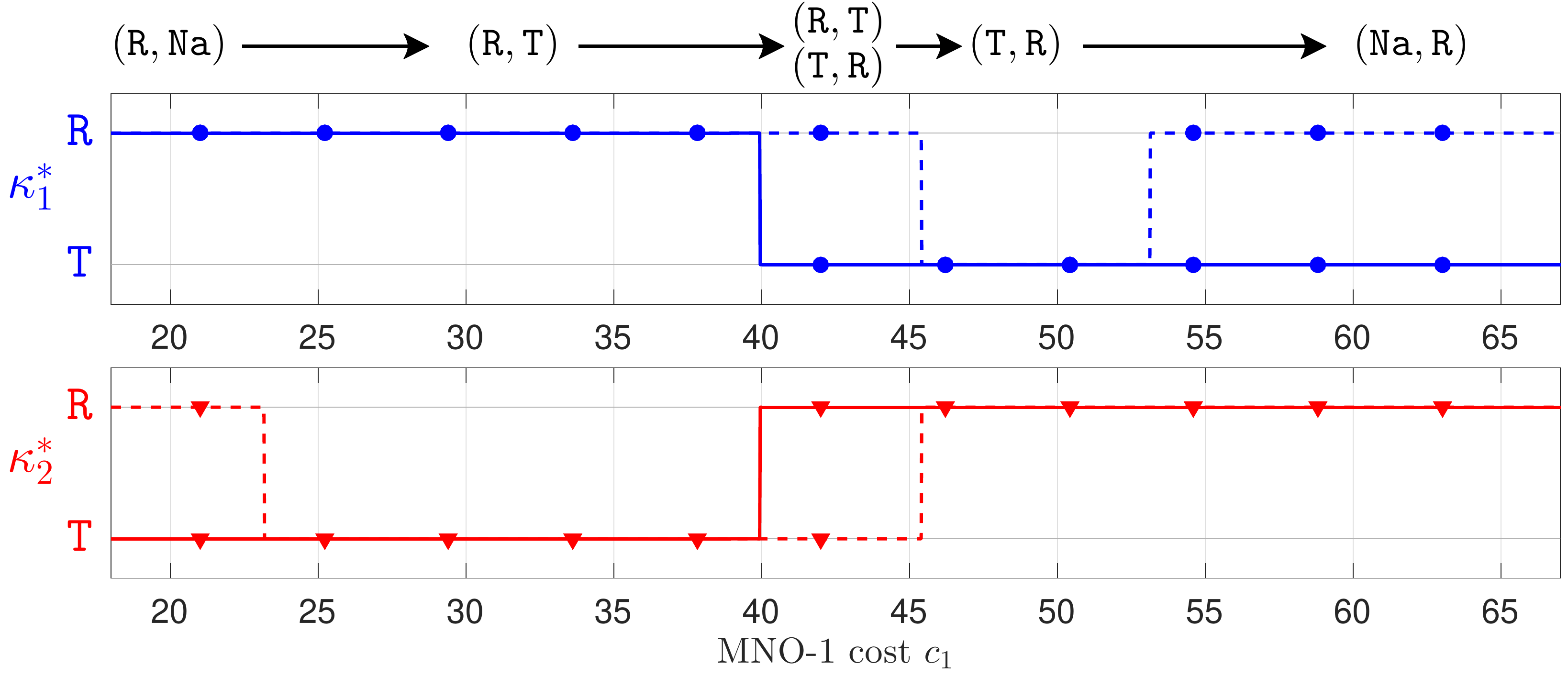}}}
	\subfigure[MNOs' profits under $\Mechanism^*$.]{\label{fig: EQ_mechanism_1_Profit}{\includegraphics[width=0.95\linewidth]{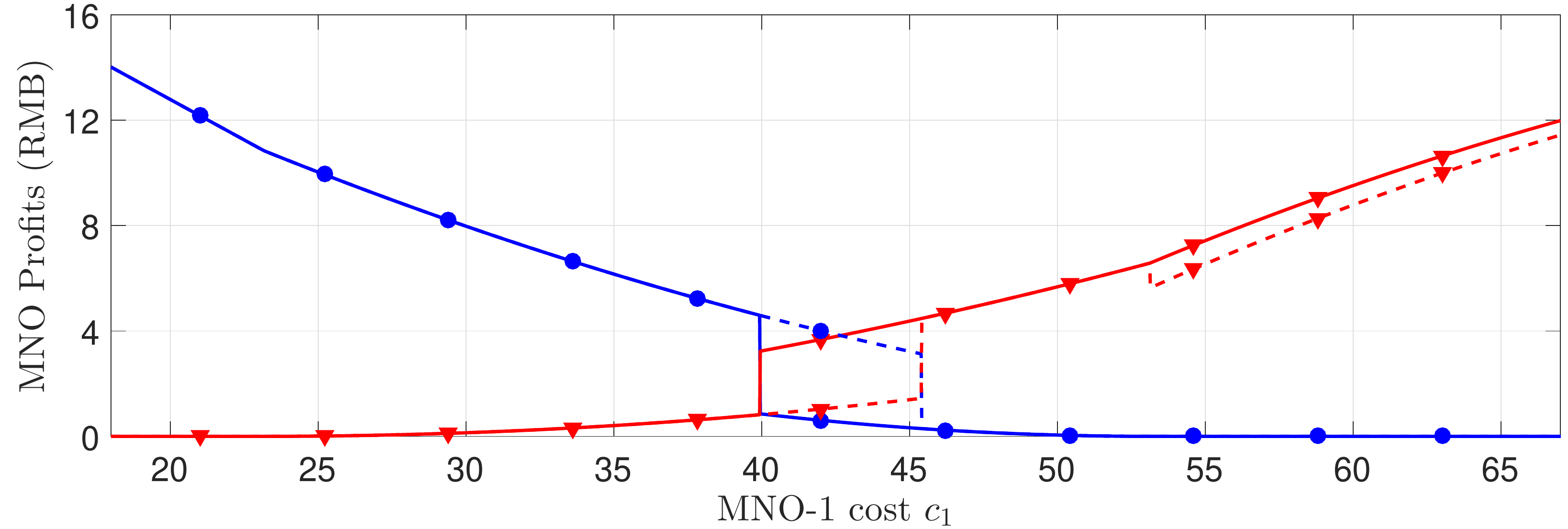}}} 
	\caption{Game \ref{Game: Mechanism} equilibrium when $\QoS_2=0.99$.}
	\label{fig: EQ rho=0.99}
\end{figure}
\begin{figure}  
	\centering
	\setlength{\abovecaptionskip}{0pt}
	\setlength{\belowcaptionskip}{0pt}
	\subfigure[MNO-1]{\label{fig: ProfitCompare_1_MNO1}{\includegraphics[width=0.49\linewidth]{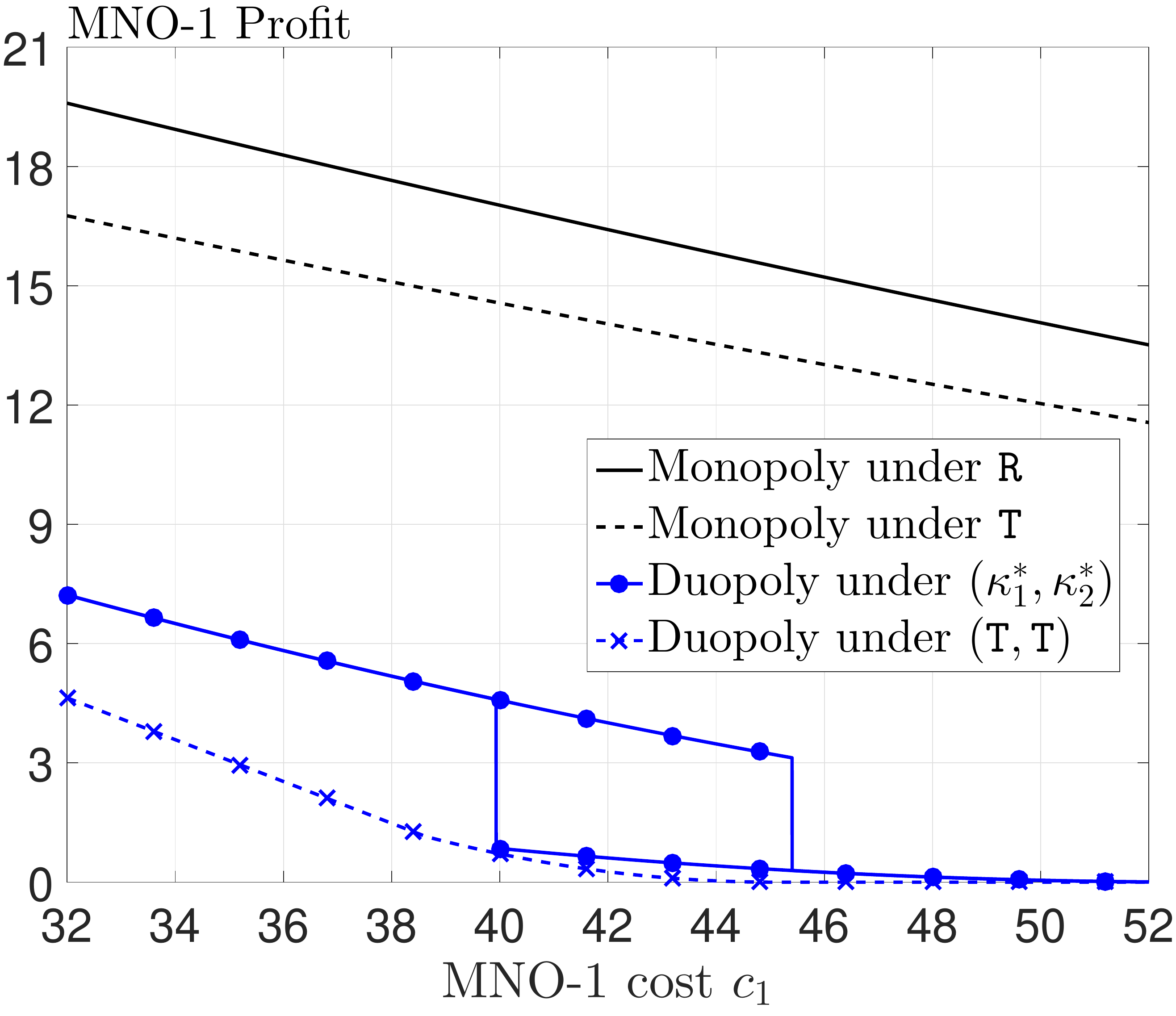}}}
	\subfigure[MNO-2]{\label{fig: ProfitCompare_1_MNO2}{\includegraphics[width=0.49\linewidth]{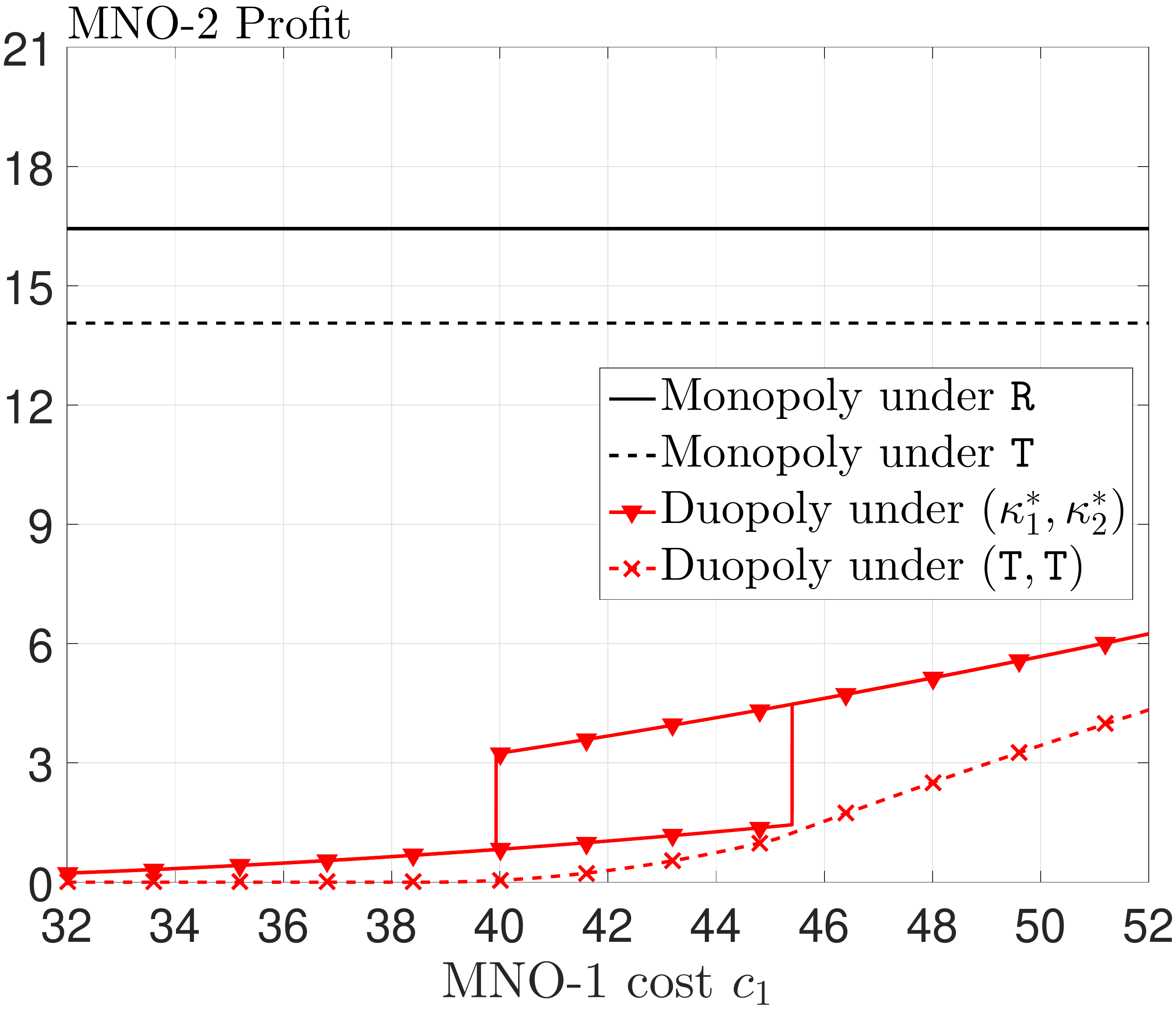}}} 
	\caption{Impact of rollover mechanism on duopoly market.}
	\label{fig: Impact of rollover mechanism 1}
\end{figure}

Similarly, we plot the data mechanism equilibrium $\Mechanism^*$ versus $c_1$ in Fig. \ref{fig: EQ_mechanism_1}, and label $\Mechanism^*$ on the top of the figure.
We find that as MNO-1's cost $c_1$ increases, the equilibrium changes from $(\roll,\trad)$ to $(\trad,\roll)$ through the symmetric equilibrium $\{(\roll,\trad),(\trad,\roll)\}$.
The symmetric equilibrium $\{(\roll,\trad),(\trad,\roll)\}$ is the key difference between $\QoS_2=0.99$ and $\QoS_2=0.95$.
Specifically, the symmetric equilibrium $\{(\roll,\trad),(\trad,\roll)\}$  results from  the two MNOs' homogeneity (due to the negligible QoS advantage and comparable costs).
That is, if the two MNOs are homogeneous, no matter who chooses the rollover mechanism $\roll$, the competitor has to choose the traditional mechanism $\trad$.

Fig. \ref{fig: EQ_mechanism_1_Profit} plots the corresponding profits under the equilibrium $\Mechanism^*$.
When the symmetric equilibrium $\{(\roll,\trad),(\trad,\roll)\}$ emerge, i.e., $40$ RMB/GB $<c_1<45.05$ RMB/GB, we find that no matter who adopts the rollover mechanism $\roll$, it obtains higher profit than that when it adopts the traditional mechanism $\trad$.

Fig. \ref{fig: ProfitCompare_1_MNO1} and Fig. \ref{fig: ProfitCompare_1_MNO2} plot MNO-1's and MNO-2' profits versus its cost $c_1$ in four scenarios.
The insights are similar to the previous results. 
The rollover mechanism increases MNO-1's profit $452\%$ (on average) and MNO-2's profit $391\%$ (on average) in the market competition.

\end{document}